\documentclass{article}

\usepackage{arxiv}
\usepackage{prodint}
\usepackage{amsmath}

\usepackage[utf8]{inputenc} % allow utf-8 input
\usepackage[T1]{fontenc}    % use 8-bit T1 fonts
\usepackage{hyperref}       % hyperlinks
\usepackage{url}            % simple URL typesetting
\usepackage{booktabs}       % professional-quality tables
\usepackage{amsfonts}       % blackboard math symbols
\usepackage{nicefrac}       % compact symbols for 1/2, etc.
\usepackage{microtype}      % microtypography
\usepackage{lipsum}
\usepackage{graphicx}
\graphicspath{ {./images/} }

\RequirePackage{amsthm,amsmath,amsfonts,amssymb}
\RequirePackage[numbers]{natbib}

\usepackage{graphicx}
\usepackage{prodint}
\usepackage{hyperref}
\hypersetup{colorlinks=true,linkcolor=blue,filecolor=magenta,urlcolor=cyan}

\newtheorem{theorem}{Theorem}[section]

\newtheorem{lemma}[theorem]{Lemma}
\newtheorem{proposition}[theorem]{Proposition}

\newcommand{\Mreps}{\mbox{\tt Mreps}}
\renewcommand{\Pr}{\mathbf{P}}
\newcommand{\trp}{^{\tt T}}
\newcommand{\IW}{\mathfrak{I}_{\mathfrak{W}}}
\newcommand{\IV}{\mathfrak{I}_{\mathfrak{V}}}
\newcommand{\IQ}{\mathfrak{I}_Q}
\newcommand{\DG}{\mbox{Dg}}
\newcommand{\plim}{\mbox{pr-lim}}

\newcommand{\BER}{\mbox{BER}}

\title{Joint Dynamic Models and Statistical Inference for Recurrent Competing Risks, Longitudinal Marker, and Health Status}

\author{
 Lili Tong \\
  Division of Biostatistics \\ Department of Population Health \\ 
  NYU Langone, New York, NY 10016 \\
  \texttt{tonglili91@hotmail.com} \\
   \And
 Piaomu Liu \\
  Department of Mathematical Sciences \\
  Bentley University \\
  Waltham, MA 02452 \\
  \texttt{pliu@bentley.edu} \\
  \And
 Edsel A. Pe\~{n}a \\
  Department of Statistics\\
  University of South Carolina\\
  Columbia, SC 29208 \\
  \texttt{pena@stat.sc.edu} \\
}

\begin{document}
\maketitle
\begin{abstract}
Consider a subject or unit in a longitudinal biomedical, public health, engineering, economic, or social science study which is being monitored over a possibly random duration. Over time this unit experiences competing recurrent events and a longitudinal marker transitions over a discrete state-space. In addition, its ``health or performance'' status also transitions over a discrete state-space with some states possibly absorbing states. A vector of covariates will also be associated with this unit. If there are absorbing states, of interest for this unit is its time-to-absorption of its health status process, which could be viewed as the unit's lifetime. Aside from being affected by its covariate vector, there could be associations among the recurrent competing risks processes, the longitudinal marker process, and the health status process in the sense that the time-evolution of each process is associated with the other processes. To obtain more realistic models and enhance inferential performance, a joint dynamic stochastic model for these components is proposed and statistical inference methods are developed. This joint model, formulated via counting processes and continuous-time Markov chains, has the potential of facilitating `personalized' interventions. This could enhance, for example, the implementation and adoption of precision medicine in medical settings. Semi-parametric and likelihood-based inferential methods for the model parameters are developed when a sample of these units is available. Finite-sample and asymptotic properties of estimators of model parameters, both finite- and infinite-dimensional, are obtained analytically or through simulation studies. 
%\RED{The developed procedures are illustrated using a real data set.}
\end{abstract}

% keywords can be removed

\noindent\textsc{MSC2020 Subject Classifications}: {\bf Primary:} 62N01, 60J27; {\bf Secondary:} 62G05, 62M02. 

\section{Introduction and Motivation}
\label{sec: intro}

Consider a unit --- for example, a human subject in a medical or a social science study, an experimental animal in a biological experiment, a machine in an engineering or reliability setting, or a company in an economic or business situation --- in a longitudinal study monitored over a period $[0,\tau$], where $\tau$ is possibly random. Associated with the unit is a covariate row vector $\mathbf{X} = (X_1, X_2, \ldots, X_p)$. Over time, the unit will experience occurrences of $Q$ competing types of recurrent events, its recurrent competing risks (RCR) component; transitions of a longitudinal marker (LM) process $W(t)$ over a discrete state space $\mathfrak{W}$; and transitions of a `health or performance' status (HS) process $V(t)$ over a discrete state space $\mathfrak{V} = \mathfrak{V}_0 \bigcup \mathfrak{V}_1$, with $\mathfrak{V}_0$ being the set of absorbing states, which is allowed to be an empty set, $\emptyset$; {while states in $\mathfrak{V}_1$ are transient states}. If the health status process transitions into an absorbing state prior to $\tau$, then monitoring of the unit ceases, so time-to-absorption serves as the lifetime of the unit. To demonstrate pictorially, the two panels, consisting of three plots each, in Figure \ref{fig: observables} depict the time-evolution for two distinct units, where $Q = 4$, $\mathfrak{W} = \{w_1, w_2, w_3,w_4\}$, $\mathfrak{V} =\{v_1, v_2, v_3,v_4\}$ where $v_4$ is an absorbing state, while $v_1$, $v_2$ and $v_3$ are transient states. In panel 1, the unit did not transition to an absorbing state prior to reaching $\tau$; whereas in panel 2, the unit reached an absorbing state (state 4) prior to $\tau$. Two major questions arise: (a) how do we specify a dynamic stochastic model that could be a generative model for such a data; and (b) how do we make statistical inferences for the model parameters and, {possibly a prediction,} of a unit's lifetime from a sample of such data?

\medskip
\centerline{[INSERT FIGURE \ref{fig: observables} HERE]}
\medskip

To address these questions, the major goals of this paper are (i) to propose a joint stochastic model for the random observables consisting of the RCR, LM, and HS components for such units; and (ii) to develop appropriate statistical inference methods for the proposed joint model when a sample of units are observed. Achieving these two goals {may} enable statistical prediction of the (remaining) lifetime of a, possibly new, unit; allow for the examination of the synergistic association among the RCR, LM, and HS components; and provide a vehicle to compare different groups of units and/or study the effects of concomitant variables or factors. More importantly, a joint stochastic model endowed with proper statistical inference methods could potentially enable unit-based interventions which are performed after a recurrent event occurrence or a transition in either the LM or HS processes. As such it could enhance the implementation of precision or personalized decision-making; for instance, precision medicine, {though this topic is deferred in a future paper}.

A specific situation where such a data accrual occurs is in a medical study. For example, a subject may experience different types of recurring cancer, with the longitudinal marker being the prostate-specific antigen (PSA) level categorized into a finite ordinal set, while the health status is categorized into either a healthy, diseased, or dead state, with the last state absorbing. Another situation pertains to recurrent competing falls of an elderly person, with the LM process pertaining to whether the person is in a hospital or in a nursing home, while the HS process pertains to the state of physical or cognitive fitness (see, for instance, \cite{AboutFalls2021}).
% and the first example in section \ref{sec-Other Examples} of the Online Supplementary Materials). \RED{A variety of other situations in a public health, engineering, social science, and economics settings, where such data structure arise, are further described in Section \ref{sec-Other Examples} of the Online Supplementary Materials.} 
 See also Chapter 1 of \cite{Nel2003} for many practical situations where recurrent events occur.

Several previous works dealing with modeling have either focused in the marginal modeling of each of the three data components, or in the joint modeling of two of the three data components. In this paper we tackle the problem of {\em simultaneously} modeling all three data components: RCR, LM, and HS, in order to account for the associations among these components, which would not be possible using either the marginal modeling approach or the joint modeling of pairwise combinations of these three components. A joint full model could also subsume these previous marginal or joint models -- in fact, our proposed class of models subsumes as special cases models that have been considered in the literature. In contrast, only by imposing restrictive assumptions, such as the independence of the three model components, could one obtain a joint full model from marginal or pairwise joint models. As such, a joint full model will be less likely to be mis-specified, thereby reducing the potential biases that could accrue from mis-specified models among estimators of model parameters or when predicting residual lifetime.

A joint modeling approach has been extensively employed in previous works. For instance, joint models for an LM process and a survival or time-to-event (TE) process have been proposed in \cite{tsiatis1995modeling}, \cite{wulfsohn1997joint}, \cite{song2002semiparametric}, \cite{henderson2000joint}, \cite{tsiatis2004joint}, and \cite{mclain2015}.  The joint modeling of an LM process and a recurrent event process has also been discussed in \cite{han2007} and \cite{efen2013}, while the joint modeling of a recurrent event and a lifetime has also been done such as in \cite{liu2004}. An important and critical theoretical aspect that could not be ignored in these settings is that when an event occurrence is terminal (e.g., death) or when there is a finite monitoring period, informative censoring naturally occurs in the RCR, LM, or HS components, since when a terminal event or when the end of monitoring is reached, then the next recurrent event occurrence, the next LM transition, or the next HS transition will not be observed. 

Another aspect that needs to be taken into account in a dynamic modeling approach is that of performed interventions, usually upon the occurrence of a recurrent event. For instance, in engineering or reliability systems, when a component in the system fails, this component will either be minimally or partially repaired, or altogether replaced with a new component (called a perfect repair); while, with human subjects, when a cancer relapses, a hospital infection transpires, a gout flares-up, or alcoholism recurs, some form of intervention will be performed or instituted. Such interventions will impact the time to the next occurrence of the event, hence it is critical that such intervention effects be reflected in the model; see, for instance, \cite{gonzalez2005modelling} and \cite{han2007}. In addition, models should take into consideration the impacts of the covariates and the effects of accumulating event occurrences on the unit. Models that take into account these considerations have been studied in \cite{pena2006dynamic} and \cite{pena2007semiparametric}. Appropriate statistical inference procedures for these dynamic models of recurrent events and competing risks have been developed in \cite{pena2006dynamic} and \cite{taylor2014nonparametric}. Extensions of these joint dynamic models for both RCR and TE can be found in \cite{liu2015dynamic}. Some other recent works in joint modeling included the modeling of the three processes: LM, RCR (mostly, a single recurrent event), and TE simultaneously in \cite{kim2012joint}, \cite{cai2017joint}, \cite{krol2017tutorial}, \cite{mauguen2013dynamic} and \cite{blanche2015quantifying}. The joint model that will be proposed in this paper will take into consideration these important aspects.

We outline the contents of this paper. Prior to describing formally the joint model in Section \ref{sec: mathform}, we first present in Section \ref{sec: scenarios}  a concrete situation in a medical setting where the data accrual described above could arise and where the joint model will be relevant. 
%Other examples in the engineering, social, and economic disciplines are described in the online supplement. 
Section \ref{sec: mathform} formally describes the joint model using counting processes and continuous-time Markov chains (CTMCs), and provide interpretations of model parameters. In subsection \ref{subsec: special case} we discuss in some detail a special case of this joint model obtained using independent Poisson processes and homogeneous CTMCs. Section \ref{sec: estimation - parametric} deals with the estimation of the parameters in the afore-mentioned special case to pinpoint some intricacies of joint modeling and its inferential aspects. The general joint model contains nonparametric (infinite-dimensional) parameters, so in section \ref{sec: estimation - semiparametric} we will describe a semi-parametric estimation procedure for this general model. Section \ref{sec: Properties} will present asymptotic properties of the estimators. Section \ref{sec-Illustration} will then demonstrate the semi-parametric estimation approach through the use of a synthetic or simulated data using {\tt R} \cite{RCitation} and {\tt Fortran} programs we developed. To make implementation more computationally efficient, the {\tt R} programs call the {\tt Fortran} subroutines in portions requiring intensive calculations. In Section \ref{sec: simulation}, we perform simulation studies to investigate the finite-sample properties of the estimators arising from the semi-parametric estimation procedure. 
%and compare these finite-sample results to the theoretical asymptotic results. 
%\RED{An illustration of the semi-parametric inference procedure using a real medical data set is presented in Section \ref{sec: realdata}.} 
Section \ref{sec: conclusion} contains concluding remarks and describes some open research problems.   

\section{A Concrete Situation in the Medical Arena}
\label{sec: scenarios}
 
To demonstrate potential applicability of the proposed joint model, we describe a concrete situation in the medical field where the data accrual described in Section \ref{sec: intro} could arise. 
%{Other examples in the reliability, engineering, socio-economic, and financial settings are presented in section \ref{sec-Other Examples} of the online supplement.}
%\RED{[Refer to the paper by \cite{SinGaf2020}]} 
Gout is a form of arthritis characterized by sudden and severe attacks of pain, swelling, redness and tenderness in one or more joints in the toes, ankles, knees, elbows, wrists, and fingers. See, for instance, the website in Mayo Clinic about gout
{\small
\begin{center}
\url{https://www.mayoclinic.org/diseases-conditions/gout/symptoms-causes/syc-20372897}
\end{center}
}
According to \cite{SinGaf2020}, gout is the most prevalent inflammatory arthritis in the world. When a gout flare occurs, it renders the person incapacitated -- personally attested by the senior author (EP) -- and the debilitating condition may last for several days, possibly weeks. Since the location of the gout flare could vary, we may consider gout as competing recurrent events --- competing with respect to the location of the flare, and recurrent since it could keep coming back. Gout occurs when urate crystals accumulate in the joints, which in turn is associated with high levels of uric acid in the blood. 
The level of uric acid is measured by the {Serum Urate Level (SUR)}, which can be categorized as {Hyperuricemia} (if SUR $>$ 7.2 mg/dL for males; if SUR $>$ 6.0 mg/dL for females), or {Normal} (if 3.5 mg/dL $\le$ SUR $\le$ 7.2 mg/dL for males; if 2.6 mg/dL $\le$ SUR $\le$ 6.0 mg/dL for females). The SUR level could be considered a longitudinal marker. 

Gout is linked to other comorbidities (see \cite{SinGaf2020}), and one of these pertains to chronic kidney disease (CKD) since kidneys are associated with the excretion of uric acid in the body.  The ordinal stages of CKD, based on the value of the {Glomerular Filtration Rate (GFR)}, are as follows: {Stage 1 (Normal)} if  GFR $\ge$ 90 mL/min; {Stage 2 (Mild CKD)} if 60 mL/min $\le$ GFR $\le$ 89 mL/min; {Stage 3A (Moderate CKD)} if 45 mL/min $\le$ GFR $\le$ 59 mL/min; {Stage 3B (Moderate CKD)} if 30 mL/min $\le$ GFR $\le$ 44; {Stage 4 (Severe CKD)} if 15 mL/min $\le$ GFR $\le$ 29 mL/min; and {Stage 5 (End Stage CKD)} if GRF $\le$ 14 mL/min. The state of Stage 5 (End Stage CKD) can be viewed as an absorbing state. The CKD status could be viewed as the ``health status'' of the person. 
Other covariates, both environmental and genetic, such as gender, blood pressure, weight, diet, alcohol consumption, genetic markers, etc., could also impact the occurrence of gout flares, uric acid level, and CKD. 

When a gout flare occurs, lifestyle interventions could be performed such as (i) consuming skim milk powder enriched with the two dairy products glycomacropeptide (GMP) and G600 milk fat extract; (ii) consuming standard milk or lactose powder; or (iii) {taking prescribed medications such as colchicine, alluporinal, etc.} The purpose of such interventions is to lessen gout flare recurrences through control of the SUR level. 

In such a study, of major interest is to {\em jointly} model the competing gout recurrences, the categorized SUR process, and the CKD process. The dynamic joint model proposed in this paper could provide a potentially viable model for this situation.

\section{Joint Model of RCR, LM, and HS Processes}
\label{sec: mathform}

\subsection{Data Observables for One Unit}

Denote by $(\Omega, \mathfrak{F}, \Pr)$ the basic filtered probability space with filtration $\mathcal{F} = \{\mathcal{F}_s: s \ge 0\}$ where all random entities under consideration are defined. We begin by describing the joint model for the data observable components for {\bf one unit}. 

Let $\tau$, the end of monitoring period, have a distribution function $G(\cdot)$, which may be degenerate. The covariate vector will be ${X} = (X_1, \ldots, X_p)$, assumed to be time-independent, though the extension to time-dependent covariates is possible with additional assumptions.
For the RCR component, let $N^R = \{N_q^R(s) \equiv (N^R(s;q), q \in \IQ):\ s \ge 0\}$, with index set 
$$\IQ = \{1,\ldots,Q\},$$
be a $Q$-dimensional multivariate counting (column) vector process such that, for $q \in \IQ$, $N_q^R(s)$ is the number of {\em observed} occurrences of the recurrent event of type $q$ over $[0,s]$, with $N_q^R(0) = 0$. Thus, the sample path $s \mapsto N_q^R(s)$ takes values in $\mathbb{Z}_{0,+} = \{0,1,2,\ldots\}$, is a non-decreasing step-function, and is right-continuous with left-hand limits. We denote by $dN_q^R(s) = N_q^R(s) - N_q^R(s-)$, the jump at time $s$ of $N_q^R$. 

For the LM process, let $W = \{W(s): s \geq 0\}$, where $W(s)$ takes values in a finite state space $\mathfrak{W}$ with cardinality $|\mathfrak{W}|$. $W(s)$ represents the state of the longitudinal marker at time $s$. The sample path $s \mapsto W(s)$ is a step-function which is right-continuous with left-hand limits.  By introducing the lexicographically-ordered index set 
$$\mathfrak{I}_{\mathfrak{W}} = \{\{(w_1,w_2): w_2 \in \mathfrak{W}, w_2 \ne w_1\}, w_1 \in \mathfrak{W}\}\footnote{Thus, if $\mathfrak{W} = \{1,2,3\}$, then $\mathfrak{I}_{\mathfrak{W}} = \{(1,2),(1,3),(2,1),(2,2),(3,1),(3,2)\}$.},$$
so that $|\mathfrak{I}_{\mathfrak{W}}| = 2{\binom{|\mathfrak{W}|} {2}}$.
We can convert $W$ into a (column) $|\mathfrak{I}_{\mathfrak{W}}|$-dimensional multivariate counting process $\{N^W \equiv (N^W(s;w_1,w_2), (w_1,w_2) \in \IW): s \ge 0\}$, where $N^W(s;w_1,w_2)$ is the number of observed transitions of $W$ from state $w_1$ into $w_2$ over the period $[0,s]$, that is, for $(w_1,w_2) \in \IW$,
$$N^W(s;w_1,w_2) = \sum_{t \le s} I\{W(t-) = w_1,W(t) = w_2\},$$
with $I(\cdot)$ denoting indicator function.
Thus, for each $(w_1,w_2) \in \IW$, the sample path $s \mapsto N^W(s;w_1,w_2)$ takes values in $\mathbb{Z}_{0,+}$, is a non-decreasing step-function, right-continuous with left-hand limits, and with $N^W(0;w_1,w_2) = 0$. Furthermore, $\sum_{(w_1,w_2) \in \IW} dN^W(s;w_1,w_2) \in \{0,1\}$ for every $s \ge 0$.

For the HS process, let $V = \{V(s): s \ge 0\}$, where $V(s)$ takes values in the finite state space $\mathfrak{V} = \mathfrak{V}_0 \bigcup \mathfrak{V}_1$, where states in $\mathfrak{V}_0$ are absorbing states, and with $|\mathfrak{V}_0| \ge 0$; {while states in $\mathfrak{V}_1$ are transient states}. $V(s)$ is the state occupied by the HS process at time $s$, so that if $V(s) \in \mathfrak{V}_0$, then $V(s^\prime) = V(s)$ for all $s^\prime > s$. Similar to the LM process, define the lexicographically-ordered index set
$$\IV = \{\{(v_1,v): v \in \mathfrak{V}, v \ne v_1\}, v_1 \in \mathfrak{V}_1\},$$ 
whose cardinality is $|\IV| = |\mathfrak{V}_1| ( |\mathfrak{V}| - |\mathfrak{V}_1| )$. We convert $V$ into a (column) $|\IV|$-dimensional multivariate counting process $\{N^V \equiv (N^V(s;v_1,v), (v_1,v) \in \IV): s \ge 0\}$, where $N^V(s;v_1,v)$ is the number of observed transitions of $V$ from state $v_1$ into state $v$ over the period $[0,s]$, that is, for $(v_1,v) \in \IV$,
$$N^V(s;v_1,v) =\sum_{t \le s} I\{V(t-) = v_1, V(t) = v\}.$$
For each $(v_1,v) \in \IV$, the sample path $s \mapsto N^V(s;v_1,v)$ takes values in $\mathbb{Z}_{0,+}$, and is a non-decreasing step-function, right-continuous with left-hand limits, and with $N^V(0;v_1,v) = 0$. In addition, $\sum_{(v_1,v) \in \IV} dN^V(s;v_1,v) \in \{0, 1\}$ for every $s \ge 0$.
Next, we combine the multivariate counting processes $N^R$, $N^W$, and $N^V$ into one (column) multivariate counting process $N = \{N(s): s \ge 0\}$ of dimension $Q + |\IW| + |\IV|$, where, with $\trp$ denoting vector/matrix transpose,
$$N(s) = \left[(N^R(s))\trp, (N^W(s))\trp, (N^V(s))\trp\right]\trp.$$
{Since the components of $N$ are non-decreasing, clearly $N$ is an $\mathcal{F}$-submartingale.}

An important point needs to be stated regarding the observables in the study, which will have an impact in the interpretation of the parameters of the joint model. This pertains to the ``competing risks'' nature of all the possible events at each time point $s$. The possible $Q$ recurrent event types, as well as the potential transitions in the LM and HS processes, are all competing with each other. Thus, suppose that at time $s_0$, the event that occurred is a recurrent event of type $q_0$, that is, $dN^R(s_0;q_0) = 1$. This means that this event has occurred {\em in the presence} of the potential recurrent events from the other $Q-1$ risks, and the potential transitions from either the LM and HS processes. {This will entail the use of {\em crude} hazards, instead of {\em net} hazards, in the joint modeling, and this observation will play a critical role in the dynamic joint model since each of the competing event occurrences at a given time point $s$ from all the possible event sources (RCR, LM, and HS) will be affected by the history of all these processes just before time $s$. This is the aspect that exemplifies the synergistic association among the three components.} See \cite{Tsiatis1975,Tsiatis2005} for some aspects regarding competing risks modeling.

Another observable process for our joint model is a vector of effective (or virtual) age processes $\mathcal{E} = \{(\mathcal{E}_1(s),\ldots,\mathcal{E}_Q(s)): s \ge 0\}$, whose components are $\mathcal{F}$-predictable processes with sample paths that are non-negative, left-continuous, piece-wise non-decreasing, and differentiable. These effective age processes will carry the impact of interventions performed after each recurrent event occurrence or a transition in either the LM process or the HS process. For articles dealing with virtual ages, see the philosophically-oriented article \cite{FinCha21} and the  review article \cite{Beu2021}. {We point out that utilizing this notion of effective or virtual ages makes our model quite general since many current models are just special cases of our model with the proper choice of effective age processes.}

Finally, we define the time-to-absorption of the unit to be
$\tau_A = \inf\{s \ge 0: V(s) \in \mathfrak{V}_0\}$
with the convention that $\inf \emptyset = \infty$. Using this $\tau_A$ and $\tau$, we define the unit's at-risk process to be $Y = \{Y(s): s \ge 0\}$, with
$Y(s) = I\{\min(\tau,\tau_A) \ge s\}.$
In addition, we define LM-specific and HS-specific at-risk processes as follows: For $w \in \mathfrak{W}$, define $Y^W(s;w) = I\{W(s-) = w\}$; and, for $v \in \mathfrak{V}_1$, define $Y^V(s;v) = I\{V(s-) = v\}$.
For a unit, we could then concisely summarize the random observables in terms of stochastic processes as:
\begin{equation}
\label{data: one unit}
D = (X,N,\mathcal{E},Y,Y^W,Y^V) \equiv \{X,(N(s),\mathcal{E}(s),Y(s),Y^W(s),Y^V(s): s \ge 0)\}.
\end{equation}
Note that the processes are undefined for $s > \min(\tau_A,\tau) \equiv \inf\{s \ge 0: Y(s+) = 0\}$ since monitoring of the unit had by then ceased.

\subsection{Joint Model Specification for One Unit}
\label{subsec: joint model}

The joint model specification will be through the specification of the compensator process vector and the predictable quadratic variation (PQV) process matrix of the multivariate counting process $N$, which is possible because $N$ is a submartingale. The predictable compensator process vector $A = \{A(s): s  \ge 0\}$ is of dimension $Q + |\IW| + |\IV|$ and is such that the vector process $M = N - A = \{M(s) = N(s) - A(s): s \ge 0\}$ is a zero-mean square-integrable martingale process with PQV matrix process $\langle M, M \rangle$. The vectors $A$ and $M$ are actually partitioned into three vector components reflecting the RCR, LM, and HS components, according to
\begin{displaymath}
A = \left[ (A^R)\trp, (A^W)\trp, (A^V)\trp \right]\trp \quad \mbox{and} \quad
M = \left[(M^R)\trp, (M^W)\trp, (M^V)\trp\right],
\end{displaymath}
where, with $q \in \IQ$, $(w_1,w_2) \in \IW$, $(v_1,v) \in \IV$, and $s \ge 0$,
\begin{eqnarray*}
& A^R = \{(A^R(s;q))\} \quad \mbox{and} \quad M^R = \{(M^R(s;q))\}; & \\
& A^W = \{(A^W(s;(w_1,w_2)))\} \quad \mbox{and} \quad M^W = \{(M^W(s;(w_1,w_2)))\}; & \\
& A^V = \{(A^V(s;(v_1,v)))\} \quad \mbox{and} \quad M^V = \{(M^V(s;(v_1,v)))\}, &
\end{eqnarray*}
with $A^R$ and $M^R$ of dimensions $Q$; $A^W$ and $M^W$ of dimensions $|\IW|$; and $A^V$ and $M^V$ of dimensions $|\IV|$. The matrix $\langle M, M \rangle$ could then be partitioned similarly to reflect these block components.

We can now proceed with the specification of the compensator process vector and the PQV process matrix. 
For this purpose, we introduce the following quantities, functions, and processes.
\begin{itemize}
    \item For each $q \in \IQ$ there is a baseline hazard rate function $\lambda_{0q}(\cdot)$ with associated baseline cumulative hazard function $\Lambda_{0q}(\cdot) = \int_0^\cdot \lambda_{0q}(v) dv$. We also denote by $\bar{F}_{0q}(\cdot) = \prodi_{v=0}^{\cdot} [1 - \Lambda_{0q}(dv)]$ the associated baseline survivor function, where $\prodi$ is the product-integral symbol.
    \item For each $q \in \IQ$ there is a mapping $\rho_q(\cdot|\cdot): \mathbb{Z}_{0,+}^Q \times \Re^{d_q} \rightarrow \Re_{0,+}$, where the $d_q$'s are known positive integers. There is an associated vector $\alpha_q \in \Re^{d_q}$.
    \item There is a collection of non-negative real numbers $$\eta = \{\eta(w_1,w_2): (w_1,w_2) \in \IW\},$$ and we define for every $w_1 \in \mathfrak{W}$, $\eta(w_1,w_1) = -\sum_{w \in \mathfrak{W}; w \ne w_1} \eta(w_1,w).$
    \item There is a collection of non-negative real numbers $$\xi = \{\xi(v_1,v): (v_1, v) \in  \IV\},$$ and we define for every $v_1 \in \mathfrak{V}$, $\xi(v_1,v_1) = -\sum_{v \in \mathfrak{V}; v \ne v_1} \xi(v_1,v),$ and with $\xi(v_0,v) = 0$ for every $v_0 \in \mathfrak{V}_0$ and $v \in \mathfrak{V}$.
    \item There are $\Re_1$-valued link functions $$\{(\psi^R(s|\theta^R),\psi^W(s|\theta^W),\psi^V(s|\theta^V)): s \ge 0\}$$ which are non-negative, bounded, and $\mathfrak{F}$-predictable processes and with $(\theta^R,\theta^W,\theta^V)$ being finite-dimensional parameters. These processes are assumed to be twice-differentiable with respect to the elements of the parameters. 
\end{itemize}

We describe the specific form of the $\psi$-processes that we utilize. To do so, define the generic mapping $\iota$ given a finite set $A =\{a_1,\ldots,a_m\}:$
$$\iota_{A}(a)=(I(a_2=a),\ldots,I(a_m=a)),$$
an $(m-1)$ row vector with $\{0,1\}$-valued elements. Thus, the mapping $\iota$ is a converter to dummy variables. With row vectors
%
%\begin{eqnarray*}
% B^R(s)=[X,\iota_{\mathfrak{V}}(V(s)),\iota_{\mathfrak{W}}(W(s))] \quad & \mbox{and} & \quad \theta^R=[(\beta^R)\trp,(\gamma^R)\trp,(\kappa^R)\trp]\trp;  \\
% B^W(s)=[X,\iota_{\mathfrak{V}}(V(s)),N^R(s)] \quad & \mbox{and} & \quad \theta^W=[(\beta^W)\trp,(\gamma^W)\trp,(\nu^W)\trp]\trp;  \\
% B^V(s)=[X,\iota_{\mathfrak{W}}(W(s)),N^R(s)]  \quad & \mbox{and} & \quad \theta^V=[(\beta^V)\trp,(\kappa^V)\trp,(\nu^V)\trp]\trp, 
%\end{eqnarray*}
%%
%
\begin{eqnarray*}
& B^R(s)  =  [\iota_{\mathfrak{W}}(W(s)),\iota_{\mathfrak{V}}(V(s)),X]; 
B^W(s)  =  [\log(1 + N^R(s)),\iota_{\mathfrak{V}}(V(s)),X];  \\
& B^V(s)  =  [\log(1 + N^R(s)),\iota_{\mathfrak{W}}(W(s)),X] &
\end{eqnarray*}
and column vectors
\begin{eqnarray*}
& \theta^R  =  [(\theta^{RL})\trp, (\theta^{RH})\trp, (\theta^{RX})\trp]\trp; 
\theta^W  =  [(\theta^{LR})\trp, (\theta^{LH})\trp, (\theta^{LX})\trp]\trp; & \\
& \theta^V =  [(\theta^{HR})\trp, (\theta^{HL})\trp, (\theta^{HX})\trp]\trp, &
\end{eqnarray*}
we let the link functions to be
\begin{eqnarray*}
& \psi^R(s|\theta^R) =  \exp\{B^R(s-) \theta^R\}; 
\psi^W(s|\theta^W)  =  \exp\{B^W(s-) \theta^W\}; & \\ &  \psi(s|\theta^V)  =  \exp\{B^V(s-) \theta^V\}. &
\end{eqnarray*}
Note that we used the form $\log(1 + N^R(s))$ for the internal covariate associated with the RCR process, but in general we could some other mapping $N^R(s) \mapsto g(N^R(s))$ for some suitable map $g(\cdot)$.

In addition to the, possibly infinite-dimensional, parameters $\lambda_{0q}(\cdot)$'s, the finite-dimensional parameters $\alpha_q$'s, $\eta$, $\xi$, the $\theta^R$, $\theta^W$, and $\theta^V$, will constitute all of the model parameters.
For the proposed model, the compensator process vector $A$ will be assumed to satisfy the property that, with $\DG(a)$ denoting the diagonal matrix formed from vector $a$, 
%
%\begin{displaymath}
%\RED{d\langle M, M \rangle (s) = \DG[(dA^R(s))\trp, %(dA^W(s))\trp, (dA^V(s))\trp].}
%\end{displaymath}
%
%This  will in turn imply that
%
\begin{displaymath}
\langle M, M \rangle (s) = \DG[A(s)] = \DG[((A^R(s))\trp, (A^W(s))\trp, (A^V(s))\trp)\trp].
\end{displaymath}
The components of this compensator process vector are given by, for $q \in \IQ$, $(w_1,w_2) \in \IW$, and $(v_1,v) \in \IV$:
\begin{eqnarray*}
& A^R(s;q)  =  \int_0^s Y(t) \lambda_{0q}[\mathcal{E}_q(t)] \rho_q[N^R(t-)|\alpha_q]  \psi^R(s|\theta^R) dt; & \\
& A^W(s;w_1,w_2)  =  \int_0^s Y(t) Y^W(t;w_1) \eta(w_1,w_2) \psi^W(s|\theta^W) dt; \\
& A^V(s;v_1,v)  =  \int_0^s Y(t) Y^V(t;v_1) \xi(v_1,v) \psi^V(s|\theta^V) dt. &
\end{eqnarray*}
In the left-hand side of the equations above, we have suppressed writing the dependence on the model parameters. %With $\DG(a)$ denoting the diagonal matrix formed from %vector $a$, the PQV process satisfies
%
%\begin{displaymath}
%\langle M, M \rangle (s) = %\DG[(A^R(s))\trp, %(A^W(s))\trp, (A^V(s))\trp].
%\end{displaymath}
%
%\RED{
%\begin{displaymath}
%d\langle M, M \rangle (s) = \DG[(dA^R(s))\trp, (dA^W(s))\trp, (dA^V(s))\trp].
%\end{displaymath}
%}
%
Note that the argument of $\lambda_{0q}(\cdot)$ in $A^R(\cdot;q)$ is the (random) effective age $\mathcal{E}_q(\cdot)$. In addition, observe the dynamic nature of this model in that an event occurrence at an infinitesimal time interval $[s,s+ds)$ depends on the history of the processes before time $s$. According to the theory of counting processes, we have the following probabilistic interpretations (statements are almost surely):
\begin{eqnarray*}
& E\{dN^R(s;q)|\mathcal{F}_{s-}\}  =   dA^R(s;q); & \\
& E\{dN^W(s;w_1,w_2)|\mathcal{F}_{s-}\}  =   dA^W(s;w_1,w_2); & \\
& E\{dN^V(s;v_1,v)|\mathcal{F}_{s-}\}  =   dA^V(s;v_1,v), &
\end{eqnarray*}
and
\begin{eqnarray*}
& Var\{dN^R(s;q)|\mathcal{F}_{s-}\} = dA^R(s;q); & \\
& Var\{dN^W(s;w_1,w_2)|\mathcal{F}_{s-}\} =  dA^W(s;w_1,w_2); & \\
& Var\{dN^V(s;v_1,v)|\mathcal{F}_{s-}\} =   dA^V(s;v_1,v), &
\end{eqnarray*}
together with the conditional covariance, given $\mathcal{F}_{s-}$, between any pair of elements of $dN(s)$ being equal to zero, e.g., $Cov\{dN^R(s;q),dN^W(s;w_1,w_2)|\mathcal{F}_{s-}\} = 0$. However, note that we are {\em not} assuming that the components of $N^R$, $N^W$, and $N^V$ are independent, nor even conditionally independent. A way to view this model is that, given $\mathcal{F}_{s-}$, the history just before time $s$, on the infinitesimal interval $[s,s+ds)$, $dN(s) = (dN^R(s)\trp, dN^W(s)\trp, dN^V(s)\trp)\trp$ has a multinomial distribution with parameters $1$ and $dA(s) = (dA^R(s)\trp,dA^W(s)\trp,dA^V(s)\trp)\trp$. As such, for every $s \ge 0$, the following constraint holds:
\begin{displaymath}
dN_\bullet(s) = dN_\bullet^R(s) + dN_\bullet^W(s) + dN_\bullet^V(s)  \in \{0,1\},
\end{displaymath}
with the notation that a subscript of $\bullet$ means the sum over all the appropriate index set, e.g., $dN_\bullet^R(s) = \sum_{q \in \IQ} dN^R(s;q)$ and $dA_\bullet(s) = dA_\bullet^R(s) + dA_\bullet^W(s) + dA_\bullet^V(s)$. The multinomial distribution above could actually be approximated by independent Bernoulli distributions. To see this, if we have real numbers $p_k, k=1,\ldots,K,$ with $0 < p_k \approx 0$ for each $k = 1, \ldots, K,$ and with $\sum_{k=1}^K p_k \approx 0$, then we have the approximation
\begin{displaymath}
\left(1 - \sum_{k=1}^K p_k\right) \approx \prod_{k=1}^K (1 - p_k).
\end{displaymath}
Consequently, in the equation below, the multinomial probability on the left-hand side is approximately the product of (independent) Bernoulli probabilities in the right-hand side.
\begin{eqnarray*}
\lefteqn{ \left\{\prod_{q \in \IQ} [dA_q(s)]^{dN_q^R(s)}\right\} \left\{\prod_{(w_1,w_2) \in \IW} [dA^W(s;w_1,w_2)]^{dN^W(s;w_1,w_2)} \right\} \times } \\ &&
 \left\{\prod_{(v_1,v) \in \IV} [dA^V(s;v_1,v)]^{dN^V(s;v_1,v )}  \right\}  \left\{[1 - dA_\bullet(s)]^{1 - dN_\bullet(s)}\right\} \\ & \approx &
 \left\{\prod_{q \in \IQ} [dA_q(s)]^{dN_q^R(s)} [1 - dA_q^R(s)]^{1- dN_q^R(s)}\right\} \times \\ && \left\{\prod_{(w_1,w_2) \in \IW} [dA^W(s;w_1,w_2)]^{dN^W(s;w_1,w_2)}  
 %\times \right. \\ && \left. 
 [1 - dA^W(s;w_1,w_2)]^{1 - dN^W(s;w_1,w_2)}\right\} \times  \\ && \left\{\prod_{(v_1,v) \in \IV} [dA^V(s;v_1,v)]^{dN^V(s;v_1,v )} [1 -  dA^V(s;v_1,v)]^{1 - dN^V(s;v_1,v )}\right\}.
\end{eqnarray*}
This approximate equivalence informs the manner in which we may generate data from the model later in the sections dealing with an illustrative data set (Section \ref{sec-Illustration}) and the simulation studies (Section \ref{sec: simulation}) where we could have used this product of Bernoulli probabilities in generating synthetic or simulated data.

%\BLUE{(some of likelihood construction goes to supplement)} 

Consider a unit that is still at risk just before time $s$ whose LM process is at state $w_1$ and HS process at state $v_1 \notin \mathfrak{V}_0$. Two questions of interest are: 
\begin{itemize}
\item[(a)] What is the distribution of the next event occurrence? 
\item[(b)] Given in addition that an event occurred at time $s+t$, what are the conditional probabilities of each of the possible events? 
\end{itemize}
Denote by $T$ the time to the next event occurrence starting from time $s$. Then, 
\begin{eqnarray*}
\lefteqn{ \Pr\{T > t|\mathcal{F}_{s-}\} } \\ & = &  \prodi_{u=s}^{s+t} \left[1 - \left(\sum_{q=1}^Q \lambda_{0q}[\mathcal{E}_{q}(u)] \rho_{q}[N^R(u-)|\alpha_q] \exp\{B^R(u-) \theta^R\} - \right.\right. \\ && \left.\left. \eta(w_1,w_1) \exp\{B^W(u-)\theta^W\} - \xi(v_1,v_1) \exp\{B^V(u-)\theta^V\}\right) du \right]  \\
& = & \exp\left\{-\int_{s}^{s+t} \left(\sum_{q=1}^Q \lambda_{0q}[\mathcal{E}_{q}(u)] \rho_{q}[N^R(u-)|\alpha_q] \exp\{B^R(u-) \theta^R\} - \right.\right. \\ && \left.\left. \eta(w_1,w_1) \exp\{B^W(u-)\theta^W\} - \xi(v_1,v_1) \exp\{B^V(u-)\theta^V\}\right) du \right\} \\
& = & \exp\left\{-\exp\{B^R(s-) \theta^R\}\sum_{q=1}^Q \rho_{q}[N^R(s-)|\alpha_q]   \int_{s}^{s+t}  \lambda_{0q}[\mathcal{E}_{q}(u)] du + \right. \\
&& \left.  \eta(w_1,w_1) \exp\{B^W(s-)\theta^W\} t + \xi(v_1,v_1) \exp\{B^V(s-)\theta^V\} t\right\},
\end{eqnarray*}
with the second equality obtained by invoking the product-integral identity 
$$\prodi_{s \in I}  [1 - dA(s)]^{1-dN(s)} = \exp\left\{-\int_{s \in I} dA(s)\right\}$$ 
and since no events in $[s,s+t)$ means $dN_\bullet^R(u) + dN_\bullet^W(u) + dN_\bullet^V(u) = 0$ for $u \in [s,s+t)$, and the last equality arising since, prior to the next event, there will be {\em no changes} in the values of $N^R$, $B^R$, $B^W$, and $B^V$ from their respective values just before time $s$. Note that in particular, if the $\lambda_{0q}$s are constants, corresponding to the hazard rates of an exponential distribution, then the distribution of the time to the next event occurrence is exponential. Given that the event occurred at time $s+t$, then the conditional probability that it was an RCR-type $q$ event is 
$\exp{\{B^R(s-) \theta^R\}} \rho_{q}[N^R(s-)|\alpha_q]  \lambda_{0q}[\mathcal{E}_{q}(s+t)]/C(s,t),$
with
\begin{eqnarray*}
\lefteqn{C(s,t) = \exp{\{B^R(s-) \theta^R \}} \sum_{q^\prime = 1}^Q \rho_{q^\prime}[N^R(s-)|\alpha_{q^\prime}]  \lambda_{0q^\prime}[\mathcal{E}_{q^\prime}(s+t)] - } \\
&& \eta(w_1,w_1) \exp{\{B^W(s-)\theta^W\}} -  \xi(v_1,v_1) \exp{\{B^V(s-)\theta^V\}}.
\end{eqnarray*}
Similarly, the conditional probability that it was a transition to state $w_2$ for the LM process is
$\eta(w_1,w_2) \exp{\{B^W(s-)\theta^W\}}/C(s,t),$
and the conditional probability that it was a transition to state $v$, possibly to an absorbing state, for the HS process is
$\xi(v_1,v) \exp{\{B^V(s-)\theta^V\}}/C(s,t).$
These probabilities demonstrate the competing risks nature of the different possible events (see, for instance, \cite{Tsiatis2005}). They also provide a computational approach to iteratively and dynamically generate data from the joint model for use in simulation studies, with the basic idea being to first generate a time to any type of event, then to mark the type of event or update each of the counting processes by using the above conditional probabilities.

Denoting by $\Theta$ the set of all parameters of the model, the likelihood function arising from observing $D$, with $p_{(W,V)}(\cdot,\cdot)$ the initial joint probability mass function of $(W(0),V(0))$, is given by
\begin{eqnarray}
\label{lik for one unit} \lefteqn{ \mathcal{L}(\Theta|D) = p_{(W,V)}(W(0),V(0)) %\times }  \\  
%&& 
\prodi_{s=0}^\infty \left\{
\left[ \prod_{q \in \IQ} [dA^R(s;q)]^{dN^R(s;q)} \right] \right. \times \nonumber } \\ 
&& \left. \left[ \prod_{(w_1,w_2) \in \IW} [dA^W(s;w_1,w_2)]^{dN^W(s;w_1,w_2)} \right]   \right. \times \nonumber \\ 
&& \left. \left[ \prod_{(v_1,v) \in \IV} [dA^V(s;v_1,v)]^{dN^V(s;v_1,v)} \right]  \left[1 - dA_\bullet(s)\right]^{1-dN_\bullet(s)} 
\right\}.  \nonumber
\end{eqnarray}
This likelihood could be rewritten  in the form
\begin{eqnarray}
\lefteqn{\mathcal{L}(\Theta|D)  =  p_{(W,V)}(W(0),V(0))  \times}  \nonumber  \\
&& \prodi_{s=0}^\infty \left\{
\left[ \prod_{q \in \IQ} [dA^R(s;q)]^{dN^R(s;q)} \right]  \left[ \prod_{(w_1,w_2) \in \IW} [dA^W(s;w_1,w_2)]^{dN^W(s;w_1,w_2)} \right]  \right. \nonumber \\
&& \left. \left[ \prod_{(v_1,v) \in \IV} [dA^V(s;v_1,v)]^{dN^V(s;v_1,v)} \right] 
\right\}  \exp\{-A_\bullet(\infty)\}. \label{lik: one unit alt} 
\end{eqnarray}
Let us examine the meaning of $A_\bullet(\infty) \equiv A_\bullet^R(\infty) + A_\bullet^W(\infty) + A_\bullet^V(\infty)$. Simplifying, we see that this equals
$A_\bullet(\infty) = \int_0^\infty Y(s) T(s) ds = \int_0^{\tau\wedge\tau_A} T(s) ds$,
where
\begin{eqnarray*}
T(s) & = &  \sum_{q \in \IQ} \lambda_{0q}[\mathcal{E}_q(s)] \rho_q[N^R(s-)|\alpha_q] \exp\{B^R(s-) \theta^R\} -  \\
&& \sum_{w_1 \in \mathfrak{W}} Y^W(s;w_1) \eta(w_1,w_1) \exp\{B^W(s-) \theta^W\} - \\
&& \sum_{v_1 \in \mathfrak{V}_1} Y^V(s;v_1) \xi(v_1,v_1) \exp\{B^V(s-) \theta^V\}.
\end{eqnarray*}
Recall that $\eta(w_1,w_1)$ and $\xi(v_1,v_1)$ are non-positive real numbers. Thus, $T(s) ds$ could be interpreted as the total risk of an event, either a recurrent event in the RCR component or a transition in the LM or HS components, occurring from all possible sources (RCR, LM, HS) that the unit is exposed to at the infinitesimal time interval $[s, s+ds)$, given the history $\mathcal{F}_{s-}$ just before time $s$.

We provide further explanations of the elements of the joint model. First, there is a tacit assumption that no more than one event of any type could occur at any given time $s$. {Second, the event rate at any time point $s$ for any type of event is in the presence of the other possible risk events. Thus, consider a specific $q_0 \in \IQ$ and assume that the unit is still at-risk at time $s_0$. Then,
\begin{eqnarray}
\label{RCR probability}
\Pr\{dN_{q_0}^R(s_0) = 1|\mathcal{F}_{s_0-}\} & \approx & \lambda_{0q_0}[\mathcal{E}_{q_0}(s_0)] 
\rho_{q_0}[N^R(s_0-)|\alpha_{q_0}] e^{B^R(s_0-)\theta^R} d{s_0}
\end{eqnarray}
is the conditional probability, given $\mathcal{F}_{s_0-}$, that an event occurred at $[s_0,s_0+ds_0)$ and is of RCR type $q_0$ {\em and} all other event types did not occur, which is the essence of what is called a crude hazard rate, instead of a net hazard rate.} Third, the effective (or virtual) age functions $\mathcal{E}_{q}(\cdot)$s, which are assumed to be dynamically determined and not dependent on any unknown parameters, encodes the impact of performed interventions after each event occurrence. Several possible choices of these functions are:
\begin{itemize}
\item $\mathcal{E}_q(s) = s$ for all $s \ge 0$ and $q \in \IQ$. This is usually referred to as a minimal repair intervention, corresponding to the situation where an intervention simply puts back the system at the age just before the event occurrence. Note that this is the typical effective age process specified in many existing models, which may be interpreted, {\em erroneously}, as the model having {\em no} effective age process.
\item $\mathcal{E}_q(s)  = s - S_{N_\bullet(s-)}$ where $0 = S_0 < S_1 < S_2 < \ldots$ are the successive event times. This corresponds to a perfect intervention, which has the effect of re-setting the time to zero for each of the RCRs after each event occurrence. In a reliability setting, this means that all $Q$ components (unless having exponential lifetimes) are replaced by corresponding identical, but new, components at each event occurrence.
\item $\mathcal{E}_q(s) = s - S_{N_\bullet^R(s-)}^R$ where $0 = S_0^R < S_1^R < S_2^R < \ldots$ are the successive event times of the occurrences of the RCR events.
\item $\mathcal{E}_q(s) = s - S_{qN^R(s-;q)}^R$ where $0 = S_{q0}^R < S_{q1}^R < S_{q2}^R < \ldots$ are the successive event times of the occurrences of RCR events of type $q$. This corresponds to performing a perfect repair at each risk-specific event occurrence for each of the $Q$ risks. {\em This is the effective age function utilized in generating the simulated data for illustrating the estimation procedures in Section \ref{sec-Illustration}, as well as in the simulation study in Section \ref{sec: simulation}.}
\item Other general forms are possible, such as those in  \cite{dorado1997} and \cite{gonzalez2005modelling}, the latter employing ideas of Kijima (see \cite{Kijima89}). See also the discussion on the `reality' of virtual or effective ages in paper \cite{FinCha21}, as well as the  review paper \cite{Beu2021}.
\end{itemize}

One may think that the use of these effective age processes just adds more complications. However, note that as can be seen from some of the possible forms of these processes, one actually gains generality, since whatever analytical results are obtained for the general model will then apply to these specific situations, which includes many existing models. In addition, it should be pointed out that models that seem to be devoid of an effective age, such as the Cox PHM or the Andersen-Gill model, do in fact have an {\em intrinsic} effective age, which is either $\mathcal{E}_q(s) = s$  or $\mathcal{E}_q(s)  = s - S_{N_\bullet^R(s-)}$ and since we usually have $Q=1$ in these earlier models.

Fourth, the impact of accumulating RCR event occurrences, which could be adverse, but could also be beneficial as in software engineering applications, is incorporated in the model through the $\rho_q(\cdot|\cdot)$ functions. One possible choice is an exponential function, such as 
$$\rho_q(N^R(s-)|\alpha_q) = \exp\{[\log(1 + N^R(s-))]\trp \alpha_q\},\ q=1,\ldots,Q,$$ 
the form that was utilized in generating the illustrative simulated data. However, other choices could be made as well. Finally, the modulating exponential link function in the model is for the impact of the covariates as well as the values of the RCR, LM, and HS processes just before the time of interest, with the vector of coefficients quantifying the effects of the covariates. The use of $\log(1+N^R(s-))$ in the model could be viewed as using them as internal covariates, or that the dynamic model is a self-exciting model.

Similar interpretations hold for the parameters $\{\eta(w_1,w_2): (w_1,w_2) \in \IW\}$ and $\{\xi(v_1,v): (v_1,v) \in \IV\}$. Thus, if just before time $s_0$, $W(s_0-) = w_1$ and $V(s_0-) = v_1$, indicated by $\mathcal{F}_{s_0-}(w_1,v_1)$, then 
\begin{eqnarray}
\lefteqn{ \Pr\{W(s_0+ds_0) =  w_2|\mathcal{F}_{s_0}(w_1,v_1)\}  \approx } \nonumber  \\  &&  \eta(w_1,w_2)
\exp\{[\log(1+N^R(s_0-))]\trp \theta^{LR} + \theta^{LH}_{j(v_1)} + X\theta^{LX}\} ds_0;  \label{LM probability}  \\
\lefteqn{ \Pr\{V(s_0+ds_0) = v|\mathcal{F}_{s_0}(w_1,v_1)\}  \approx }  \nonumber  \\  &&  \xi(v_1,v)
\exp\{ [\log(1+N^R(s_0-))]\trp \theta^{HR} + \theta^{HL}_{j(w_1)} + X\theta^{HX}\} ds_0, \label{HS probability}
\end{eqnarray}
where $j(v_1)$ is the index associated with $v_1$ in $\mathfrak{V}_1$ and $j(w_1)$ is the index associated with $w_1$ in $\mathfrak{W}$.

\subsection{Special Case: Independent Poisson Processes and CTMCs for One Unit}
\label{subsec: special case}

There is a special case arising from this general joint model obtained when we set $\lambda_{0q}(s) = \lambda_{0q}$, $q \in \IQ$; $\rho_q = 1$; $\theta^R = 0$; $\theta^W = 0$; and $\theta^V = 0$. In this situation, we have
\begin{eqnarray*}
dA^R(s;q) & = & \lambda_{0q} Y(s) ds, q \in \IQ; 
\\ dA^W(s;w_1,w_2) & = & \eta(w_1,w_2) Y(s) Y^W(s;w_1) ds, (w_1,w_2) \in \IW; 
\\ dA^V(s;v_1,v) & = & \xi(v_1,v) Y(s) Y^W(s,v_1) ds, (v_1,v) \in \IV.
\end{eqnarray*}
It is easy to see that this particular model coincides with the model where we have the following situations:
\begin{itemize}
\item[(i)] $N^R(\cdot;q), q \in \IQ,$ are independent homogeneous Poisson processes with respective rates $\lambda_{0q}, q \in \IQ$;
\item[(ii)] $W(\cdot)$ is a continuous-time Markov chain (CTMC) with infinitesimal generator matrix (IGM)  consisting of $\{\eta(w_1,w_2)\}$;
\item[(iii)] $V(\cdot)$ is a CTMC with IGM consisting of $\{\xi(v_1,v_2)\}$; 
\item[(iv)] $N^R$, $W$, and $V$ are independent; and
\item[(v)] the processes are observed over $[0,\min(\tau,\tau_A)]$, where $\tau$ is the end of monitoring period, while $\tau_A$ is the absorption time of $V$ into $\mathfrak{V}_0$ if $\mathfrak{V}_0 \ne \emptyset$.
\end{itemize}
In this special situation, the $\lambda_{0q}$'s are both crude and net hazard rates. Also, due to the memoryless property of the exponential distribution, interventions performed after each event occurrence will have no impact in the succeeding event occurrences. This specific joint model further allows us to provide an operational interpretation of the model parameters. Thus, suppose that at time $s$, the LM process is at state $w_1$ and the HS process is at state $v_1 \notin \mathfrak{V}_0$. Then, the distribution of the time to the next event occurrence of any type (the holding or sojourn time at the current state configuration) has an exponential distribution with parameter $C = \lambda_{0\bullet} - \eta(w_1,w_1) - \xi(v_1,v_1)$. When an event occurs, then the (conditional) probability that it is (i) an RCR event of type $q$ is $\lambda_{0q}/C$; (ii) a transition to LM state $w_2 \ne w_1$ is $\eta(w_1,w_2)/C$; or (iii) a transition to an HS state $v \ne v_1$ is $\xi(v_1,v)/C$. This is the essence of the competing risks nature of all the possible event types: an RCR event, an LM transition, and an HS transition.
As such, the more general model could be viewed as an extension of this basic model with independent Poisson processes for the RCR component and CTMCs for the LM and HS components. 
For this special case, the likelihood function in (\ref{lik: one unit alt}) simplifies to the expression
\begin{eqnarray}
\lefteqn{\mathcal{L}(\Theta|D)  =  p_{(W,V)}(W(0),V(0)) %\times \nonumber }  \\
\exp\left\{-\int_0^{\tau\wedge\tau_A} T(s) ds\right\}
\left[\prod_{q \in \IQ} \lambda_{0q}^{N^R(\tau\wedge\tau_A;q)}\right]   \times  \nonumber} \\
&& 
\left[\prod_{(w_1,w_2) \in \IW} \eta(w_1,w_2)^{N^W(\tau\wedge\tau_A;w_1,w_2)}\right] %\times \nonumber   \\
%&& 
\left[\prod_{(v_1,v) \in \IW} \xi(v_1,v)^{N^V(\tau\wedge\tau_A;v_1,v)}\right]. \label{lik: special case one unit}  
\end{eqnarray}
with 
$$T(s) = \lambda_{0\bullet} - \sum_{w_1 \in \mathfrak{W}} \eta(w_1,w_1) Y^W(s;w_1) - \sum_{v_1 \in \mathfrak{V}_1} \xi(v_1,v_1) Y^V(s;v_1).$$ 
Note that $T(s) = T(s;\lambda_0,\eta,\xi)$ is a quantity, since it depends on the parameters, instead of a sample statistic.

\section{Parametric Estimation of Model Parameters}
\label{sec: estimation - parametric}

%\subsection{Parametric Estimation}
%\label{subsec: estimation - parametric}

Having introduced the joint model, we address in this section the problem of making inferences about the model parameters. We assume that we are able to observe $n$ independent units, with the $i$th unit having data $D_i = (X_i,N_i,\mathcal{E}_i,Y_i,Y_i^W, Y_i^V)$ as in (\ref{data: one unit}). The full sample data is represented by
\begin{equation}
\label{sample data}
\mathbf{D} = (D_1,D_2,\ldots,D_n),
\end{equation}
while the model parameters is represented by, with the convention that $q \in \IQ$,  $(w_1,w_2) \in \IW$, and $(v_1,v) \in \IV$,
\begin{eqnarray*}
\Theta & \equiv & \left[\{\lambda_{0q}(\cdot), \alpha_q\}, \{\eta(w_1,w_2)\},
\{\xi(v_1,v)\}, \theta^R, \theta^W, \theta^V\right].
\end{eqnarray*}
The $\lambda_{0q}$s could be parametrically-specified, hence will have finite-dimensional parameters, so $\Theta$ will also then be finite-dimensional. Except for the special case mentioned above, our main focus will be the case where the $\lambda_{0q}$s are nonparametric. The distributions, $G_i$s, of the end of monitoring periods, $\tau_i$s, also have model parameters, but they are nuisance parameters and not of main interest. To visualize the type of sample data set that accrues, Figure \ref{sample data picture} provides a picture of a sample of 30 units from a simulated sample data with $n = 100$ units.

\medskip
\centerline{[INSERT FIGURE \ref{sample data picture} HERE]}
\medskip

The full likelihood function, given $\mathbf{D}$, is
%
%\begin{equation}
%\label{lik: full sample}
$\mathcal{L}(\Theta|\mathbf{D}) = \prod_{i=1}^n \mathcal{L}(\Theta|\mathbf{D}_i),$
%Ò\end{equation}
%
where the $\mathcal{L}(\Theta|D_i)$ is of the form in (\ref{lik: one unit alt}). If the $\lambda_{0q}(\cdot)$s are parametrically-specified, then estimators of the finite-dimensional model parameters could be obtained as the maximizers of this full likelihood function, and their finite and asymptotic properties will follow from the general results for maximum likelihood estimators based on counting processes; see, for instance,  \cite{Bor84} and \cite{ander1993}.

We illustrate this situation for the special case of the model given in subsection \ref{subsec: special case}, so that the parameter is simply $\Theta = [\{\lambda_{0q}\},\{\eta(w_1,w_2)\},\{\xi(v_1,v)\}]$. In this situation, from (\ref{lik: special case one unit}), the full likelihood reduces to, with $\tau_i^* = \tau_i \wedge \tau_{iA}$,
\begin{eqnarray*}
\lefteqn{\mathcal{L}(\Theta|\mathbf{D}) = \prod_{i=1}^n p_{(W,V)}(W_i(0),V_i(0)) \times} \\
&& \left[\prod_{q\in\IQ} \lambda_{0q}^{\sum_{i=1}^n N_i^R(\tau_i^*;q)} \right] 
\left[\prod_{(w_1,w_2)\in\IW} \eta(w_1,w_2)^{\sum_{i=1}^n N_i^W(\tau_i^*;w_1,w_2)}\right] \times \\
&& \left[\prod_{(v_1,v)\in\IV} \xi(v_1,v)^{\sum_{i=1}^n N_i^V(\tau_i^*;v_1,v)}\right]
\exp\left\{-\sum_{i=1}^n \int_0^{\tau_i^*} T_i(s;\lambda_0,\eta,\xi) ds\right\},
\end{eqnarray*}
where
\begin{displaymath}
T_i(s;\lambda_0,\eta,\xi) = \lambda_{0\bullet} - \sum_{w_1\in\mathfrak{W}} \eta(w_1,w_1) Y_i^W(s;w_1) - \sum_{v_1\in\mathfrak{V}}  \xi(v_1,v_1) Y_i^V(s;v_1).
\end{displaymath}
The score function $U(\Theta|\mathbf{D}) = \nabla_\Theta \log \mathcal{L}(\Theta|\mathbf{D})$ has elements,
where indices are $q \in \IQ$, $(w_1,w_2) \in \IW$, and $(v_1,v) \in \IV$,
\begin{eqnarray*}
U^R(\Theta;q) & = & \frac{\sum_{i=1}^n N_i^R(\tau_i^*;q)}{\lambda_{0q}} - \sum_{i=1}^n \tau_i^*; \\
U^W(\Theta;w_1,w_2) & = & \frac{\sum_{i=1}^n N_i^W(\tau_i^*;w_1,w_2)}{\eta(w_1,w_2)} - \sum_{i=1}^n \int_0^{\tau_i^*} Y_i^W(s;w_1) ds; \\
U^V(\Theta;v_1,v) & = & \frac{\sum_{i=1}^n N_i^V(\tau_i^*;v_1,v)}{\xi(v_1,v)} - \sum_{i=1}^n \int_0^{\tau_i^*} Y_i^V(s;v_1) ds.
\end{eqnarray*}
Equating these equations to zeros yield the ML estimators of the parameters, which are given below. 
\begin{eqnarray*}
\hat{\lambda}_{0q} & = & \frac{\sum_{i=1}^n N_i^R(\tau_i^*;q)}{\sum_{i=1}^n \tau_i^*} = \frac{\sum_{i=1}^n \int_0^\infty dN_i^R(s;q)}{\sum_{i=1}^n \int_0^\infty Y_i(s) ds}; \\
\hat{\eta}(w_1,w_2) & = & \frac{\sum_{i=1}^n N_i^W(\tau_i^*;w_1,w_2)}{\sum_{i=1}^n \int_0^{\tau_i^*} Y_i^W(s;w_1) ds} = \frac{\int_0^\infty dN_i^W(s;w_1,w_2)}{\sum_{i=1}^n \int_0^\infty Y_i(s) Y_i^W(s;w_1) ds}; \\
\hat{\xi}(v_1,v) & = & \frac{\sum_{i=1}^n N_i^V(\tau_i^*;v_1,v)}{\sum_{i=1}^n \int_0^{\tau_i^*} Y_i^V(s;v_1) ds} = \frac{\int_0^\infty dN_i^V(s;v_1,v)}{\sum_{i=1}^n \int_0^\infty Y_i(s) Y_i^V(s;v_1) ds}.
\end{eqnarray*}
These estimators possess the interpretation of being the observed ``occurrence-exposure" rates, since
the numerators are total event counts:
\begin{itemize}
\item $\sum_{i=1}^n N_i^R(\tau_i^*;q)$ is the total number of observed RCR type $q$ events; 
\item $\sum_{i=1}^n N_i^W(\tau_i^*;w_1,w_2)$ is the total number of observed transitions in the LM process from state $w_1$ into $w_2$;
\item $\sum_{i=1}^n N_i^V(\tau_i^*;v_1,v)$ is the total number of observed transitions in the HS process from state $v_1$ into $v$;
\end{itemize}
while the denominators are total observed exposure times:
\begin{itemize}
\item $\sum_{i=1}^n \int_0^\infty Y_i(s) ds$ is the total time at-risk for all the units; 
\item $\sum_{i=1}^n \int_0^\infty Y_i(s)Y_i^W(s;w_1) dw$ is the total observed time of all units that they were at-risk for a transition in the LM process from state $w_1$;
\item $\sum_{i=1}^n \int_0^\infty Y_i(s) Y_i^V(s;v_1) ds$ is the total observed time of all units that they were at-risk for a transition in the HS process from state $v_1$. 
\end{itemize}
An important and crucial point to emphasize is that the exposure times take into account the time after the last observed events in each component process until the end of monitoring, whether it is a censoring (reaching $\tau_i$) or an absorption (reaching $\tau_{iA}$). If one ignores these right-censored times, then the estimators could be severely biased. This is a critical aspect we mentioned in the introductory section and re-iterate at this point that this should not be glossed over when dealing with recurrent event models.

The elements of the observed information matrix, $I(\Theta;\mathbf{D}) = - \nabla_{\Theta\trp} U(\Theta|\mathbf{D})$, which is a diagonal matrix, have diagonal elements:
\begin{eqnarray*}
I^R(\Theta;q) & = & \frac{\sum_{i=1}^n N_i^R(\tau_i^*;q)}{\lambda_{0q}^2}, q \in \IQ; \\
I^W(\Theta;w_1,w_2) & = & \frac{\sum_{i=1}^n N_i^W(\tau_i^*;w_1,w_2)}{\eta(w_1,w_2)^2}, (w_1,w_2) \in \IW; \\
I^V(\Theta;v_1,v) & = &  \frac{\sum_{i=1}^n N_i^V(\tau_i^*;v_1,v)}{\xi(v_1,v)^2}, (v_1,v) \in \IV.
\end{eqnarray*}
Abbreviating the estimators into $\hat{\Theta} = (\hat{\lambda}_0, \hat{\eta}, \hat{\xi})$, we obtain the asymptotic result, that as $n \rightarrow \infty$,
\begin{equation}
\label{approx asymptotic special case}
\hat{\Theta}  \sim \mbox{AsyMVN}(\Theta,I(\hat{\Theta};\mathbf{D})^{-1}),
\end{equation}
with \mbox{AsyMVN} meaning asymptotically multivariate normal. Thus, this result seems to indicate that the RCR, LM, and HS components or the estimators of their respective parameters do not have anything to do with each other, which appears intuitive since the RCR, LM, and HS processes were assumed to be independent processes to begin with. But, let us examine this issue further. The result in (\ref{approx asymptotic special case}) is an approximation to the theoretical result that
\begin{equation}
\label{asymptotic special case}
\hat{\Theta}  \sim \mbox{AsyMVN}\left(\Theta,\frac{1}{n}\mathfrak{I}(\Theta)^{-1}\right),
\end{equation}
where $\frac{1}{n} I(\hat{\Theta};\mathbf{D}) \stackrel{pr}{\rightarrow} \mathfrak{I}(\Theta)$. Evidently, $\mathfrak{I}$ is a diagonal matrix, so let us examine its diagonal elements. Let $q \in \IQ$. Then we have, with `$\plim$' denoting in-probability limit,
\begin{eqnarray*}
\lefteqn{ \lambda_{0q}^2 \mathfrak{I}^R(\Theta;q)  = \plim_{n\rightarrow\infty} \frac{1}{n} \sum_{i=1}^n \int_0^\infty dN_i^R(s;q) } \\
& = & \plim_{n\rightarrow\infty} \frac{1}{n} \sum_{i=1}^n \int_0^\infty dM_i^R(s;q) + %\\ && 
\lambda_{0q}  \left[\plim_{n\rightarrow\infty} \frac{1}{n} \sum_{i=1}^n \int_0^\infty Y_i(s) ds \right].
\end{eqnarray*}
The first term on the right-hand side (RHS) converges in probability to zero by the weak law of large numbers and the zero-mean martingale property. The second term in the RHS converges in probability to its expectation, hence
\begin{displaymath}
\mathfrak{I}^R(\Theta;q) = \frac{1}{\lambda_{0q}}  \left[ \int_0^\infty \left\{ \lim_{n\rightarrow\infty} \frac{1}{n} \sum_{i}^n E[Y_i(s)]\right\} ds \right].
\end{displaymath}
But, now,
\begin{eqnarray*}
\lim_{n\rightarrow\infty} \frac{1}{n} \sum_{i}^n E[Y_i(s)] & = & \lim_{n\rightarrow\infty} \frac{1}{n} \sum_{i}^n \Pr\{\tau_i \ge s\} \Pr\{\tau_i^A \ge s\} \\
& = & \lim_{n\rightarrow\infty} \frac{1}{n} \sum_{i}^n  \bar{G}_i(s-) \Pr\{V_i(u) \notin \mathfrak{V}_0, u \le s\}.
\end{eqnarray*}
with $\bar{G}_i = 1 - G_i$. The last probability term above will depend on the generators $\{\xi(v_1,v): (v_1,v) \in \IV\}$ of the CTMC $\{V(s): s \ge 0\}$, so that the theoretical Fisher information or the asymptotic variance associated with the estimator $\hat{\lambda}_{0q}$ depends after all on the HS process, as well as on the $G_i$s, contrary to the seemingly intuitive expectation that it should not depend on the LM and HS processes. This result is a subtle one which arises because of the structure of the observation processes. If  $G_i = G, i=1,\ldots,n$, then we have that
\begin{displaymath}
\mathfrak{I}^R(\Theta;q) = \frac{1}{\lambda_{0q}} \int_0^\infty \bar{G}(s-) \Pr\{V(u) \notin \mathfrak{V}_0, u \le s\} ds,
\end{displaymath}
since $\Pr\{V_i(u) \notin \mathfrak{V}_0, u \le s\} = \Pr\{V(u) \notin \mathfrak{V}_0, u \le s\}, i=1,\ldots,n$. If we denote by $\Gamma$ the generator matrix of $\{V(s): s \ge 0\}$ and let $\Gamma_1$ be the sub-matrix associated with the $\mathfrak{V}_1$ states, then if $p_0^V \equiv (p_V(v_1), v_1 \in \mathfrak{V}_1)\trp$ is the initial probability mass function of $V(0)$, we have
\begin{displaymath}
\Pr\{V(u) \notin \mathfrak{V}_0, u \le s\} = \Pr\{V(s) \in \mathfrak{V}_1\} = (p_0^V)\trp \left[\exp\{s\Gamma_{11}\}\right] 1_{|\mathfrak{V}_1|}
\end{displaymath}
where the matrix exponential is
%
%\begin{displaymath}
$\exp\{s\Gamma_{11}\}  \equiv \sum_{k=0}^\infty {s^k \Gamma_{11}^k}/{k!}$
%\end{displaymath}
%
and $1_K$ is a column vector of $1$s of dimension $K$. Thus, we obtain
\begin{displaymath}
\mathfrak{I}^R(\Theta;q) = \frac{1}{\lambda_{0q}}  (p_0^V)\trp \sum_{k=0}^\infty  \left[ \int_0^\infty \bar{G}(s-) \frac{s^k}{k!} ds \right] \Gamma_{11}^k 1_{|\mathfrak{V}_1|}. 
\end{displaymath}
For example, if $\bar{G}(s) = \exp(-\nu s)$, that is, $\tau_i$'s are exponentially-distributed with mean $1/\nu$, then the above expression simplifies to
$$\mathfrak{I}^R(\Theta;q) = \frac{1}{\lambda_{0q}} \frac{1}{\nu} (p_0^V)\trp \left[ \sum_{k=0}^\infty \left(\frac{\Gamma_{11}}{\nu}\right)^k\right] 1_{|\mathfrak{V}_1|}.$$
For computational purposes, one may use an eigenvalue decomposition of $\Gamma_{11}$: $\Gamma_{11} = U \DG(d) U^{-1}$ where $d$ consists of the eigenvalues of $\Gamma_{11}$ and $U$ is the matrix of eigenvectors associated with the eigenvalues $d$. The main point of this example though is the demonstration that estimators of the parameters associated with the RCR, LM, or HS process will depend on features of the other processes, {\em even} when one starts with independent processes.

We remark that the estimators $\hat{\lambda}_{0q}$s, $\hat{\eta}(w_1,w_2)$s, and $\hat{\xi}(v_1,v)$s could also be derived as method-of-moments estimators using the martingale structure. The inverse of the observed Fisher information matrix coincides with an estimator using the optional variation (OV) matrix process, while the inverse of the Fisher information matrix coincides with the limit-in-probability of the predictable quadratic variation matrix. To demonstrate for $\lambda_{0q}$, we have that
$$\left\{\sum_{i=1}^n M_i^R(s;q) = \sum_{i=1}^n \left[N_i^R(s;q) - \int_0^s Y_i(t) \lambda_{0q} dt\right]: s \ge 0\right\}$$
is a zero-mean square-integrable martingale. Let $s \rightarrow \infty$ and $\sum_{i=1}^n M_i^R(\infty;q) = 0$. Then solving for $\lambda_{0q}$ we obtain $\hat{\lambda}_{0q}$. Next, we have
\begin{eqnarray*}
\hat{\lambda}_{0q} & =  &  \frac{\sum_{i=1}^n \int_0^\infty dN_i^R(s;q)}{\sum_{i=1}^n \int_0^\infty Y_i(s) ds} = \lambda_{0q} + \frac{\sum_{i=1}^n \int_0^\infty dM_i^R(s;q)}{\sum_{i=1}^n \int_0^\infty Y_i(s) ds}
\end{eqnarray*}
so that
\begin{eqnarray*}
\sqrt{n}[\hat{\lambda}_{0q} - \lambda_{0q}] = \left( \int_0^\infty \frac{1}{n} \sum_{i=1}^n Y_i(s) ds \right)^{-1} \frac{1}{\sqrt{n}} \sum_{I=1}^n \int_0^\infty dM_i^R(s;q).
\end{eqnarray*}
We have already seen where $\int_0^\infty \frac{1}{n} \sum_{i=1}^n Y_i(s) ds$ converges in probability, whereas by the Martingale Central Limit Theorem, we have that
\begin{eqnarray*}
\frac{1}{\sqrt{n}} \sum_{i=1}^n \int_0^\infty dM_i^R(s;q) \stackrel{d}{\rightarrow} N(0,\sigma^2_R(q))
\end{eqnarray*}
with 
\begin{eqnarray*}
\sigma_R^2(q) & = & \int_0^\infty \left\{\plim_{n \rightarrow \infty} \frac{1}{n} \sum_{i=1}^n d \langle M_i^R(\cdot;q), M_i^R(\cdot;q) \rangle (s) \right\} \\
 & = &  \int_0^\infty \left\{\plim_{n\rightarrow\infty} \frac{1}{n} \sum_{i=1}^n Y_i(s) \lambda_{0q} \right\} ds.
\end{eqnarray*}
Therefore, we have
\begin{eqnarray*}
\lefteqn{ \sqrt{n} [\hat{\lambda}_{0q} - \lambda_{0q}] \stackrel{d}{\rightarrow}  N\left(0,\frac{\lambda_{0q}}{\int_0^\infty \left[ \plim_{n\rightarrow\infty} \frac{1}{n} \sum_{i=1}^n Y_i(s) \right] ds}\right) } \\
&& =  N\left(0,\frac{\lambda_{0q}}{\int_0^\infty \left[ \lim_{n\rightarrow\infty} \frac{1}{n} \sum_{i=1}^n E[Y_i(s)] \right] ds} \right), 
\end{eqnarray*}
which is the same result stated above using ML theory. Analogous asymptotic derivations can be done for $\hat{\eta}(w_1,w_2)$ and $\hat{\xi}(v_1,v)$, though the resulting limiting variances will involve expected occupation times for their respective states of the $W_i$-processes coming from the $Y_i^W(\cdot;w_1)$ terms and the $V_i$-processes from the $Y_i^V(\cdot;v_1)$ terms. Note that, asymptotically, these estimators are independent, {\em but} their limiting variances depend on the parameters from the other processes.

\section{Semi-Parametric Estimation of Model Parameters}
\label{sec: estimation - semiparametric}

In this section we consider the estimation of model parameters when the baseline hazard rate functions $\lambda_{0q}(\cdot)$s are specified nonparametrically. Recall that $\Lambda_{0q}(t) = \int_0^t \lambda_{0q}(u) du, q \in \IQ$ are the associated baseline cumulative hazard functions and $\bar{F}_{0q}(t) = \exp\{-\Lambda_{0q}(t)\}, q\in\IQ$ are the associated baseline survivor functions. 
%To simplify notation, we let
%
%\begin{eqnarray*}
%& \psi_R(s|\theta^R)  =  \exp\{B^R(s) \theta^R\};  & \\
% & \psi^W(s|\theta^W)  =  \exp\{B^W(s) \theta^W\};  & \\
% & \psi^V(s|\theta^V)  =  \exp\{B^V(s) \theta^V\}. &
%\end{eqnarray*}
%
%Using these functions, 
Recall that, for $q \in \IQ$, $(w_1,w_2) \in \IW$, and $(v_1,v) \in \IV$, we have
\begin{eqnarray*}
dA^R(s;q) & = & Y(s) \lambda_{0q}[\mathcal{E}_q(s)] \rho_q(N^R(s-)|\alpha_q) \psi^R(s|\theta^R) ds; \\
dA^W(s;w_1,w_2) & = & Y(s) Y^W(s;w_1) \eta(w_1,w_2) \psi^W(s|\theta^W) ds; \\
dA^V(s;v_1,v) & = & Y(s) Y^V(s;v_1) \xi(v_1,v) \psi^V(s|\theta^V) ds.
\end{eqnarray*}
We abbreviate the vector of model parameters into $\Theta  \equiv (\Lambda_0,\alpha,\eta,\xi,\theta^R,\theta^W,\theta^V)$. Our goal is to obtain estimators for these parameters based on the sample data $\mathbf{D} = (D_1, D_2,\ldots,D_n)$. 

The basic approach to obtain estimators is to first assume that $(\alpha,\theta^R,\theta^W,\theta^V)$ are known, then obtain `estimators' of $(\Lambda_0,\eta,\xi)$. Having obtained these `estimators', in quotes since they are not yet estimators when $(\alpha,\theta^R,\theta^W,\theta^V)$ are unknown, we plug them into the likelihood function to obtain a profile likelihood function. From the resulting profile likelihood function, which depends on $(\alpha,\theta^R,\theta^W,\theta^V)$, we obtain its maximizers with respect to these finite-dimensional parameters to obtain their estimators. These estimators are then plugged into the `estimators' of $(\Lambda_0,\eta,\xi)$ to obtain their final estimators.

The full likelihood function based on the sample data $\mathbf{D} = (D_1,\ldots,D_n)$ could be written as a product of three ``major'' likelihood functions corresponding to the three model components:
\begin{eqnarray}
%\label{full likelihood}
\lefteqn{\mathcal{L}(\Theta|\mathbf{D})  = 
\left[ \prod_{q \in \IQ} \mathcal{L}^R(\Lambda_{0q},\alpha,\theta^R;q|\mathbf{D}) \right] \times \nonumber}  \\
&& \left[ \prod_{(w_1,w_2) \in \IW} \mathcal{L}^W(\eta,\theta^V;w_1,w_2|\mathbf{D}) \right] 
\left[ \prod_{(v_1,v) \in \IV} \mathcal{L}^V(\xi,\theta^W;v_1,v|\mathbf{D}) \right] \label{full lik factor}  
\end{eqnarray}
where, suppressing writing of the parameters in the functions,
\begin{eqnarray*}
\mathcal{L}^R(\Lambda_{0q},\alpha,\theta^R;q|\mathbf{D}) & = & \left\{\prod_{i=1}^n \prodi_{s=0}^\infty \left[dA_i^R(s;q)\right]^{dN_i^R(s;q)}\right\} \times \\ && 
\exp\left\{-\sum_{i=1}^n \int_0^\infty dA_i^R(s;q)\right\}; \\
\mathcal{L}^W(\eta,\theta^W;w_1,w_2|\mathbf{D}) & = & \left\{ \prod_{i=1}^n \prodi_{s=0}^\infty \left[dA_i^W(s;w_1,w_2)\right]^{dN_i^W(s;w_1,w_2)}\right\} \times \\ && \exp\left\{-\sum_{i=1}^n\int_0^\infty dA_i^W(s;w_1,w_2)\right\}; \\
\end{eqnarray*}
\begin{eqnarray*}
\mathcal{L}^V(\eta,\theta^V;v_1,v|\mathbf{D}) & = & \left\{ \prod_{i=1}^n \prodi_{s=0}^\infty \left[dA_i^V(s;v_1,v)\right]^{dN_i^V(s;v_1,v)}\right\} \times \\ && \exp\left\{-\sum_{i=1}^n\int_0^\infty dA_i^V(s;v_1,v)\right\}.
\end{eqnarray*}
Let $0 = S_0 < S_1 < S_2 < \ldots < S_K < S_{K+1} =\infty$ be the ordered distinct times of \underline{any} type of event occurrence for all the $n$ sample units. Also, let $0 = T_0 < T_1 < T_2 < \ldots  < T_L < T_{L+1} = \infty$ be the ordered distinct values of the set $\{\mathcal{E}_{iq}(S_j): i = 1, \ldots, n; q \in \IQ; j=0,1,\ldots,K\}$. Recall that $\tau_i^* = \tau_i \wedge \tau_{iA}$. Observe that both $\{S_k: k=0,1,\ldots,K,K+1\}$ and $\{T_l: l=0,1,\ldots,L,L+1\}$ partition $[0,\infty)$. For each $i=1,\ldots,n$, and $q \in \IQ$, $\mathcal{E}_{iq}(\cdot)$ is observed, hence defined, only on $[0,\tau_i^*)$. However, for notational convenience, we define $\mathcal{E}_{iq}(s) = 0$ for $s > \tau_i^*$. In addition, on each non-empty interval $(S_{j-1} \wedge \tau_i^*,S_j \wedge \tau_i^*]$, $\mathcal{E}_{iq}(\cdot)$ has an inverse which will be denoted by $\mathcal{E}_{iqj}^{-1}(\cdot)$. Henceforth, for brevity of notation, we adopt the convention that $0/0 = 0$.

Following the use of doubly-indexed processes in \cite{pena2007semiparametric} for a setting with only one recurrent event type and without LM and HS processes, define the doubly-indexed processes $\{(N_i^R(s,t;q),A_i^R(s,t; q|\Lambda_{0q},\alpha_1,\theta^R): (s,t) \in \Re_+^2\}$ via
\begin{eqnarray*}
 Z_{iq}(v,t) & = &  I\{\mathcal{E}_{iq}(t) \le v\};  \\
N_i^R(s,t;q) & = & \int_0^s Z_{iq}(v,t) dN_i^R(v;q); \\ A_i^R(s,t;q|\Lambda_{0q},\alpha_q,\theta^R) & = & \int_0^s Z_{iq}(v,t) dA_i^R(v;q|\Lambda_{0q},\alpha_q,\theta^R).
\end{eqnarray*}
Also, let 
\begin{eqnarray*}
N_\bullet^R(s,t;q) & = & \sum_{i=1}^n N_i^R(s,t;q); \\  A_\bullet^R(s,t;q|\Lambda_{0q},\alpha_q,\theta^R) & = & \sum_{i=1}^n A_i^R(s,t;q|\Lambda_{0q},\alpha_q,\theta^R).
\end{eqnarray*}
Then, for fixed $t$, $\{M_\bullet^R(s,t;q|\Lambda_{0q},\alpha_q,\theta^R) = N_\bullet^R(s,t;q) - A_\bullet^R(s,t;q|\Lambda_{0q},\alpha_q,\theta^R): s \ge 0\}$ is a zero-mean square-integrable martingale. Thus, $E[N_\bullet^R(s,t;q)] = E[A_\bullet^R(s,t;q|\Lambda_{0q},\alpha_q,\theta^R)]$.
It can be shown (see Proposition \ref{prop: change-of-variable}) that for $q \in \IQ$, 
$$A_\bullet^R(s,t;q|\Lambda_{0q},\alpha_q,\theta^R) = \int_0^t S_q^{0R}(s,u|\alpha_q,\theta^R) d\Lambda_{0q}(u),$$
where
\begin{eqnarray}
S_q^{0R}(s,u|\alpha_q,\theta^R) & = & \sum_{i=1}^n\sum_{j=1}^{K+1} \left\{ \left[\frac{ \rho_q[N_i^R(\mathcal{E}_{iqj}^{-1}(u)-)|\alpha_q] \psi_i^R(\mathcal{E}_{iqj}^{-1}(u)|\theta^R)}{\mathcal{E}_{iq}^\prime[\mathcal{E}_{iqj}^{-1}(u)]}\right] \times \right. \nonumber \\ && \left. 
I\{\mathcal{E}_{iq}(S_{j-1} \wedge \tau_i^* \wedge s) < u \le \mathcal{E}_{iq}(S_{j} \wedge \tau_i^* \wedge s)\} \right\};
\label{defn Sq0R}
\end{eqnarray}
and $dN_\bullet^R(s,T_l;q) = \sum_{i=1}^n \sum_{j=1}^K I\{S_j \le s; \mathcal{E}_{iq}(S_j) = T_l\} dN_i^R(S_j;q)$, where $\psi_i^R$ uses the data for the $i$th unit, e.g., $\psi_i^R(s|\theta^R) = \exp\{B_i^R(s-)\theta^R\}$. %

\begin{proposition}
\label{estimator of Lambda when parameter known}
If $(\alpha_q,\theta^R)$ is known, then an estimator of $\Lambda_{0q}(t)$ is
\begin{equation}
\label{NAE-type 1}
\hat{\Lambda}_{0q}(s,t|\alpha_q,\theta^R) = \sum_{l : T_l \le t} \frac{dN_\bullet^R(s,T_l;q)}{S_q^{0R}(s,T_l|\alpha_q,\theta^R)} =\int_0^t  \frac{N_\bullet^R(s,du;q)}{S_q^{0R}(s,u|\alpha_q,\theta^R)}.
\end{equation}
\end{proposition}

\begin{proof}
This is immediate using method-of-moments type argument.
\end{proof}

The estimator in (\ref{NAE-type 1}) is reminiscent of the Aalen-Breslow-Nelson (ABN) type estimator.
When $s \rightarrow \infty$, $\hat{\Lambda}_{0q}(s,t|\alpha_q,\theta^R)$ converges to the nonparametric maximum likelihood estimator (NPMLE) $\hat{\Lambda}_{0q}(t|\alpha_q,\theta^R)$, given by
\begin{eqnarray}
\label{cumulative hazard}
\hat{\Lambda}_{0q}(t;\alpha_q,\theta^R) = \sum_{l: T_l \le t} \left[\frac{\sum_{i=1}^n \sum_{j=1}^K I\{\mathcal{E}_{iq}(S_j) = T_l\} dN_i^R(S_j;q)}{S_q^{0R}(T_l|\alpha_q,\theta^R)}\right] 
\end{eqnarray}
where
\begin{eqnarray}
S_q^{0R}(u|\alpha_q,\theta^R) & = & \sum_{i=1}^n\sum_{j=1}^{K+1} \left\{ \left[\frac{ \rho_q[N_i^R(\mathcal{E}_{iqj}^{-1}(u)-)|\alpha_q] \psi_i^R(\mathcal{E}_{iqj}^{-1}(u)|\theta^R)}{\mathcal{E}_{iq}^\prime[\mathcal{E}_{iqj}^{-1}(u)]}\right] \times \right.  \nonumber \\ && \left. 
I\{\mathcal{E}_{iq}(S_{j-1} \wedge \tau_i^*) < u \le \mathcal{E}_{iq}(S_{j} \wedge \tau_i^*)\} \right\}.   \label{sum of at-risk} 
\end{eqnarray}
Next, we obtain estimators of the $\eta(w_1,w_2)$s and $\xi(v_1,v)$s, again by first assuming that $\theta^W$ and $\theta^V$ are known.

\begin{proposition}
\label{prop-est of nu and xi}
If $(\theta^W,\theta^V)$ are known, the ML estimators of the $\eta(w_1,w_2)$s and $\xi(v_1,v)$s are the ``occurrence-exposure'' rates
\begin{eqnarray}
\hat{\eta}(w_1,w_2|\theta^W) & = &
\frac{\sum_{i=1}^n\int_0^\infty dN_i^W(s;w_1,w_2)}{\int_0^\infty S^{0W}(s;w_1|\theta^W) ds}, \forall (w_1,w_2) \in \IW; \label{eta estimator} \\
\hat{\xi}(v_1,v|\theta^V) & = & 
\frac{\sum_{i=1}^n\int_0^\infty dN_i^V(s;v_1,v)}{\int_0^\infty S^{0V}(s;v_1|\theta^V) ds}, \forall (v_1,v) \in \IV,   \label{xi estimator}
\end{eqnarray}
where
$$S^{0W}(s;w_1|\theta^W) = \sum_{i=1}^n Y_i(s) Y_i^W(s;w_1) \psi_i^W(s|\theta^W);$$
$$S^{0V}(s;v_1|\theta^V) = \sum_{i=1}^n Y_i(s) Y_i^V(s;v_1) \psi_i^V(s|\theta^V).$$
\end{proposition}

\begin{proof}
Follows immediately by maximizing the likelihood functions $\mathcal{L}^W$ and $\mathcal{L}^V$ with respect to the $\eta(w_1,w_2)$s and $\xi(v_1,v)$s, respectively.
\end{proof}

We form the profile likelihoods for $(\{\alpha_q, q \in \IQ\}, \theta^R, \theta^W, \theta^V)$. These are the likelihoods that are obtained after plugging-in the `estimators' $\hat{\Lambda}_{0q}(\cdot;\alpha_q,\theta^R)$s, $\hat{\eta}(w_1,w_2)$s, and $\hat{\xi}(v_1,v)$s in the full likelihoods. 

\begin{proposition}
\label{profile liks}
The three profile likelihood functions $\mathcal{L}_{pl}^R$, $\mathcal{L}_{pl}^W$ and $\mathcal{L}_{pl}^V$ are given by
\begin{eqnarray*}
\lefteqn{\mathcal{L}^R_{pl}(\alpha_q,q\in\IQ,\theta^R|\mathbf{D}) = }  \\ && \prod_{q\in\IQ} \prod_{i=1}^n \prodi_{s=0}^\infty \prodi_{t=0}^\infty
\left[\frac{\rho_q[N_i^R(s-)|\alpha_q] \psi_i^R(s|\theta^R)}{S_q^{0R}(t|\alpha_q,\theta^R)}\right]^{N_i^R(ds,dt;q)};
\end{eqnarray*}
\begin{displaymath}
\mathcal{L}_{pl}^W(\theta^W|\mathbf{D}) = \prod_{w_1 \in \mathfrak{W}} \prod_{i=1}^n \prodi_{s=0}^\infty \left[\frac{\psi_i^W(s|\theta^W)}{\int_0^\infty S^{0W}(t,w_1|\theta^W) dt }\right]^{dN_i^W(s;w_1,\bullet)};
\end{displaymath}
\begin{displaymath}
\mathcal{L}_{pl}^V(\theta^V|\mathbf{D}) = \prod_{v_1 \in \mathfrak{V}_1} \prod_{i=1}^n \prodi_{s=0}^\infty \left[\frac{\psi_i^V(s|\theta^V)}{\int_0^\infty S^{0V}(t,v_1|\theta^V) dt }\right]^{dN_i^V(s;v_1,\bullet)}.
\end{displaymath}
with $$dN_i^W(s;w_1,\bullet) = \sum_{w_2 \in \mathfrak{W};\ w_2 \ne w_1} dN_i^W(s;w_1,w_2),$$ the number of transitions from state $w_1$ at time $s$ for unit $i$, and $$dN_i^V(s;v_1,\bullet) = \sum_{v \in \mathfrak{V};\ v \ne v_1} dN_i^V(s;v_1,v),$$ the number of transitions from state $v_1$ at time $s$ for unit $i$.
\end{proposition}

\begin{proof}
These follow immediately by plugging-in the `estimators' into the three main likelihoods in (\ref{full lik factor}) and then simplifying.
\end{proof}

These three profile likelihoods are reminiscent of the partial likelihood function in Cox's proportional hazards model \cite{cox1972,AndGil1982}. From these profile likelihoods, we could obtain estimators of the parameters $\alpha_q$s, $\theta^R$, $\theta^W$, and $\theta^V$ as follows:
\begin{eqnarray*}
& (\hat{\alpha}_q, q\in\IQ,\hat{\theta}^R)  =  {\arg\max}_{(\alpha_q, \theta^R)} \mathcal{L}^R_{pl}(\alpha_q,q\in\IQ,\theta^R|\mathbf{D}); & \\
& \hat{\theta}^W  =  \arg\max_{\theta^W} \mathcal{L}_{pl}^W(\theta^W|\mathbf{D}) \quad \mbox{and} \quad
\hat{\theta}^V  =  \arg\max_{\theta^V}  \mathcal{L}_{pl}^V(\theta^V|\mathbf{D}). &
\end{eqnarray*}
Equivalently, these estimators are maximizers of the logarithm of the profile likelihoods. These log-profile likelihoods are more conveniently expressed in terms of stochastic integrals as follows:
\begin{eqnarray*}
l_{pl}^R & = & \sum_{q\in\IQ}\sum_{i=1}^n \int_0^\infty\int_0^\infty \left[\log\rho_q[N_i^R(s-)|\alpha_q] + \log\psi_i^R(s|\theta^R) - \right. \\  && \left. \log S_q^{0R}(t;\alpha_q,\theta^R)\right] N_i^R(ds,dt;q); \\
%\end{eqnarray*}
%%
%\begin{displaymath}
l_{pl}^W & = & \sum_{w_1 \in \mathfrak{W}} \sum_{i=1}^n \int_0^\infty \left[\log\psi_i^W(s|\theta^W) - \log \int_0^\infty S^{0W}(t;w_1|\theta^W) dt\right] \times \\ && N_i^W(ds;w_1,\bullet); \\
%\end{displaymath}
%%
%\begin{displaymath}
\end{eqnarray*}
\begin{eqnarray*}
l_{pl}^V & = & \sum_{v_1 \in \mathfrak{V}_1} \sum_{i=1}^n \int_0^\infty \left[\log\psi_i^V(s|\theta^V) - \log \int_0^\infty S^{0V}(t;v_1|\theta^V) dt\right] \times \\ && N_i^V(ds;v_1,\bullet).
\end{eqnarray*}
Gradient-based approaches could be used to numerically approximate the maximizers of these log-profile likelihoods. For instance, associated with each of these log-profile likelihood functions are their profile score vector (the gradient or vector of partial derivatives) and profile observed information matrix (negative of the matrix of second partial derivatives): $U^R(\alpha,\theta^R)$, $U^W(\theta^W)$, and $U^V(\theta^V)$ for the score vectors; and $I^R(\alpha,\theta^R)$, $I^W(\theta^W)$, and $I^V(\theta^V)$ for the observed information matrices. The estimators could then be obtained as the solutions of the set of equations
\begin{displaymath}
U_{pl}^R(\hat{\alpha},\hat{\theta}^R) = 0; \quad U_{pl}^W(\hat{\theta}^W) = 0; \quad U_{pl}^V(\hat{\theta}^V) = 0.
\end{displaymath}
A possible gradient-based computational approach to obtaining the estimates is via the Newton-Raphson iteration procedure:
\begin{eqnarray*}
&  (\hat{\alpha},\hat{\theta}^R)  \leftarrow  (\hat{\alpha},\hat{\theta}^R) + I^R(\hat{\alpha},\hat{\theta}^R)^{-1} U^R(\hat{\alpha},\hat{\theta}^R); & \\
& \hat{\theta}^W \leftarrow  \hat{\theta}^W + I^W(\hat{\theta}^W)^{-1} U^W(\hat{\theta}^W); \
\hat{\theta}^V \leftarrow  \hat{\theta}^V + I^V(\hat{\theta}^V)^{-1} U^V(\hat{\theta}^V). &
\end{eqnarray*}
Having obtained the estimates of $\alpha$, $\theta^W$ and $\theta^V$, we plug them into $\hat{\Lambda}_{0q}(t|\alpha_q,\theta^R)$s, $\eta(w_1,w_2|\theta^W)$s and $\xi(v_1,v|\theta^V)$s to obtain the estimators $\hat{\Lambda}_{0q}(t)$s, $\hat{\eta}(w_1,w_2)$s, and $\hat{\xi}(v_1,v)$s:
\begin{equation}
\label{final estimators}
\hat{\Lambda}_{0q}(t) = \hat{\Lambda}_{0q}(t|\hat{\alpha}_q,\hat{\theta}^R);\  \hat{\eta}(w_1,w_2) = \hat{\eta}(w_1,w_2|\hat{\theta}^W);\  \hat{\xi}(v_1,v) = \hat{\xi}(v_1,v|\hat{\theta}^V).
\end{equation}
The estimators of the baseline survivor functions, $\bar{F}_{0q}$s, in the RCR component are obtained using their product-integral representations in terms of the $\Lambda_{0q}$s via, for $q \in \IQ$,
\begin{equation}
\label{estimators baseline survivor functions}
\hat{\bar{F}}_{0q}(t) = \prodi_{w=0}^t \left[1 - \hat{\Lambda}_{0q}(dw)\right].
\end{equation}
These estimators of $\bar{F}_{0q}$s can be viewed as generalized product-limit or Kaplan-Meier estimators.

\section{Asymptotic Properties of Estimators}
\label{sec: Properties}

%%%%New Commands%%%%
\newcommand{\T}{\mathcal{T}}
\newcommand{\intT}{\int_{\mathcal{T}}}
\newcommand{\W}{\mathfrak{W}}
\newcommand{\V}{\mathfrak{V}}
\newcommand{\conprob}{\stackrel{p}{\rightarrow}}
\newcommand{\condist}{\stackrel{d}{\rightarrow}}
%%%%

In this section we present asymptotic properties of  estimators of the model parameters when $\Lambda_{0q}(\cdot)$s are nonparametrically-specified. We use a subscript or superscript of ``0'' to indicate the true value of the parameter. Thus, the true values of the model parameters are, for $q \in \IQ, (w_1,w_2) \in \IW, (v_1,v) \in \IV$,
$$(\Lambda_{0q}^0(\cdot), \bar{F}_{0q}^0(\cdot),\alpha_{0q}, \eta_0(w_1,w_2), \xi_0(v_1,v), \theta_0^W,\theta_0^V,\theta_0^R).$$

\subsection{Parameters for the LM and HS Components}

We first present for the parameters of the LM component of the model. We shall let $\mathcal{T} = [0,s^*]$, where $s^* < \infty$ is the maximal observation time for all the units. Thus, in what follows, $\int_{\mathcal{T}} \equiv \int_0^{s^*}$. Let
$$\hat{\eta}(w_1,w_2|\theta^W) = \frac{\sum_{i=1}^n \intT N_i^W(ds;w_1,w_2)}{\intT S^{0W}(s;w_1|\theta^W) ds}$$
where
$$S^{0W}(s;w_1,|\theta^W) = \sum_{i=1}^n Y_i(s) Y_i^W(s;w_1) \psi_i^W(s|\theta^W).$$
The log-profile likelihood function is
\begin{eqnarray*}
l_{pl}^W(\theta^W) & = & \sum_{w_1\in \W} \sum_{i=1}^n \intT \left[
\log \psi_i^W(s|\theta^W) - \log \intT S^{0W}(t;w_1|\theta^W) dt \right] \times \\ && N_i^W(ds;w_1,\bullet)
\end{eqnarray*}
and the estimator of $\theta^W$ based on data observed over $\mathcal{T}$ is
$$\hat{\theta}^W = \arg\max_{\theta^W} l_{pl}^W(\theta^W).$$
The estimator of $\eta(w_1,w_2)$ for $(w_1,w_2) \in \IW$ is then
$$\hat{\eta}(w_1,w_2) = \hat{\eta}(w_1,w_2|\hat{\theta}^W); \quad \hat{\eta}(w_1,w_1) = -\sum_{w_2 \ne w_1} \hat{\eta}(w_1,w_2).$$
Recall that $\psi_i^W(s|\theta^W) = \exp\{B_i^W(s-) \theta^W\}$. Then,
\begin{eqnarray*}
 S^{1W}(s;w_1|\theta^W)  & \equiv &  \nabla_{\theta^W} S^{0W}(s;w_1|\theta^W) \\ & = & \sum_{i=1}^n B_i^W(s-) Y_i(s) Y_i^W(s;w_1) \psi_i^W(s|\theta^W);   \\
S^{2W}(s;w_1|\theta^W)  & \equiv & \nabla_{(\theta^W)^{\otimes 2}} S^{0W}(s;w_1|\theta^W) \\ & = & \sum_{i=1}^n [B_i^W(s-)]^{\otimes 2} Y_i(s) Y_i^W(s;w_1) \psi_i^W(s|\theta^W). 
\end{eqnarray*}
Let
%
%\begin{eqnarray*}
$$ E^W(w_1|\theta^W) = \frac{\intT S^{1W}(t;w_1|\theta^W) dt}{\intT S^{0W}(t;w_1|\theta^W) dt}; $$
$$ V^W(w_1|\theta^W) = \frac{\intT S^{2W}(t;w_1|\theta^W) dt}{\intT S^{0W}(t;w_1|\theta^W) dt} - \left[E^W(w_1|\theta^W)\right]^{\otimes 2}. $$
%\end{eqnarray*}
%
The profile score and observed information for $\theta^W$ are
\begin{eqnarray*}
U_{pl}^W(\theta^W) & \equiv & \nabla_{\theta^W} l_{pl}^W(\theta^W) \\ & = &
\sum_{w_1 \in \W} \sum_{i=1}^n \intT \left[B_i^W(s-) - E^W(w_1|\theta^W)\right] N_i^W(ds;w_1,\bullet); \\
 I_{pl}^W(\theta^W) & \equiv & -\nabla_{(\theta^W)^{\otimes 2}} l_{pl}^W(\theta^W) =
\sum_{w_1 \in \W} V^W(w_1|\theta^W) N_\bullet^W(s^*;w_1,\bullet).
\end{eqnarray*}

We now state needed regularity conditions for the LM component needed for the asymptotic results.

\medskip

\begin{center}
\underline{LM-Component Regularity Conditions}
\end{center}

\medskip

There exists a neighborhood $\Theta_0^W \subset \Theta^W$ of $\theta_0^W$, and a scalar function, a $\dim(\theta^W) \times 1$ vector function, and a $\dim(\theta^W) \times \dim(\theta^W)$ matrix function defined on $\T \times \W \times \Theta^W$ given, respectively, by
$$s^{0W}(v;w_1|\theta^W),\quad s^{1W}(v;w_1|\theta^W), \quad \mbox{and} \quad s^{2W}(v;w_1|\theta^W),$$
which depend possibly on all the true values of the parameters, and satisfying the following conditions:
\begin{itemize}
\item[(a)]
For $k = 0, 1, 2$, and $\forall w_1 \in \W$, as $n \rightarrow \infty$,
$$\sup_{\theta^W \in \Theta_0^W; v \in \T} \left|\left|\frac{1}{n} S^{kW}(v;w_1|\theta^W) - s^{kW}(v;w_1|\theta^W)\right|\right| \conprob 0;$$

\item[(b)]
For $k = 0, 1, 2$ and $\forall w_1 \in \W$, $s^{kW}(v;w_1|\theta^W)$ is continuous in $\theta^W$, uniformly in $v$, and bounded on $\T \times \Theta_0^W$;

\item[(c)]
$\forall w_1 \in \W$, $\intT s^{0W}(v;w_1|\theta_0^W) dv > 0;$

\item[(d)]
$\forall w_1 \in \W$, $s^{0W}(v;w_1|\theta^W)$ is twice-differentiable in $\theta^W$ for $\forall v \in \T$, and
$$s^{1W}(v;w_1|\theta^W) = \nabla_{\theta^W} s^{0W}(v;w_1|\theta^W);$$
$$s^{2W}(v;w_1|\theta^W) = \nabla_{(\theta^W)^{\otimes 2}} s^{0W}(v;w_1|\theta^W);$$

\item[(e)]
$\forall w_1 \in \W$, with
%
%\begin{eqnarray*}
$$ e^W(w_1|\theta^W) = \frac{\intT s^{1W}(v;w_1|\theta^W) dv}{\intT s^{0W}(v;w_1|\theta^W) dv}; $$
$$ v^W(w_1|\theta^W) =  \frac{\intT s^{2W}(v;w_1|\theta^W) dv}{\intT s^{0W}(v;w_1|\theta^W) dv} - [e^W(w_1|\theta^W)]^{\otimes 2}, $$
%\end{eqnarray*}
%
and $\eta_0(w_1,\bullet) = \sum_{w_2 \in \W; w_2 \ne w_1} \eta_0(w_1,w_2)$, the matrix
$$\Sigma^W = \sum_{w_1 \in \W} \eta_0(w_1,\bullet) v^W(w_1|\theta_0^W) \intT s^{0W}(v;w_1|\theta^0) dv$$
is positive definite.
\end{itemize}

Before stating the theorem containing the asymptotic properties of estimators of parameters in the LM component of the model, we first introduce a more compact notation. Suppose that $\{h(w_1,w_2): (w_1,w_2) \in \IW\}$ is a collection of functions, possibly vector-valued. Then, when we write $h$, this will represent the vector or matrix with the $|\IW|$ rows consisting of the $h(w_1,w_2)\trp$ corresponding to the order of the $(w_1,w_2)$ in $\IW$. Thus, for example, $\eta$ is the $|\IW|$ vector with elements $\eta(w_1,w_2)$; while $s^{0W}(\cdot|\theta^W)$ is the $|\IW|$ vector with block elements $1_{|\W|-1} \otimes  s^{0W}(\cdot;w_1|\theta_0^W)$ for $w_1 \in \W$. Similarly, $s^{1W}(\cdot|\theta^W)$ is the $|\IW| \times \dim(\theta^W)$ matrix with block elements $[s^{1W}(\cdot;w_1|\theta_0^W)]\trp \otimes 1_{|\W|-1}$ for $w_1 \in \W$. Using this notation, let
\begin{eqnarray*}
\Xi^W & = & \left[\DG(\eta_0)\right] \left[\DG\left(\intT s^{0W}(s|\theta_0^W) ds\right)\right]^{-1} \\ & = & \DG\left(
\frac{\eta_0(w_1,w_2)}{\intT s^{0W}(s;w_1|\theta_0^W) ds}: (w_1,w_2) \in \IW\right),
\end{eqnarray*}
a $(|\IW| \times |\IW|)$-dimensional matrix; and
$$e_0^W = (1_{|\W|-1} \otimes [e_0(w_1|\theta_0^W)]\trp: w_1 \in \W),$$ 
a $(|\IW| \times \dim(\theta^W))$-dimensional matrix. We also mention prior to presenting the theorems and proofs that the development of the asymptotic properties of the estimators mimics those in \cite{ander1993} and \cite{borg1984} for the LM and HS components, and those in \cite{pena2016} for the RCR component.

\begin{theorem}
\label{theo-asymptotics for LM}
Under the ``LM-Component Regularity Conditions", as $n \rightarrow \infty$,
\begin{itemize}
\item[(i)] (Consistency)
$\hat{\theta}^W \conprob \theta_0^W$ and $\forall (w_1,w_2) \in \IW: \hat{\eta}(w_1,w_2) \conprob \eta_0(w_1,w_2)$;

\item[(ii)] (Joint Asymptotic Normality)
$$\sqrt{n} \left[
\begin{array}{c}
\hat{\theta}^W - \theta_0^W \\
\hat{\eta} - \eta_0
\end{array}
\right]
\condist
N\left(0,\Gamma^W\right)$$
with
\begin{eqnarray*}
\Gamma^W = \left[
\begin{array}{cc}
\left(\Sigma^W\right)^{-1} &
-\left(\Sigma^W\right)^{-1} e_0^W \DG(\eta_0) \\
-\DG(\eta_0) (e_0^W)\trp \left(\Sigma^W\right)^{-1} &
\Xi^W + \DG(\eta_0) (e_0^W)\trp \left(\Sigma^W\right)^{-1} e_0^W \DG(\eta_0)
\end{array}
\right].
\end{eqnarray*}
%%
%where
%%
%$$\Xi^W = \DG\left(
%\frac{\eta_0(w_1,w_2)}{\intT s^{0W}(s;w_1|\theta_0^W) ds}: (w_1,w_2) \in \IW\right);$$
%$$\eta_0 = \left(\eta_0(w_1,w_2): (w_1,w_2) \in \IW\right);$$
%$$e_0^W = \left(e^W(w_1|\theta_0^W): (w_1,w_2) \in \IW\right).$$
%%
%[{\bf Note:} $\Xi^W$ is a square matrix of dimension $|\IW|$; $\eta_0$ is a $|\IW| \times 1$ vector; and $e_0^W$ is a $|\IW| \times p$ vector, with $p$ the dimension of $\theta^W$; and $\Sigma^W$ is a square matrix of dimension $p$, so $\Gamma^W$ is a square matrix of dimension $(p+|\IW|)$.]
\end{itemize}
\end{theorem}

\begin{proof}
For $s \in \T$, let
\begin{eqnarray*}
\lefteqn{C^W(s|\theta^W) = } \\ && \sum_{w_1 \in \W} \sum_{i=1}^n \int_0^s \left[
\log \psi_i^W(v|\theta^W) - \log \int_0^s S^{0W}(t;w_1|\theta^W) dt \right] N_i^W(dv; w_1,\bullet).
\end{eqnarray*}
Define the process $\{G^W(s|\theta^W): s \in \T\}$, using the fact that $\psi_i^W(s|\theta^W) = \exp\{B_i(s-)\theta^W\}$, via
\begin{eqnarray*}
\lefteqn{G^W(s|\theta^W)  =  \frac{1}{n} \left[ C^W(s|\theta^W) - C^W(s|\theta_0^W) \right] } \\
& = & \frac{1}{n} \sum_{w_1 \in \W} \sum_{i=1}^n \int_0^s
\left[
B_i^W(v-) (\theta^W - \theta_0^W) - \log\frac{\int_0^s S^{0W}(t;w_1|\theta^W) dt}{\int_0^s S^{0W}(t;w_1|\theta_0^W) dt}
\right] \times  \\
&& N_i^W(dv;w_1,\bullet).
\end{eqnarray*}
Let the process $\{\tilde{G}^W(s|\theta^W): s \in \T\}$ be such that
\begin{eqnarray*}
\lefteqn{\tilde{G}^W(s|\theta^W)  = } \\ &&  \frac{1}{n} \sum_{w_1 \in \W} \sum_{i=1}^n \int_0^s
\left[
B_i^W(v-) (\theta^W - \theta_0^W) - \log\frac{\int_0^s S^{0W}(t;w_1|\theta^W) dt}{\int_0^s S^{0W}(t;w_1|\theta_0^W) dt}
\right] \times \\
&& \eta_0(w_1,\bullet) Y_i(v) Y_i^W(v;w_1) \psi_i^W(v|\theta_0^W) dv.
\end{eqnarray*}
Then the process $\{G^W(s|\theta^W) - \tilde{G}^W(s|\theta_0^W): s \in \T\}$ is a zero-mean square-integrable martingale with predictable quadratic variation process  $\{\langle G^W(\cdot|\theta^W) - \tilde{G}^W(\cdot|\theta_0^W) \rangle(s): s \in \T\}$ with
\begin{eqnarray*}
\lefteqn{ \langle G^W(\cdot|\theta^W) - \tilde{G}^W(\cdot|\theta_0^W) \rangle(s)  =  
\frac{1}{n^2} \sum_{w_1 \in \W} \eta_0(w_1,\bullet) \times } \\ 
&& \sum_{i=1}^n \left[
B_i^W(v-) (\theta^W - \theta_0^W) - \log\frac{\int_0^s S^{0W}(t;w_1|\theta^W) dt}{\int_0^s S^{0W}(t;w_1|\theta_0^W) dt}
\right]^2 \times \\ && 
Y_i(v) Y_i^W(v;w_1) \psi_i^W(v|\theta_0^W) dv.
\end{eqnarray*}
Using the regularity conditions, $\langle G^W(\cdot|\theta^W) - \tilde{G}^W(\cdot|\theta_0^W) \rangle(s^*) = o_p(1)$ on $\Theta_0^W$, while $\tilde{G}^W(s|\theta^W)$ converges in probability to some function $g^W(s|\theta^W)$ on $\T \times \Theta_0^W$. It follows, using Lenglart's Inequality that $G^W(s|\theta^W)$ converges in probability to $g^W(s|\theta^W)$ on $\T \times \Theta_0^W$, with the limiting function being
\begin{eqnarray*}
g^W(s|\theta^W) & = & \sum_{w_1 \in \W} \eta_0(w_1,\bullet) 
\left\{\left(\int_0^s s^{1W}(v|\theta_0^W) dv\right) (\theta^W - \theta_0^W) - \right. \\ && \left. \left( \log
\frac{\int_0^s s^{0W}(t;w_1|\theta^W) dt}{\int_0^s s^{0W}(t;w_1|\theta_0^W) dt} \right) \int_0^s s^{0W}(v;w_1|\theta_0^W) dv\right\}.
\end{eqnarray*}
Observe that $g^W(s|\theta_0^W) = 0$ for $s \in \T$. Also,
\begin{eqnarray*}
\nabla_{\theta^W} g^W(s|\theta^W) & = & \sum_{w_1\in\W} \eta_0(w_1,\bullet) \left\{
\int_0^s s^{1W}(v|\theta_0^W) - \right. \\ && \left. \left(\frac{\int_0^s s^{1W}(t;w_1|\theta^W) dt}{\int_0^s s^{0W}(t;w_1|\theta^W) dt}\right)
\left(\int_0^s s^{0W}(v;w_1|\theta_0^W) dv\right)\right\}
\end{eqnarray*}
so that $\nabla_{\theta^W} g^W(s^*|\theta^W)|_{\theta^W = \theta_0^W} = 0$.
We also have
\begin{eqnarray*}
\nabla_{(\theta^W)^{\otimes 2}} g^W(s|\theta^W) & = & - \sum_{w_1\in\W} \eta_0(w_1,\bullet) \left\{
\frac{\int_0^s s^{2W}(t;w_1|\theta^W) dt}{\int_0^s s^{0W}(t;w_1|\theta^W) dt} - \right. \\
&& \left. \left( \frac{\int_0^s s^{2W}(t;w_1|\theta^W) dt}{\int_0^s s^{0W}(t;w_1|\theta^W) dt} \right)^{\otimes 2} \right\}
\int_0^s s^{0W}(v;w_1|\theta_0^W) dv,
\end{eqnarray*}
so that
\begin{eqnarray*}
\lefteqn{ \nabla_{(\theta^W)^{\otimes 2}} g^W(s^*|\theta^W)|_{\theta^W = \theta_0^W}  =  } \\ &&
 -\sum_{w_1 \in \W} \eta_0(w_1,\bullet) v^W(w_1|\theta_0^W)
\int_0^{s^*} s^{0W}(v;w_1|\theta_0^W) dv
\end{eqnarray*}
which is negative definite. As such, $\theta_0^w$ is a maximizer of $g^W(s^*|\theta^W)$. Since $\hat{\theta}^W$ is a maximizer of $G^W(s^*|\theta^W)$, and $G^W(s^*|\theta^W) \conprob g^W(s^*|\theta^W)$ on $\Theta_0^W$, then we can conclude that $\hat{\theta}^W \conprob \theta_0^W$.

Next, observe by Taylor expansion and using the just-established result that $\hat{\theta}^W \conprob \theta_0^W$, that
\begin{eqnarray*}
\left(\frac{1}{n} \intT S^{0W}(v;w_1|\hat{\theta}^W) dv\right)^{-1} = \left(\frac{1}{n} \intT S^{0W}(v;w_1|{\theta}_0^W) dv\right)^{-1}  + o_p(1).
\end{eqnarray*}
Therefore, $\forall (w_1,w_2) \in \IW$,
\begin{eqnarray*}
\lefteqn{\hat{\eta}(w_1,w_2)  =  \hat{\eta}(w_1,w_2|\hat{\theta}^W) = \frac{\sum_{i=1}^n \intT N_i^W(dv;w_1,w_2)}{\intT S^{0W}(v;w_1|\hat{\theta}^W) dv} } \\
& = & \left\{\frac{1}{n} \sum_{i=1}^n \intT M_i^W(dv;w_1,w_2) + \eta_0(w_1,w_2) \frac{1}{n} \intT S^{0W}(v;w_1|{\theta}_0^W) dv\right\} \times \\
& & \left\{\left(\frac{1}{n} \intT S^{0W}(v;w_1|{\theta}_0^W) dv\right)^{-1}  + o_p(1)\right\} \\
& = & \eta_0(w_1,w_2) +  \left(\frac{1}{n} \sum_{i=1}^n M_i^W(dv;w_1,w_2) \right) \left(\frac{1}{n} \intT S^{0W}(v;w_1|{\theta}_0^W) dv\right)^{-1}  \\ && + o_p(1)
\end{eqnarray*}
Since $\frac{1}{n} \sum_{i=1}^n \intT M_i^W(dv;w_1,w_2) = o_p(1)$, then $\hat{\eta}(w_1,w_2) = \eta_0(w_1,w_2) + o_p(1)$, proving the consistency of $\hat{\eta}(w_1,w_2)$ for $\eta_0(w_1,w_2)$.

Next, we prove the asymptotic normality result. Recall that $\hat{\theta}^W$ satisfies the condition that
$$0 = U_{pl}^W(\hat{\theta}^W) = \sum_{i=1}^n \sum_{(w_1,w_2)\in\IW} \intT [B_i^W(s-) - E^W(w_1|\hat{\theta}^W)] N_i^W(ds;w_1,w_2).$$
Taylor expansion yields
$$E^W(w_1|\hat{\theta}^W) = E^W(w_1|\theta_0^W) + V^W(w_1|\theta_*^W) (\hat{\theta}^W - \theta_0^W)$$
where $\theta_*^W \in [\hat{\theta}^W,\theta_0^W]$. Letting
$$H_i^W(s;w_1) = B_i^W(s-) - E^W(w_1|\theta_0^W), i=1,\ldots,n; w_1\in\W,$$
we then have the representation
\begin{eqnarray*}
\lefteqn{\sqrt{n}[\hat{\theta}^W - \theta_0^W]  = 
\left[\frac{1}{n} \sum_{w_1\in\W}  V^W(w_1|\theta_*^W) N_\bullet^W(s^*;w_1,\bullet)\right]^{-1} \times} \\ &&
\left\{\frac{1}{\sqrt{n}} \sum_{i=1}^n \sum_{(w_1,w_2) \in \IW} \intT H_i^W(s;w_1) N_i^W(ds;w_1,w_2)\right\} \\
& = & \left[\Sigma^W\right]^{-1} 
\left\{\frac{1}{\sqrt{n}} \sum_{i=1}^n \sum_{(w_1,w_2) \in \IW} \intT H_i^W(s;w_1) M_i^W(ds;w_1,w_2)\right\}  + o_p(1)
\end{eqnarray*}
since 
$$\sum_{i=1}^n \intT H_i^W(s) \eta_0(w_1,w_2) Y_i(s) Y_i^W(s) \psi_i^W(s|\theta_0^W) ds = 0;$$
$$\frac{1}{n} \sum_{w_1\in\W}  V^W(w_1|\theta_*^W) N_\bullet^W(s^*;w_1,\bullet) = \Sigma^W + o_p(1),$$
using the fact that $\theta_*^W \conprob \theta_0^W$. Also, using Taylor expansion, we have $\forall (w_1,w_2) \in \IW$, that
\begin{eqnarray*}
\lefteqn{\sqrt{n}[\hat{\eta}(w_1,w_2) - \eta_0(w_1,w_2)]  =  -\eta_0(w_1,w_2) [E^W(w_1|\theta_0^W)]\trp \sqrt{n}[\hat{\theta}^W - \theta_0^W] + } \\ &&
\left[\frac{1}{n}\intT S^{0W}(s;w_1|\theta_0^W) ds\right]^{-1} \left[\frac{1}{\sqrt{n}} \sum_{i=1}^n M_i^W(ds;w_1,w_2)\right] + o_p(1) \\
& = &  -\eta_0(w_1,w_2) [e^W(w_1|\theta_0^W)]\trp \sqrt{n}[\hat{\theta}^W - \theta_0^W] +  \\ &&
\left[\intT s^{0W}(s;w_1|\theta_0^W) ds\right]^{-1} \left[\frac{1}{\sqrt{n}} \sum_{i=1}^n M_i^W(ds;w_1,w_2)\right] + o_p(1)
\end{eqnarray*}
Define 
\begin{eqnarray*}
& Z_1^W  =  \sqrt{n} [\hat{\theta}^W - \theta_0^W]; & \\
& Z_2^W  =  \sqrt{n} [\hat{\eta} - \eta_0] + \DG(\eta_0) (e_0^W)\trp \sqrt{n} [\hat{\theta}^W - \theta_0^W]. &
\end{eqnarray*}
Note that $Z_1$ is a $\dim(\theta^W)$-dimensional vector, whereas $Z_2^W$ is a $|\IW|$-dimensional vector. From the results above, and if we let
$$H_i^W(s) = [B_i^W(s-) - E^W(w_1|\theta_0^W), (w_1,w_2) \in \IW]$$
be the $(\dim(\theta^W) \times |\IW|)$-dimensional matrix, and $I_{|\IW|}$ be the $(|\IW| \times |\IW|)$-dimensional identity matrix, and since $M_i^W(s)$ is the $(|\IW| \times 1)$-dimensional vector with elements $M_i^W(s;w_1,w_2)$, we therefore have the asymptotic representation
\begin{eqnarray*}
\left[
\begin{array}{c}
Z_1^W \\
Z_2^W
\end{array}
\right] & = &
\left[
\begin{array}{cc}
\left(\Sigma^W\right)^{-1} & 0 \\
0 & \left(\DG(\intT s^{0W}(s|\theta_0^W) ds)\right)^{-1}
\end{array}
\right]  \times \\ &&
\left[ 
\frac{1}{\sqrt{n}} \sum_{i=1}^n \intT
\left[ \begin{array}{c}
H_i^W(s) \\
I_{|\IW|}
\end{array}
\right] 
M_i(ds)\right] 
+ o_p(1).
\end{eqnarray*}
By the Martingale Central Limit Theorem,
\begin{eqnarray*}
\frac{1}{\sqrt{n}} \sum_{i=1}^n \intT
\left[ \begin{array}{c}
H_i^W(s) \\
I_{|\IW|}
\end{array}
\right] 
M_i(ds) \condist N(0,\Upsilon^W)
\end{eqnarray*}
where
\begin{eqnarray*}
\Upsilon^W & = & \plim \frac{1}{n} \sum_{i=1}^n \intT 
\left[ \begin{array}{c}
H_i^W(s) \\
I_{|\IW|}
\end{array}
\right] 
\DG(\eta_0^W) \DG(Y_i^W(s))
\left[ \begin{array}{c}
H_i^W(s) \\
I_{|\IW|}
\end{array}
\right]\trp \times \\ &&
Y_i(s) \psi_i^W(s|\theta_0^W) ds 
 =  \left[
\begin{array}{cc}
\Sigma^W & 0 \\
0 & Dg(\eta_0^W)
\end{array}
\right]
\end{eqnarray*}
where the final equality is obtained by using the facts that
\begin{eqnarray*}
\lefteqn{\sum_{(w_1,w_2) \in \IW} \eta_0(w_1,w_2) \frac{1}{n} \sum_{i=1}^n [B_i^W(s-) - E^W(w_1|\theta_0^W)]^{\otimes 2} \times} 
\\ && Y_i(s) Y_i^W(s;w_1) \psi_i^W(s|\theta_0^W) \conprob \Sigma^W; \\
&& \sum_{i=1}^n [B_i^W(s-) - E^W(w_1|\theta_0^W)] Y_i(s) Y_i^W(s;w_1) \psi_i^W(s|\theta_0^W) = 0.
\end{eqnarray*}
It follows that 
\begin{eqnarray*}
\left[
\begin{array}{c}
Z_1^W \\
Z_2^W
\end{array}
\right] \condist N\left(
\begin{array}{c}
0 \\ 0
\end{array},
\left[
\begin{array}{cc}
(\Sigma^W)^{-1} & 0  \\
0 & \Xi^W
\end{array}
\right]
\right).
\end{eqnarray*}
where $$\Xi^W = \DG(\eta_0) \left[\DG\left(\intT s^{0W}(s|\theta_0^W) ds\right)\right]^{-1}.$$
Finally, since
\begin{eqnarray*}
\sqrt{n} \left[
\begin{array}{c}
\hat{\theta}^W - \theta_0^W \\
\hat{\eta} - \eta_0
\end{array}
\right] =  \left[
\begin{array}{cc}
I & 0 \\
-\DG(\eta_0) e_0^W & I
\end{array}
\right]
\left[
\begin{array}{c}
Z_1^W \\
Z_2^W
\end{array}
\right]
\end{eqnarray*}
then result (ii) of the theorem follows since
\begin{eqnarray*}
\lefteqn{\Gamma^W  = 
\left[
\begin{array}{cc}
I & 0 \\
-\DG(\eta_0) e_0^W & I
\end{array}
\right]
\left[
\begin{array}{cc}
(\Sigma^W)^{-1} & 0  \\
0 & \Xi^W
\end{array}
\right]
\left[
\begin{array}{cc}
I & 0 \\
-\DG(\eta_0) e_0^W & I
\end{array}
\right]\trp } \\
& = & 
\left[
\begin{array}{cc}
\left(\Sigma^W\right)^{-1} &
-\left(\Sigma^W\right)^{-1} e_0^W \DG(\eta_0) \\
-\DG(\eta_0) (e_0^W)\trp \left(\Sigma^W\right)^{-1} &
\Xi^W + \DG(\eta_0) (e_0^W)\trp \left(\Sigma^W\right)^{-1} e_0^W \DG(\eta_0)
\end{array}
\right].
\end{eqnarray*}
This completes the proof of the theorem.
\end{proof}

Some remarks are in order. On the above, we did not explicitly dealt with the estimation of $\eta(w_1,w_1)$ for $w_1 \in \W$. Since $$\eta(w_1,w_1) = -\sum_{w_2 \in W; w_2 \ne w_1} \eta(w_1,w_2),$$ then a natural estimator of $\eta(w_1,w_1)$ is $$\hat{\eta}(w_1,w_1) = - \sum_{w_2 \in W; w_2 \ne w_1} \hat{\eta}(w_1,w_2).$$ Properties of this estimator will then follow straight-forwardly from the results in Theorem \ref{theo-asymptotics for LM}. In terms of consistently estimating $\Gamma$, the asymptotic covariance matrix, one option is to simply replace the limiting terms by their empirical counterparts.

Analogous results and regularity conditions hold for the parameters in the HS component: for the $\theta^V$ and the $\{\xi(v_1,v): (v_1,v) \in \IV\}$ parameters. We do not anymore present explicitly these regularity conditions and asymptotic results.

\subsection{Parameters for the RCR Component}

To simplify notation, we let the finite-dimensional parameter vector for the RCR component be denoted by
$\phi = \{(\alpha_q, q\in\IQ), \theta^R\}$
taking values in the sub-parameter space $\Phi$. The true value of the parameter will be denoted by $\phi_0 = \{(\alpha_{0q}, q\in\IQ), \theta_0^R\}$. First- and second-order differentiation operators with respect to $\phi$ will be 
$$\nabla_\phi g = \frac{\partial}{\partial\phi} g = \stackrel{\bullet}{g} \ \mbox{and} \ 
\nabla_{\phi^{\otimes 2}} g = \frac{\partial^2}{\partial\phi^{\otimes 2}} g = \frac{\partial^2}{\partial\phi\partial\phi\trp} g = \stackrel{\bullet\bullet}{g}.$$ 
The true function of the infinite-dimensional parameter $\Lambda_{0q}(\cdot)$ will be $\Lambda_{0q}^0(\cdot)$ for $q\in\IQ$. Define, for $q\in\IQ$ and $i=1,2,\ldots,n$,
$$\kappa_{iq}(s|\phi) = \rho_q[N_i^R(s-)|\alpha_q] \psi_i^R(s|\theta^R).$$
Recall that $0 \equiv S_0 < S_1 < S_2 < \ldots < S_K < S_{K+1} \equiv \infty$ are the ordered distinct calendar times in which an event of any type (e.g., in RCR, LM, or HS components) occurs. Note that the random variable $K$ is informative about all the model parameters. In addition, recall that, for $j=1,2,\ldots,K+1$, $\mathcal{E}_{iqj}(s)$ is the {\em restriction} of $\mathcal{E}_{iq}(s)$ on $s \in (S_{j-1}, S_j]$, so that, for $t \in (\mathcal{E}_{iq}(S_{j-1}),\mathcal{E}_{iq}(S_j)]$, then $\mathcal{E}_{iqj}^{-1}(t) \in (S_{j-1}, S_j]$. In (\ref{defn Sq0R}), we have also defined, for $q\in\IQ$, but now in terms of $\phi$,
\begin{eqnarray*}
\lefteqn{S_q^{0R}(s^*,w|\phi) = } \\
&& \sum_{i=1}^n \sum_{j=1}^{K+1} I\{w \in (\mathcal{E}_{iq}(S_{j-1} \wedge \tau_i \wedge s^*), \mathcal{E}_{iq}(S_{j} \wedge \tau_i \wedge s^*)]\} \frac{\kappa_{iq}[\mathcal{E}_{iqj}^{-1}(w)|\phi]}{\mathcal{E}_{iq}[\mathcal{E}_{iqj}^{-1}(w)]}.
\end{eqnarray*}
We now define empirical probability functions, for $q\in\IQ$, given $(s^*,w,\phi)$, defined over the random set 
$$\mathfrak{K} = \{(i,j): i=1,2,\ldots,n; j=1,2,\ldots,K+1\}$$
such that for each $(i,j) \in \mathfrak{K}$,
\begin{equation}
\label{q empirical probability}
\mathbb{P}_q(s^*,w|\phi)(i,j) = \frac{I\{w \in (\mathcal{E}_{iq}(S_{j-1} \wedge \tau_i \wedge s^*), \mathcal{E}_{iq}(S_{j} \wedge \tau_i \wedge s^*)]\} \frac{\kappa_{iq}[\mathcal{E}_{iqj}^{-1}(w)|\phi]}{\mathcal{E}_{iq}[\mathcal{E}_{iqj}^{-1}(w)]}}{S_q^{0R}(s^*,w|\phi)}.
\end{equation}
For $(I,J)$ taking values in $\mathfrak{K}$, we may then write $(I,J) \sim \mathbb{P}_q(s^*,w|\phi)$ to mean that $\Pr\{(I,J) = (i,j)|s^*,w,\phi\} = \mathbb{P}_q(s^*,w|\phi)(i,j)$. For a vector function $h$ defined over $\mathfrak{K}$, we then have the expectation and covariance operators with respect to $\mathbb{P}(s^*,w|\phi)$ defined, respectively, via
\begin{eqnarray*}
E_{\mathbb{P}_q(s^*,w|\phi)} [h(I,J)] & = & \sum_{i=1}^n\sum_{j=1}^{K+1} h(i,j) \mathbb{P}_q(s^*,w|\phi)(i,j); \\
V_{\mathbb{P}_q(s^*,w|\phi)} [h(I,J)] & = & E_{\mathbb{P}_q(s^*,w|\phi)} [h(I,J) - E_{\mathbb{P}_q(s^*,w|\phi)} h(I,J)]^{\otimes 2} \\
& = &
\left[E_{\mathbb{P}_q(s^*,w|\phi)} [h(I,J)]^{\otimes 2}\right] - \left[E_{\mathbb{P}_q(s^*,w|\phi)} h(I,J)\right]^{\otimes 2}.
\end{eqnarray*}

The restricted profile likelihood for the parameter $\phi$ in the RCR component of the model, restricted in the sense that we only consider events that happened when the effective ages are no more than $t^*$ and which have happened over $[0,s^*]$, is
\begin{eqnarray*}
\label{restricted profile RCR}
\mathcal{L}_{pl}^R(\phi) & \equiv & \mathcal{L}_{pl}^R(\phi|s^*,t^*) \equiv \mathcal{L}_{pl}^R(\phi|\mathbf{D}(s^*,t^*)) \\
& = & 
\prod_{q=1}^Q \prod_{i=1}^n \prodi_{s=0}^{s^*} 
\left[
\frac{\kappa_{iq}(s|\phi)}{S_q^{0R}(s^*,\mathcal{E}_{iq}(s)|\phi)}
\right]^{I\{\mathcal{E}_{iq}(s) \le t^*\} dN_i^R(s;q)}.
\end{eqnarray*}
Note that as we let $t^* \rightarrow \infty$ and $s^* \rightarrow \infty$, we obtain the profile likelihood in (\ref{profile liks}). The restricted log-profile likelihood function is then
\begin{eqnarray*}
l_{pl}^R(\phi|s^*,t^*) & = & \sum_{q=1}^Q \sum_{i=1}^n \int_0^{s^*} \left\{\log \kappa_{iq}(s|\phi) - \log S_q^{0R}(s^*,\mathcal{E}_{iq}(s)|\phi)\right\} \times \\
&&  I\{\mathcal{E}_{iq}(s) \le t^*\} N_i^R(ds;q).
\end{eqnarray*}
Under the condition that we are able to interchange the order of differentiation (with respect to $\phi$) and integration (with respect to $s$), then we obtain the restricted profile score vector function and the observed profile information matrix function given, respectively, by
\begin{eqnarray*}
\lefteqn{ U_{pl}^R(\phi|s^*,t^*)  =  \nabla_\phi l_{pl}^R(\phi|s^*,t^*) } \\
 & = & \sum_{q=1}^Q \sum_{i=1}^n \int_0^{s^*} 
\left\{
\frac{\stackrel{\bullet}{\kappa}_{iq}}{\kappa_{iq}}(s|\phi) - \frac{\stackrel{\bullet}{S}_q^{0R}}{S_q^{0R}}(s^*,\mathcal{E}_{iq}(s)|\phi)\right\}  \times \\
&& 
I\{\mathcal{E}_{iq}(s) \le t^*\} N_i^R(ds;q);
\end{eqnarray*}
and
\begin{eqnarray*}
\lefteqn{ I_{pl}^R(\phi|s^*,t^*)  =  -\nabla_{\phi^{\otimes 2}} l_{pl}^R(\phi|s^*,t^*) } \\
 & = & -\sum_{q=1}^Q \sum_{i=1}^n \int_0^{s^*} 
 \left[
\left\{
\frac{\stackrel{\bullet\bullet}{\kappa}_{iq}}{\kappa_{iq}}(s|\phi) - \left(\frac{\stackrel{\bullet}{\kappa}_{iq}}{\kappa_{iq}}(s|\phi)\right)^{\otimes 2}\right\}  - \right. \\ &&
\left.
\left\{\frac{\stackrel{\bullet\bullet}{S}_q^{0R}}{S_q^{0R}}(s^*,\mathcal{E}_{iq}(s)|\phi) - \left(\frac{\stackrel{\bullet}{S}_q^{0R}}{S_q^{0R}}(s^*,\mathcal{E}_{iq}(s)|\phi)\right)^{\otimes 2}\right\}
\right]
 \times \\
&& 
I\{\mathcal{E}_{iq}(s) \le t^*\} N_i^R(ds;q).
\end{eqnarray*}
A remark on notation: for a function $g$, $\frac{\stackrel{\bullet}{g}(s)}{g(s)} = \frac{\stackrel{\bullet}{g}}{g}(s)$. The estimator of $\phi$, denoted by $\hat{\phi} = \hat{\phi}(s^*,t^*)$, solves the equation
$$U_{pl}^R(\phi|s^*,t^*) = 0.$$
The estimate $\hat{\phi}$ could be obtained computationally via gradient-based methods; in particular, through the Newton-Raphson iteration
$$\hat{\phi} \leftarrow \hat{\phi} + \left[I_{pl}^R(\hat{\phi}|s^*,t^*)\right]^{-1} \left[U_{pl}^R(\hat{\phi}|s^*,t^*)\right].$$

Next, we seek alternative representations of $U_{pl}^R$ and $I_{pl}^R$ that will facilitate obtaining the asymptotic properties of $\hat{\phi}$. We recall that
\begin{eqnarray*}
M_i^R(ds,t;q) & = & I\{\mathcal{E}_{iq}(s) \le t^*\} M_i^R(ds;q) \\
& = & I\{\mathcal{E}_{iq}(s) \le t^*\} \{N_i^R(ds;q) - A_i^R(ds;q)\},
\end{eqnarray*}
where
$$A_i^R(ds;q) = Y_i(s) \kappa_{iq}(s|\phi_0) \lambda_{0q}^0[\mathcal{E}_{iq}(s)] ds.$$
The process $\{M_i^R(s,t;q) = \int_0^s M_i^R(dv,t;q): s \ge 0\}$ is a zero-mean square-integrable martingale with predictable quadratic variation process $$\{\langle M_i^R(\cdot,t^*;q)\rangle(s) = \int_0^s I\{\mathcal{E}_{iq}(v) \le t^*\} A_i^R(dv;q): s \ge 0\},$$ or, in differential form,
$$\langle M_i^R(\cdot,t^*;q) \rangle(ds) = I\{\mathcal{E}_{iq}(s) \le t^*\} A_i^R(ds;q).$$
We now present ``change-of-variable'' results that will simplify the establishment of some results.

\begin{proposition}
\label{prop: change-of-variable}
For $q = 1, 2, \ldots, Q$, and functions $h$, $h_i, i=1,2,\ldots,n$, we have
\begin{eqnarray}
\lefteqn{\sum_{i=1}^n \int_0^{s^*} h[\mathcal{E}_{iq}(s)] I\{\mathcal{E}_{iq}(s) \le t^*\}  A_i^R(ds;q) = } \nonumber \\
 && \int_0^{t^*} h(w) S_q^{0R}(s^*,w|\phi_0) \Lambda_{0q}^0(dw); \label{change-of-variable result 1}
\end{eqnarray}
\begin{eqnarray}
\lefteqn{\sum_{i=1}^n \int_0^{s^*} h_i(s) I\{\mathcal{E}_{iq}(s) \le t^*\}  A_i^R(ds;q) = } \nonumber \\
 && \int_0^{t^*} E_{\mathbb{P}_q(s^*,w|\phi_0)} \{h_I[\mathcal{E}_{IqJ}^{-1}(w)]\} 
 S_q^{0R}(s^*,w|\phi_0) \Lambda_{0q}^0(dw). \label{change-of-variable result 2}
\end{eqnarray}
%
%where
%%
%$$E_{\mathbb{P}_q(s^*,w|\phi_0)} \{h_I[\mathcal{E}_{IqJ}^{-1}(w)]\} = 
%\sum_{i=1}^n \sum_{j=1}^{K+1} h_i[\mathcal{E}_{iqj}^{-1}(w)] \mathbb{P}_q(s^*,w|\phi_0)(i,j).$$
%
\end{proposition}

\begin{proof}
To prove (\ref{change-of-variable result 1}), we have
\begin{eqnarray*}
\lefteqn{\sum_{i=1}^n \int_0^{s^*} h[\mathcal{E}_{iq}(s)] I\{\mathcal{E}_{iq}(s) \le t^*\}  A_i^R(ds;q) } \\
& = & \sum_{i=1}^n\sum_{j=1}^{K+1} \int_{S_{j-1}\wedge s^*}^{S_j\wedge s^*} h[\mathcal{E}_{iq}(s)] 
I\{\mathcal{E}_{iq}(s) \le t^*\} I(s \le \tau_i^*) \kappa_{iq}(s|\phi_0) \lambda_{0q}^0[\mathcal{E}_{iq}(s)] ds \\
\end{eqnarray*}
\begin{eqnarray*}
& = & \int_0^{t^*} h(w) \sum_{i=1}^n \sum_{j=1}^{K+1} I\{w \in (\mathcal{E}_{iq}(S_{j-1}\wedge s^*\wedge \tau_i),
\mathcal{E}_{iq}(S_{j}\wedge s^*\wedge \tau_i)]\} \times \\
\end{eqnarray*}
\begin{eqnarray*}
&& \frac{\kappa_{iq}[\mathcal{E}_{iqj}^{-1}(w)|\phi_0]}{\mathcal{E}_{iq}^\prime[\mathcal{E}_{iqj}^{-1}(w)]} \Lambda_{0q}^0(dw) \quad \mbox{(by letting $w = \mathcal{E}_{iq}(s)$)} \\
& = & \int_0^{t^*} h(w) S_q^{0R}(s^*,w|\phi_0) \Lambda_{0q}^0(dw).
\end{eqnarray*}
To prove (\ref{change-of-variable result 2}), 
\begin{eqnarray*}
\lefteqn{\sum_{i=1}^n \int_0^{s^*} h_i(s) I\{\mathcal{E}_{iq}(s) \le t^*\}  A_i^R(ds;q) }  \\
& = & \sum_{i=1}^n\sum_{j=1}^{K+1} \int_{S_{j-1}\wedge s^*}^{S_j\wedge s^*} h_i(s)
I\{\mathcal{E}_{iq}(s) \le t^*\} I(s \le \tau_i^*) \kappa_{iq}(s|\phi_0) \lambda_{0q}^0[\mathcal{E}_{iq}(s)] ds \\
& = & \int_0^{t^*} \sum_{i=1}^n \sum_{j=1}^{K+1} h_i[\mathcal{E}_{iqj}^{-1}(w) I\{w \in (\mathcal{E}_{iq}(S_{j-1}\wedge s^*\wedge \tau_i),
\mathcal{E}_{iq}(S_{j}\wedge s^*\wedge \tau_i)]\} \times \\
&& \frac{\kappa_{iq}[\mathcal{E}_{iqj}^{-1}(w)|\phi_0]}{\mathcal{E}_{iq}^\prime[\mathcal{E}_{iqj}^{-1}(w)]} \Lambda_{0q}^0(dw) \quad \mbox{(by letting $w = \mathcal{E}_{iq}(s)$)} \\
& = & \int_0^{t^*} \left\{\sum_{i=1}^n \sum_{j=1}^{K+1} h_i[\mathcal{E}_{iqj}^{-1}(w) \mathbb{P}_{q}(s^*,w|\phi_0)(i,j)\right\} S_{q}^{0R}(s^*,w|\phi_0) \Lambda_{0q}^0(dw). 
\end{eqnarray*}
\end{proof}

\begin{proposition}
\label{prop: martingale representation UplR}
The restricted profile score function for $\phi$ at $\phi_0$ has the martingale representation
\begin{eqnarray*}
U_{pl}^R(\phi_0|s^*,t^*) & = & \sum_{q=1}^Q \sum_{i=1}^{K+1} \int_0^{s^*}
\left\{
\frac{\stackrel{\bullet}{\kappa}_{iq}}{\kappa_{iq}}(s|\phi_0) - \frac{\stackrel{\bullet}{S}_q^{0R}}{S_q^{0R}}(s^*,\mathcal{E}_{iq}(s)|\phi_0)\right\}  \times \\
&& 
I\{\mathcal{E}_{iq}(s) \le t^*\} M_i^R(ds;q);
\end{eqnarray*}
with predictable quadratic variation
\begin{eqnarray*}
\langle U_{pl}^R(\phi_0|\cdot,t^*)\rangle(s^*) = \sum_{q=1}^Q \int_0^{t^*} V_q(s^*,w|\phi_0) S_{q}^{0R}(s^*,w|\phi_0) \Lambda_{0q}^0(dw),
\end{eqnarray*}
where
$$V_q(s^*,w|\phi) = V_{\mathbb{P}_q(s^*,w|\phi)}\left[\frac{\stackrel{\bullet}{\kappa}_{Iq}}{\kappa_{Iq}}\left(\mathcal{E}_{IqJ}^{-1}(w)|\phi\right)\right].$$
\end{proposition}

\begin{proof}
To prove the first result it suffices to show that, for each $q \in \IQ$,
\begin{eqnarray*}
\lefteqn{\sum_{i=1}^n \int_0^{s^*} \frac{\stackrel{\bullet}{\kappa}_{iq}}{\kappa_{iq}}(s|\phi_0) I\{\mathcal{E}_{iq}(s) \le t^*\} A_i^R(ds;q) = } \\
&& \sum_{i=1}^n \int_0^{s^*} \frac{\stackrel{\bullet}{S}_q^{0R}}{{S}_q^{0R}}\left(s^*,\mathcal{E}_{iq}(s)|\phi_0\right)
I\{\mathcal{E}_{iq}(s) \le t^*\} A_i^R(ds;q).
\end{eqnarray*}
Using (\ref{change-of-variable result 2}) in Proposition \ref{prop: change-of-variable}, the left-hand side equals
$$\int_0^{t^*} \left\{E_{\mathbb{P}_q(s^*,w|\phi_0)} \left[\frac{\stackrel{\bullet}{\kappa}_{Iq}}{\kappa_{Iq}}\left(\mathcal{E}_{IqJ}^{-1}(w)|\phi_0\right)\right]\right\} S_q^{0R}(s^*,w|\phi_0) \Lambda_{0q}^0(dw).$$
Using (\ref{change-of-variable result 1}) in Proposition \ref{prop: change-of-variable}, the right-hand side equals
$$\int_0^{t^*}  \left[\frac{\stackrel{\bullet}{S}_q^{0R}}{{S}_q^{0R}}\left(s^*,w|\phi_0\right)\right] {S}_q^{0R}(s^*,w|\phi_0) \Lambda_{0q}^0(dw).$$
However,
\begin{eqnarray*}
E_q^R\left(s^*,w|\phi_0\right) \equiv \frac{\stackrel{\bullet}{S}_q^{0R}}{{S}_q^{0R}}\left(s^*,w|\phi_0\right) = 
E_{\mathbb{P}_q(s^*,w|\phi_0)} \left[\frac{\stackrel{\bullet}{\kappa}_{Iq}}{\kappa_{Iq}}\left(\mathcal{E}_{IqJ}^{-1}(w)|\phi_0\right)\right],
\end{eqnarray*}
completing the proof of the first result.

To prove the second result, we have
\begin{eqnarray*}
\lefteqn{\langle U_{pl}^R(\phi_0|\cdot,t^*)\rangle(s^*) = \sum_{q=1}^Q \sum_{i=1}^n \int_0^{s^*} 
\left[\frac{\stackrel{\bullet}{\kappa}_{iq}}{\kappa_{iq}}(s|\phi_0) - \frac{\stackrel{\bullet}{S}_q^{0R}}{{S}_q^{0R}}\left(s^*,\mathcal{E}_{iq}(s)|\phi_0\right)\right]^{\otimes 2} \times } \\ && I\{\mathcal{E}_{iq}(s) \le t^*\} A_i^R(ds;q)  \\
& = & \sum_{q=1}^Q \sum_{i=1}^n \int_0^{s^*} \left[\frac{\stackrel{\bullet}{\kappa}_{iq}}{\kappa_{iq}}(s|\phi_0)\right]^{\otimes 2} I\{\mathcal{E}_{iq}(s) \le t^*\} A_i^R(ds;q) \\
&& -2\sum_{q=1}^Q \sum_{i=1}^n \int_0^{s^*} \left[\frac{\stackrel{\bullet}{\kappa}_{iq}}{\kappa_{iq}}(s|\phi_0)\right]\left[\frac{\stackrel{\bullet}{S}_q^{0R}}{{S}_q^{0R}}\left(s^*,\mathcal{E}_{iq}(s)|\phi_0\right)\right]\trp \times \\ && I\{\mathcal{E}_{iq}(s) \le t^*\} A_i^R(ds;q) \\
&& + \sum_{q=1}^Q \sum_{i=1}^n \int_0^{s^*} \left[\frac{\stackrel{\bullet}{S}_q^{0R}}{{S}_q^{0R}}\left(s^*,\mathcal{E}_{iq}(s)|\phi_0\right)\right]^{\otimes 2} I\{\mathcal{E}_{iq}(s) \le t^*\} A_i^R(ds;q).
\end{eqnarray*}
We consider each of these three terms in the last expression. Applying the techniques in proving Proposition \ref{prop: change-of-variable}, the second term becomes
\begin{eqnarray*}
\lefteqn{-2\sum_{q=1}^Q \int_0^{t^*} \left\{E_{\mathbb{P}_q(s^*,w|\phi_0)}
\left[\frac{\stackrel{\bullet}{\kappa}_{Iq}}{\kappa_{Iq}}\left(\mathcal{E}_{IqJ}^{-1}(w)|\phi_0\right)\right]\right\} } \\
&& S_q^{0R}(s^*,w|\phi_0) \left[\frac{\stackrel{\bullet}{S}_q^{0R}}{{S}_q^{0R}}\left(s^*,w|\phi_0\right)\right]\trp \Lambda_{0q}^0(dw)  \\
& = &  -2\sum_{q=1}^Q \int_0^{t^*}  \left[\frac{\stackrel{\bullet}{S}_q^{0R}}{{S}_q^{0R}}\left(s^*,w|\phi_0\right)\right]^{\otimes 2}
S_q^{0R}(s^*,w|\phi_0)   \Lambda_{0q}^0(dw),
\end{eqnarray*}
and since the third term by Proposition \ref{prop: change-of-variable} equals
$$\sum_{q=1}^Q \int_0^{t^*}  \left[\frac{\stackrel{\bullet}{S}_q^{0R}}{{S}_q^{0R}}\left(s^*,w|\phi_0\right)\right]^{\otimes 2}
S_q^{0R}(s^*,w|\phi_0)   \Lambda_{0q}^0(dw),$$
then the last two terms simplify to
$$-\sum_{q=1}^Q \int_0^{t^*}  \left[\frac{\stackrel{\bullet}{S}_q^{0R}}{{S}_q^{0R}}\left(s^*,w|\phi_0\right)\right]^{\otimes 2}
S_q^{0R}(s^*,w|\phi_0)   \Lambda_{0q}^0(dw).$$
Finally, the first term, applying Proposition \ref{prop: change-of-variable}, equals
$$\sum_{q=1}^Q \int_0^{t^*} E_{\mathbb{P}_q(s^*,w|\phi_0)} \left\{
\left[\frac{\stackrel{\bullet}{\kappa}_{Iq}}{\kappa_{Iq}}\left(\mathcal{E}_{IqJ}^{-1}(w)|\phi_0\right)\right]^{\otimes 2}\right\}
S_q^{0R}(s^*,w|\phi_0) \Lambda_{0q}^0(dw),$$
so that 
\begin{eqnarray*}
\lefteqn{\langle U_{pl}^R(\phi_0|\cdot,t^*)\rangle(s^*)   =  \sum_{q=1}^Q \int_0^{t^*}
\left\{ V_{\mathbb{P}_q(s^*,w|\phi_0)}
\left[\frac{\stackrel{\bullet}{\kappa}_{Iq}}{\kappa_{Iq}}\left(\mathcal{E}_{IqJ}^{-1}(w)|\phi_0\right)\right] \right\} \times } \\
&& S_q^{0R}(s^*,w|\phi_0) \Lambda_{0q}^0(dw) 
 =  \sum_{q=1}^Q \int_0^{t^*} V_q^R(s^*,w|\phi_0) S_q^{0R}(s^*,w|\phi_0) \Lambda_{0q}^0(dw).
\end{eqnarray*}
\end{proof}

Next we provide a result relating the observed information matrix and the predictable quadratic variation matrix.

\begin{proposition}
\label{prop: relation between IplR and PQV}
As $n \rightarrow \infty$,
$$\frac{1}{n} I_{pl}^R(\phi|s^*,t^*) = \frac{1}{n} \langle U_{pl}^R(\phi|\cdot,t^*)\rangle(s^*) + o_p(1).$$
\end{proposition}

\begin{proof}
Using Proposition \ref{prop: change-of-variable}, we have that
\begin{eqnarray*}
\lefteqn{\sum_{i=1}^n \int_0^{s^*} \left[
\frac{\stackrel{\bullet\bullet}{\kappa}_{iq}}{\kappa_{iq}}(s|\phi) - 
\frac{\stackrel{\bullet\bullet}{S}_q^{0R}}{{S}_q^{0R}}\left(s^*,\mathcal{E}_{iq}(s)|\phi\right)
\right]
I\{\mathcal{E}_{iq}(s) \le t^*\} A_i^R(ds;q) } \\
& = & \int_0^{t^*} \left\{ 
E_{\mathbb{P}_q(s^*,w|\phi)} \left[
\frac{\stackrel{\bullet\bullet}{\kappa}_{Iq}}{\kappa_{Iq}}\left(\mathcal{E}_{IqJ}^{-1}(w)|\phi\right) 
\right] 
\right\} S_q^{0R}(s^*,w|\phi_0) \Lambda_{0q}^0(dw) - \\
&& \int_0^{t^*} \left[
\frac{\stackrel{\bullet\bullet}{S}_{q}^{0R}}{{S}_{q}^{0R}}\left(s^*,w|\phi\right)
\right] S_q^{0R}(s^*,w|\phi_0) \Lambda_{0q}^0(dw).
\end{eqnarray*}
But
$$E_{\mathbb{P}_q(s^*,w|\phi)} \left[
\frac{\stackrel{\bullet\bullet}{\kappa}_{Iq}}{\kappa_{Iq}}\left(\mathcal{E}_{IqJ}^{-1}(w)|\phi\right) 
\right]  = 
\frac{\stackrel{\bullet\bullet}{S}_{q}^{0R}}{{S}_{q}^{0R}}\left(s^*,w|\phi\right)$$
so the difference in the preceding equation equals zero. Since
\begin{eqnarray*}
\lefteqn{ \frac{1}{\sqrt{n}} \sum_{q=1}^Q \sum_{i=1}^n \int_0^{s^*} 
 \left[
\left\{
\frac{\stackrel{\bullet\bullet}{\kappa}_{iq}}{\kappa_{iq}}(s|\phi) - \left(\frac{\stackrel{\bullet}{\kappa}_{iq}}{\kappa_{iq}}(s|\phi)\right)^{\otimes 2}\right\}  - \right. } \\ &&
\left.
\left\{\frac{\stackrel{\bullet\bullet}{S}_q^{0R}}{S_q^{0R}}(s^*,\mathcal{E}_{iq}(s)|\phi) - \left(\frac{\stackrel{\bullet}{S}_q^{0R}}{S_q^{0R}}(s^*,\mathcal{E}_{iq}(s)|\phi)\right)^{\otimes 2}\right\}
\right]
 \times  \\
&& 
I\{\mathcal{E}_{iq}(s) \le t^*\} M_i^R(ds;q) = O_p(1)
\end{eqnarray*}
by invoking the Martingale Central Limit Theorem, then the claimed result is proved.
\end{proof}

\begin{proposition}
\label{prop: alternative form of VqR}
If $\kappa_{iq}(s|\phi) = \exp\{Z_i(s-) \phi\}$, where $Z_i(s-)$ is $\mathfrak{F}_{s-}$-measurable, then, for each $q = 1, 2, \ldots, Q$,
\begin{eqnarray*}
V_q^R(s^*,w|\phi) & = & V_{\mathbb{P}_q(s^*,w|\phi)}\left[\frac{\stackrel{\bullet}{\kappa}_{Iq}}{\kappa_{Iq}}\left(\mathcal{E}_{IqJ}^{-1}(w)|\phi\right)\right] \\
& = & \frac{\stackrel{\bullet\bullet}{S}_q^{0R}}{S_q^{0R}}(s^*,w|\phi) -
\left[ \frac{\stackrel{\bullet}{S}_q^{0R}}{S_q^{0R}}(s^*,w|\phi) \right]^{\otimes 2}.    
\end{eqnarray*}
\end{proposition}

\begin{proof}
The result immediately follows from the fact that
$$E_{\mathbb{P}_q(s^*,w|\phi)}\left[\frac{\stackrel{\bullet\bullet}{\kappa}_{Iq}}{\kappa_{Iq}}\left(\mathcal{E}_{IqJ}^{-1}(w)|\phi\right)\right] = 
E_{\mathbb{P}_q(s^*,w|\phi)}\left[\left(\frac{\stackrel{\bullet}{\kappa}_{Iq}}{\kappa_{Iq}}\left(\mathcal{E}_{IqJ}^{-1}(w)|\phi\right)\right)^{\otimes 2}\right]$$
because, under the exponential form of $\kappa_{iq}(\cdot|\phi)$, we have
$$E_{\mathbb{P}_q(s^*,w|\phi)} \left[
\nabla_{\phi^{\otimes 2}} \log \kappa_{iq}\left(\mathcal{E}_{IqJ}^{-1}(w)|\phi\right)
\right] = 0.$$
\end{proof}

Henceforth, we shall assume that $\kappa_{iq}(s|\phi) = \exp\{Z_i(s-) \phi\}$, where $Z_i(s-)$ is a function of $N_i^R(s-)$. This linear form may exclude, however, some load-sharing models that arise in software reliability engineering (see, for instance, the load-sharing models in \cite{KvaPen2005}). For the more general model that includes such load-sharing models, the regularity conditions presented below need to be amended to impose conditions on the more general form of the $V_q(s^*,w|\phi)$s. The exponential form above for $\kappa_{iq}$s leads to regularity conditions that mimic conditions contained in previous papers such as \cite{AndGil1982} and \cite{Gill1984} dealing with the multiplicative intensity model or the Cox proportional hazards model.

\medskip

\begin{center}
\underline{RCR-Component Regularity Conditions}
\end{center}

\medskip

There exists a neighborhood $\Phi_0 \subset \Phi$ of $\phi_0$, and for each $q \in \IQ$, a scalar function, a $\dim(\phi) \times 1$-dimensional vector function, and a $\dim(\phi) \times \dim(\phi)$-dimensional matrix function, defined on $\T \times [0,t^*] \time \Phi$, and given, respectively, by
$$s_q^{0R}(s^*,w|\phi),\quad s_q^{1R}(s^*,w|\phi),\quad \mbox{and}\quad s_q^{2R}(s^*,w|\phi),$$
which, aside from its dependence on $\phi$, may depend on all the true values of the model parameters, and satisfying the following conditions:

\begin{itemize}
\item[(a)]
For $k = 0, 1, 2$, and $q \in \IQ$, as $n \rightarrow \infty$, and with $\stackrel{\bullet}{S}_q^{0R} = S_q^{1R}$ and $\stackrel{\bullet\bullet}{S}_q^{0R} = S_q^{2R}$, we have
$$\sup_{w \in [0,t^*]; \phi \in \Phi_0} \left\| \frac{1}{n} S_q^{kR}(s^*,w|\phi) - s_q^{kR}(s^*,w|\phi) \right\| \stackrel{p}{\rightarrow} 0;$$

\item[(b)]
For $k = 0, 1, 2$, and $q \in \IQ$, and for $w \in [0,t^*]$, $s_q^{kR}(s^*,w|\phi)$ are continuous in $\phi$ uniformly in $w$, and bounded on $[0,t^*] \times \Phi_0$;

\item[(c)]
For each $q \in \IQ$, $s_q^{0R}(s^*,t^*|\phi_0) > 0$ and $$\int_0^{t^*} \frac{\Lambda_{0q}^0(dw)}{s_q^{0R}(s^*,w|\phi_0)} < \infty;$$

\item[(d)]
For $q \in \IQ$, and for $w \in [0,t^*]$, each $s_q^{0R}(s^*,w|\phi)$ is twice-differentiable in $\phi$ , and with
$$s_q^{1R}(s^*,w|\phi) = \nabla_\phi s_q^{0R}(s^*,w|\phi) \ \mbox{and} \
s_q^{2R}(s^*,w|\phi) = \nabla_{\phi^{\otimes 2}} s_q^{0R}(s^*,w|\phi);$$

\item[(e)]
For $q \in \IQ$ and $w \in [0,t^*]$, with
$$e_q^R(s^*,w|\phi) = \frac{s_q^{1R}(s^*,w|\phi)}{s_q^{0R}(s^*,w|\phi)};$$
$$v_q^R(s^*,w|\phi) = \frac{s_q^{2R}(s^*,w|\phi)}{s_q^{0R}(s^*,w|\phi)} - \left[e_q^R(s^*,w|\phi)\right]^{\otimes 2},$$
then each of the matrices
$$\Sigma_q^R(s^*,t^*) = \int_0^{t^*} v_q^R(s^*,w|\phi_0) s_q^{0R}(s^*,w|\phi_0) \Lambda_{0q}^0(dw)$$
is positive definite, so that the matrix
$$\Sigma^R(s^*,t^*) = \sum_{q=1}^Q \Sigma_q^R(s^*,t^*)$$
is also positive definite.
\end{itemize}

We now present asymptotic properties of the estimator $\hat{\phi}(s^*,t^*)$.

\begin{proposition}
\label{prop: asymptotic results of phihat}
Under the ``RCR-Component Regularity Conditions,'' as $n \rightarrow \infty$,
\begin{itemize}
\item[(i)]
$\hat{\phi}(s^*,t^*) \stackrel{p}{\rightarrow} \phi_0$; 
\item[(ii)]
$\sqrt{n}[\hat{\phi}(s^*,t^*) - \phi_0] \stackrel{d}{\rightarrow} N\left(0,\left[\Sigma^R(s^*,t^*)\right]^{-1}\right)$.
\end{itemize}
\end{proposition}

\begin{proof}
The consistency result is established using analogous techniques in proving the consistency of the estimator $\hat{\theta}^W$ in the preceding subsection, which builds on those in \cite{AndGil1982}, \cite{Bor84}, \cite{ander1993}, and \cite{pena2016}. Thus, we do not anymore present its details.

To prove the asymptotic normality result, since $U_{pl}^R(\hat{\phi}(s^*,t^*)|s^*,t^*) = 0$, then by Taylor expansion
$$U_{pl}^R(\phi_0|s^*,t^*) - I_{pl}^R(\phi^\dagger|s^*,t^*) (\hat{\phi}(s^*,t^*) - \phi_0) = 0$$
where $\phi^\dagger \in [\phi_0,\hat{\phi}(s^*,t^*)]$. Thus,
$$\left[\frac{1}{n} I_{pl}^R(\phi^\dagger|s^*,t^*)\right] \left[\sqrt{n}(\hat{\phi}(s^*,t^*) - \phi_0)\right] =
\frac{1}{\sqrt{n}} U_{pl}^R(\phi_0|s^*,t^*).$$
From the regularity conditions, and since $\phi^\dagger \rightarrow \phi_0$ by virtue of $\hat{\phi}(s^*,t^*) \stackrel{p}{\rightarrow} \phi_0$ and Proposition \ref{prop: relation between IplR and PQV}, we have
$$\left[\frac{1}{n} I_{pl}^R(\phi^\dagger|s^*,t^*)\right] \stackrel{p}{\rightarrow} \Sigma^R(s^*,t^*)$$
which is assumed positive definite. Also, by the Martingale Central Limit Theorem,
$$\frac{1}{\sqrt{n}} U_{pl}^R(\phi_0|s^*,t^*) =
\frac{1}{\sqrt{n}} \sum_{q=1}^Q \sum_{i=1}^n \int_0^{s^*}
H_{iq}^R(s^*,s|\phi_0) I\{\mathcal{E}_{iq}(s) \le t^*\} M_i^R(ds;q)$$
where
$$H_{iq}^R(s^*,s|\phi) = \frac{\stackrel{\bullet}{\kappa}_{iq}}{\kappa_{iq}}(s|\phi) - 
 \frac{\stackrel{\bullet}{S}_{q}^{0R}}{{S}_{q}^{0R}}\left(s^*,\mathcal{E}_{iq}(s)|\phi\right),$$
 converges in distribution to a $N(0,\Sigma^R(s^*,t^*))$ random vector by using Proposition \ref{prop: martingale representation UplR}. As a consequence, we obtain that
 $$\sqrt{n}[\hat{\phi}(s^*,t^*) - \phi_0] \stackrel{d}{\rightarrow} N\left(0,\left[\Sigma^R(s^*,t^*)\right]^{-1}\right).$$
\end{proof}

We now establish asymptotic results for the estimators of $\Lambda_{0q}^0(\cdot), q \in \IQ$.

\begin{theorem}
\label{theo: asymptotics of LambdaHats}
Under the ``RCR-Component Regularity Conditions", as $n \rightarrow \infty$,
we have:
\begin{itemize}
\item[(i)]
For each $q \in \IQ$, $\sup_{t \in [0,t^*]} \left|\hat{\Lambda}_{0q}(s^*,t) - \Lambda_{0q}^0(t)\right| \stackrel{p}{\rightarrow} 0$;

\item[(ii)]
The $Q$-dimensional process
\begin{eqnarray*}
\left\{
\sqrt{n}
\left[
(\hat{\Lambda}_{0q}(s^*,t) - \Lambda_{0q}^0(t)), q \in \IQ
\right]:
t \in [0,t^*]
\right\}
\end{eqnarray*}
converges weakly on Skorokhod's $\mathfrak{D}^Q[0,t^*]$-space to a zero-mean $Q$-dimensional multivariate Gaussian process $$G^R = \left\{(G_q^R(s^*,t), q \in \IQ): t \in [0,t^*]\right\}$$ with covariance kernel given by, for $q_1, q_2 \in \IQ$ and $t_1, t_2 \in [0,t^*]$,
\begin{eqnarray*}
\lefteqn{ \Gamma^R_{q_1,q_2}(s^*,t_1,t_2)  \equiv Cov\left\{G_{q_1}^R(s^*,t_1),G_{q_2}^R(s^*,t_2)\right\} } \\  
& = &  \left\{I\{q_1 = q_2\} \int_0^{\min\{t_1,t_2\}} \frac{\Lambda_{0q_1}^0(dw)}{s_{q_1}^{0R}(s^*,w|\phi_0)} \right\} + \\
&&
\left\{\left[\int_0^{t_1} e_{q_1}^R(s^*,w|\phi_0) \Lambda_{0q_1}^0(dw)\right]\trp
[\Sigma^R(s^*,t^*)]^{-1} \right. \\ &&
\left.\left[\int_0^{t_2} e_{q_2}^R(s^*,w|\phi_0) \Lambda_{0q_2}^0(dw)\right]\right\}.
\end{eqnarray*}
%

%\item[(iii)]
%Other results for $q \ne q^\prime$?
\end{itemize}
\end{theorem}

\bigskip

To systematically present the proof of Theorem \ref{theo: asymptotics of LambdaHats}, we first establish some intermediate results. We start with an asymptotic representation of $\hat{\Lambda}_{0q}(s^*,t)$. Define, for $q \in \IQ$ and $t \in [0,t^*]$,
\begin{displaymath}
\Lambda_{0q}^*(s^*,t) = \int_0^t I\{S_q^{0R}(s^*,w|\phi_0)\} \Lambda_{0q}^0(dw).
\end{displaymath}
Then, we have the following representations.

\begin{lemma}
\label{lemm: representation of LambdaHat}
Under the ``RCR-Component Regularity Conditions", for each $q\in\IQ$ and $t \in [0,t^*]$, and as $n \rightarrow \infty$,
\begin{eqnarray*}
\lefteqn{\hat{\Lambda}_{0q}(s^*,t) - \Lambda_{0q}^*(s^*,t) = } \\
&& \sum_{i=1}^n \int_0^{s^*}
\frac{I\{S_q^{0R}(s^*,\mathcal{E}_{iq}(s)|\phi_0) > 0\}}{S_q^{0R}(s^*,\mathcal{E}_{iq}(s)|\phi_0)}
I\{\mathcal{E}_{iq}(s) \le t\} M_i^R(ds;q) - \\
&&  \left[
\int_0^t I\{S_q^{0R}(s^*,w|\phi_0) > 0\} E_q^R(s^*,w|\phi_0) \Lambda_{0q}^0(dw) \right]\trp
(\hat{\phi} - \phi_0) +  o_p(1). 
%\\
%& = & \int_0^{t} \frac{I\{S_q^{0R}(s^*,w|\phi_0) > 0\}}{S_q^{0R}(s^*,w|\phi_0)} \sum_{i=1}^n M_i^R(s^*,dw;q) - \\
%&&  \left[
%\int_0^t I\{S_q^{0R}(s^*,w|\phi_0) > 0\} E_q^R(s^*,w|\phi_0) \Lambda_{0q}^0(dw) \right]
%(\hat{\phi} - \phi_0) +  o_p(1).
\end{eqnarray*}
\end{lemma}

\begin{proof}
We have that
\begin{eqnarray*}
\hat{\Lambda}_{0q}(s^*,t) & = & \int_0^{t}
\frac{I\{S_q^{0R}(s^*,w|\hat{\phi}) > 0\}}{S_q^{0R}(s^*,w|\hat{\phi})} \sum_{i=1}^n N_i^R(s^*,dw;q) \\
& = & \sum_{i=1}^n \int_0^{s^*} 
\frac{I\{S_q^{0R}(s^*,\mathcal{E}_{iq}(s)|\hat{\phi}) > 0\}}{S_q^{0R}(s^*,\mathcal{E}_{iq}(s)|\hat{\phi})}
I\{\mathcal{E}_{iq}(s) \le t\} N_i^R(ds;q).
\end{eqnarray*}
By Taylor expansion,
%
%\begin{eqnarray*}
$$ \frac{1}{S_q^{0R}(s^*,\mathcal{E}_{iq}(s)|\hat{\phi})} = \frac{1}{S_q^{0R}(s^*,\mathcal{E}_{iq}(s)|\phi_0)} -
\frac{[\stackrel{\bullet}{S}_q^{0R}(s^*,\mathcal{E}_{iq}(s)|\phi^\dagger)]\trp}{[{S}_q^{0R}(s^*,\mathcal{E}_{iq}(s)|\phi^\dagger)]^2} (\hat{\phi} - \phi_0) $$
%\end{eqnarray*}
%
where $\phi^\dagger \in [\phi_0,\hat{\phi}]$, so that $\phi^\dagger \stackrel{p}{\rightarrow} \phi_0$. As such
\begin{eqnarray*}
\lefteqn{\hat{\Lambda}_{0q}(s^*,t)  =  \sum_{i=1}^n \int_0^{s^*} 
\frac{I\{S_q^{0R}(s^*,\mathcal{E}_{iq}(s)|\hat{\phi}) > 0\}}{S_q^{0R}(s^*,\mathcal{E}_{iq}(s)|\phi_0)}
I\{\mathcal{E}_{iq}(s) \le t\} N_i^R(ds;q) - } \\
&& \left\{
\sum_{i=1}^n \int_0^{s^*}
\frac{I\{S_q^{0R}(s^*,\mathcal{E}_{iq}(s)|\hat{\phi}) > 0\}}{S_q^{0R}(s^*,\mathcal{E}_{iq}(s)|\phi^\dagger)}
\left[E_q^R(s^*,\mathcal{E}_{iq}(s)|\phi^\dagger)\right]\trp \right. \\
&& \left.
I\{\mathcal{E}_{iq}(s) \le t\} N_i^R(ds;q) \right\}
(\hat{\phi} - \phi_0).
\end{eqnarray*}
Since $\hat{\phi} \stackrel{p}{\rightarrow} \phi_0$ and $\sup_{t \in [0,t^*]; \phi\in\Gamma_0}
\left|\frac{1}{n} S_q^{kR}(s^*,t|\phi) - s_q^{kR}(s^*,t|\phi)\right| \stackrel{p}{\rightarrow} \phi_0$ for $k = 0, 1$, then
\begin{eqnarray*}
\lefteqn{\hat{\Lambda}_{0q}(s^*,t)  =  \sum_{i=1}^n \int_0^{s^*} 
\frac{I\{S_q^{0R}(s^*,\mathcal{E}_{iq}(s)|\phi_0) > 0\}}{S_q^{0R}(s^*,\mathcal{E}_{iq}(s)|\phi_0)}
I\{\mathcal{E}_{iq}(s) \le t\} N_i^R(ds;q) - } \\
&& \left\{
\sum_{i=1}^n \int_0^{s^*}
\frac{I\{S_q^{0R}(s^*,\mathcal{E}_{iq}(s)|\phi_0) > 0\}}{S_q^{0R}(s^*,\mathcal{E}_{iq}(s)|\phi_0)}
\left[E_q^R(s^*,\mathcal{E}_{iq}(s)|\phi_0)\right]\trp \right. \\
&& \left.
I\{\mathcal{E}_{iq}(s) \le t\} N_i^R(ds;q) \right\}
(\hat{\phi} - \phi_0) + o_p(1).
\end{eqnarray*}
The first-term above is
\begin{eqnarray*}
\lefteqn{\hat{\Lambda}_{0q}(s^*,t)  =  \sum_{i=1}^n \int_0^{s^*} 
\frac{I\{S_q^{0R}(s^*,\mathcal{E}_{iq}(s)|\phi_0) > 0\}}{S_q^{0R}(s^*,\mathcal{E}_{iq}(s)|\phi_0)}
I\{\mathcal{E}_{iq}(s) \le t\} N_i^R(ds;q) - } \\
& = & \sum_{i=1}^n \int_0^{s^*} 
\frac{I\{S_q^{0R}(s^*,\mathcal{E}_{iq}(s)|\phi_0) > 0\}}{S_q^{0R}(s^*,\mathcal{E}_{iq}(s)|\phi_0)}
I\{\mathcal{E}_{iq}(s) \le t\} M_i^R(ds;q) + \\
& & \sum_{i=1}^n \int_0^{s^*} 
\frac{I\{S_q^{0R}(s^*,\mathcal{E}_{iq}(s)|\phi_0) > 0\}}{S_q^{0R}(s^*,\mathcal{E}_{iq}(s)|\phi_0)}
I\{\mathcal{E}_{iq}(s) \le t\} A_i^R(ds;q).
\end{eqnarray*}
But the last-term above, using Proposition \ref{prop: change-of-variable}, equals $\Lambda_{0q}^*(s^*,t)$. This establishes the lemma.
\end{proof}

\begin{lemma}
\label{lemm: for consistency}
Under the ``RCR-Component Regularity Conditions", for each $q\in\IQ$ and as $n \rightarrow \infty$,
%
%\begin{eqnarray*}
$$\sup_{t \in [0,t^*]} \left| \hat{\Lambda}_{0q}(s^*,t) - \Lambda_{0q}^*(s^*,t) \right| \stackrel{p}{\rightarrow} 0
\quad \mbox{and} \quad
\sup_{t \in [0,t^*]} \left| \Lambda_{0q}^*(s^*,t) - \Lambda_{0q}^0(t) \right| \stackrel{p}{\rightarrow} 0.$$
%\end{eqnarray*}
%
\end{lemma}

\begin{proof}
The second result is immediate from the condition that
$$\sup_{t \in [0,t^*]; \phi \in \Gamma_0} \left|\frac{1}{n}S_q^{0R}(s^*,t|\phi) - s_q^{0R}(s^*,t|\phi)\right| \stackrel{p}{\rightarrow} 0.$$
To prove the first result, from the representation in Lemma \ref{lemm: representation of LambdaHat} and $\hat{\phi} \stackrel{p}{\rightarrow} \phi_0$, it suffices to show that
\begin{itemize}
\item[(i)]
$\sum_{i=1}^n \int_0^{s^*} 
\frac{I\{S_q^{0R}(s^*,\mathcal{E}_{iq}(s)|\phi_0) > 0\}}{S_q^{0R}(s^*,\mathcal{E}_{iq}(s)|\phi_0)}
I\{\mathcal{E}_{iq}(s) \le t\} M_i^R(ds;q) = o_p(1);$
\item[(ii)]
$\int_0^{t^*} I\{S_q^{0R}(s^*,w|\phi_0) > 0\} E_q^R(s^*,w|\phi_0) \Lambda_{0q}^0(dw) = O_p(1).$
\end{itemize}
But (ii) is immediate from the regularity conditions that, for $k = 0, 1$, $$\sup_{t \in [0,t^*]; \phi \in \Gamma_0} \left|\frac{1}{n} S_q^{kR}(s^*,t|\phi) - s_q^{kR}(s^*,t|\phi)\right| \stackrel{p}{\rightarrow} 0$$ and $s_q^{0R}(s^*,t^*|\phi_0) > 0$.
To prove (i), let $\epsilon > 0$ and $\delta > 0$, but arbitrary. We may invoke Lenglart's Inequality to obtain
\begin{eqnarray*}
\lefteqn{ \Pr\left\{\sup_{t \in [0,t^*]} \left|
\sum_{i=1}^n \int_0^{s^*} \frac{I\{S_q^{0R}(s^*,\mathcal{E}_{iq}(s)|\phi_0) > 0\}}{S_q^{0R}(s^*,\mathcal{E}_{iq}(s)|\phi_0)}
I\{\mathcal{E}_{iq}(s) \le t\} M_i^R(ds;q) \right| > \epsilon\right\} } \\
& = & \Pr\left\{\left|
\sum_{i=1}^n \int_0^{s^*} \frac{I\{S_q^{0R}(s^*,\mathcal{E}_{iq}(s)|\phi_0) > 0\}}{S_q^{0R}(s^*,\mathcal{E}_{iq}(s)|\phi_0)}
I\{\mathcal{E}_{iq}(s) \le t^*\} M_i^R(ds;q) \right| > \epsilon\right\} \\
& \le & \frac{\delta}{\epsilon^2} + \\
&& \Pr\left\{\frac{1}{n^2} \sum_{i=1}^n \int_0^{s^*} \frac{I\{\frac{1}{n}S_q^{0R}(s^*,\mathcal{E}_{iq}(s)|\phi_0) > 0\}}{[\frac{1}{n}S_q^{0R}(s^*,\mathcal{E}_{iq}(s)|\phi_0)]^2}
I\{\mathcal{E}_{iq}(s) \le t^*\} A_i^R(ds;q) > \delta\right\}.
\end{eqnarray*}
However, by Proposition \ref{prop: change-of-variable}, 
\begin{eqnarray*}
\lefteqn{\frac{1}{n^2} \sum_{i=1}^n \int_0^{s^*} \frac{I\{\frac{1}{n}S_q^{0R}(s^*,\mathcal{E}_{iq}(s)|\phi_0) > 0\}}{[\frac{1}{n}S_q^{0R}(s^*,\mathcal{E}_{iq}(s)|\phi_0)]^2}
I\{\mathcal{E}_{iq}(s) \le t^*\} A_i^R(ds;q) } \\
& = & \frac{1}{n^2} \int_0^{t^*} \frac{I\{\frac{1}{n}S_q^{0R}(s^*,w|\phi_0) > 0\}}{[\frac{1}{n}S_q^{0R}(s^*,w|\phi_0)]^2} S_q^{0R}(s^*,w|\phi_0) \Lambda_{0q}^0(dw) \\
& = & \frac{1}{n} \int_0^{t^*} \frac{I\{\frac{1}{n}S_q^{0R}(s^*,w|\phi_0) > 0\}}{[\frac{1}{n}S_q^{0R}(s^*,w|\phi_0)]} \Lambda_{0q}^0(dw) \\
& = & o_p(1)
\end{eqnarray*}
since $\frac{1}{n}S_q^{0R}(s^*,w|\phi_0) = O_p(1)$ because, by the regularity conditions we have
$$s_q^{0R}(s^*,t^*|\phi_0) > 0 \quad \mbox{and} \quad \int_0^{t^*} \frac{\Lambda_{0q}^0(dw)}{s_q^{0R}(s^*,w;\phi_0)} < \infty.$$ Thus, for any $\epsilon > 0$, the right-hand side in the inequality above could be made as small as desired for $n \ge N$ by first choosing an appropriate $\delta > 0$, and then an appropriate $N = N(\epsilon,\delta)$. This completes the proof of the lemma.
\end{proof}

Next, let us define, for $t \in [0,t^*]$,
\begin{eqnarray*}
\lefteqn{ Z^R(s^*,t)  =   \left[
\begin{array}{c}
Z_1^R(s^*,t) \\ Z_2^R(s^*,t)
\end{array}
\right] } \\
& = & \frac{1}{\sqrt{n}} \sum_{i=1}^n \int_0^{s^*}
\left[
\begin{array}{c}
H_{i1}^R(s^*,s|\phi_0) \\
H_{i2}^R(s^*,s|\phi_0)
\end{array}
\right]
\DG(I\{\mathcal{E}_{iq}(s) \le t\}, q \in \IQ)
M_i^R(ds)
\end{eqnarray*}
where $M_i^R(ds)  =  \left[M_i^R(ds;q), q \in \IQ\right]$ and 
\begin{eqnarray*}
H_{i1}^R(s^*,s|\phi_0) & = & \left[
\frac{\stackrel{\bullet}{\kappa}_{iq}}{\kappa_{iq}}(s|\phi_0) -
E_q^R(s^*,\mathcal{E}_{iq}(s)|\phi_0), q \in \IQ\right]\trp; \\
H_{i2}^R(s^*,s|\phi_0) & = & \DG\left\{
\frac{I\{\frac{1}{n} S_q^{0R}(s^*,\mathcal{E}_{iq}(s)|\phi_0) > 0\}}{\frac{1}{n} S_q^{0R}(s^*,\mathcal{E}_{iq}(s)|\phi_0)},
q \in \IQ\right\}.
\end{eqnarray*}

\begin{lemma}
\label{lemm: limit of Z-process}
Under the ``RCR-Component Regularity Conditions", the process $\{Z^R(s^*,t): t \in [0,t^*]\}$ converges weakly on Skorokhod's $\mathfrak{D}^{p+Q}[0,\tau^*]$-space, where $p = \dim(\phi)$, to a zero-mean multivariate Gaussian process 
$$\{Z_\infty^R(s^*,t): t \in [0,t^*]\} = \left\{\left[
\begin{array}{c}
Z_{1\infty}^R(s^*,t) \\
Z_{2\infty}^R(s^*,t)
\end{array}
\right]: t \in [0,t^*]
\right\}$$
whose covariance kernels are, for $t_1, t_2 \in [0,t^*]$,
\begin{eqnarray*}
\lefteqn{Cov\{Z_{1\infty}^R(s^*,t_1),Z_{1\infty}^R(s^*,t_2)\}  = } \\
&& \sum_{q=1}^Q \int_0^{\min(t_1,t_2)} v_q^R(s^*,w|\phi_0) s_q^{0R}(s^*,w|\phi_0)
\Lambda_{0q}^0(dw);  \\
&& Cov\{Z_{1\infty}^R(s^*,t_1),Z_{2\infty}^R(s^*,t_2)\}  =  0;  \\
&& Cov\{Z_{2\infty}^R(s^*,t_1),Z_{2\infty}^R(s^*,t_2)\}  =  \DG\left[\int_0^{\min(t_1,t_2)}
\frac{\Lambda_{0q}^0(dw)}{s_q^{0R}(s^*,w|\phi_0)}, q \in \IQ\right]. 
\end{eqnarray*}
In particular, note that $$Cov\{Z_{1\infty}^R(s^*,t^*),Z_{1\infty}^R(s^*,t^*)\} = \Sigma^R(s^*,t^*)$$ and $Z_{1\infty}^R$ and $Z_{2\infty}^R$ are independent.
\end{lemma}

\begin{proof}
We employ the Wold device. Let $c_1 \in \Re^p$ and $c_2 \in \Re^Q$ with $(c_1\trp,c_2\trp) \ne 0$. Let $t \in [0,t^*]$. Then
\begin{eqnarray*}
\lefteqn{c_1\trp Z_1^R(s^*,t) + c_2\trp Z_2^R(s^*,t) } \\
& = & \frac{1}{\sqrt{n}} \sum_{i=1}^n \int_0^{s^*}
\left[c_1\trp H_{i1}^R(s^*,s|\phi_0) + c_2\trp H_{i2}^R(s^*,s|\phi_0)\right] \times \\ &&
\DG\left[I\{\mathcal{E}_{iq}(s) \le t\}, q\in\IQ\right] \left[M_i^R(ds;q), q\in\IQ\right] \\
& = & \frac{1}{\sqrt{n}} \sum_{q=1}^Q \sum_{i=1}^n \int_0^{s^*}
\left[c_1\trp H_{iq1}^R(s^*,s|\phi_0) + c_2\trp H_{iq2}^R(s^*,s|\phi_0)\right] \times \\ &&
I\{\mathcal{E}_{iq}(s) \le t\} M_i^R(ds;q).
\end{eqnarray*}
By the Martingale Central Limit Theorem, this converges in distribution to a zero-mean normal variable with variance being the in-probability limit of
\begin{eqnarray*}
\lefteqn{\frac{1}{n}\sum_{q=1}^Q \int_0^{s^*} 
\left[c_1\trp H_{iq1}^R(s^*,s|\phi_0) + c_2\trp H_{iq2}^R(s^*,s|\phi_0)\right] ^2
I\{\mathcal{E}_{iq}(s) \le t\} A_i^R(ds;q) } \\
& = & \sum_{q=1}^Q \frac{1}{n} \sum_{i=1}^n \int_0^{s^*} \left[c_1\trp H_{iq1}^R(s^*,s|\phi_0)^{\otimes 2} c_1\right] I\{\mathcal{E}_{iq}(s) \le t\} A_i^R(ds;q) + \\
&& 2 \sum_{q=1}^Q \frac{1}{n} \sum_{i=1}^n \int_0^{s^*} \left[c_1\trp H_{iq1}^R(s^*,s|\phi_0) H_{iq2}^R(s^*,s|\phi_0)\trp c_2\right] \times \\ && I\{\mathcal{E}_{iq}(s) \le t\} A_i^R(ds;q) + \\
&& \sum_{q=1}^Q \frac{1}{n} \sum_{i=1}^n \int_0^{s^*} \left[c_2\trp H_{iq2}^R(s^*,s|\phi_0)^{\otimes 2} c_2\right] I\{\mathcal{E}_{iq}(s) \le t\} A_i^R(ds;q).
\end{eqnarray*}
Using earlier results, the first-term converges in probability to $c_1\trp \Sigma^R(s^*,t) c_1$; the second-term equals 0; and the third-term converges in probability to $$c_2\trp \DG\left[\int_0^t \frac{\Lambda_{0q}^0(dw)}{s_q^{0R}(s^*,w|\phi_0)}, q\in\IQ\right] c_2.$$
As such, for each $t \in [0,t^*]$, $$c_1\trp Z_1^R(s^*,t) + c_2\trp Z_2^R(s^*,t) \stackrel{d}{\rightarrow} c_1\trp Z_{1\infty}^R(s^*,t) + c_2\trp Z_{2\infty}^R(s^*,t)$$ for every $(c_1\trp,c_2\trp) \ne 0$.

For $t_1, t_2 \in [0,t^*]$, joint convergence in distribution to a zero-mean bivariate normal distribution with the covariance matrix specified in the statement of the lemma is established similarly and by noting that
$$I\{\mathcal{E}_{iq}(s) \le t_1\} I\{\mathcal{E}_{iq}(s) \le t_2\} = I\{\mathcal{E}_{iq}(s) \le \min(t_1,t_2)\}.$$
So the finite-dimensional distributions of $\{Z^R(s^*,t): t \in [0,t^*]\}$ converges in distribution to zero-mean multivariate normal distributions with the covariance matrix specified in the statement of the lemma.

To upgrade convergence to weak convergence on Skorokhod's $\mathfrak{D}^{p+Q}[0,t^*]$-space, tightness needs to be established. This could be done using similar arguments in proving the Master Theorem presented in \cite{pena2000}, with this theorem used in \cite{pena2001nonparametric}.
\end{proof}

Establishing the results in Theorem \ref{theo: asymptotics of LambdaHats} now immediately follows from these lemmas.

\begin{proof}[Proof of Theorem \ref{theo: asymptotics of LambdaHats}]
The consistency result in Theorem \ref{theo: asymptotics of LambdaHats} follows from Lemma \ref{lemm: for consistency}.
The weak convergence result then follows by noting that
$$\sqrt{n}[\hat{\phi} - \phi_0] =  [\Sigma^R(s^*,t^*)]^{-1} Z_1^R(s^*,t^*) + o_p(1)$$
and
\begin{eqnarray*}
\lefteqn{\sqrt{n}[\hat{\Lambda}_{0q}(s^*,t) - \Lambda_{0q}^*(s^*,t)] = Z_2^R(s^*,t) + } \\
&&  \left[
\int_0^t I\{S_q^{0R}(s^*,w|\phi_0) > 0\} E_q^R(s^*,w|\phi_0) \Lambda_{0q}^0(dw) \right]\trp \\ 
&& [\Sigma^R(s^*,t^*)]^{-1} Z_1^R(s^*,t^*)  + o_p(1)
\end{eqnarray*}
and the fact that, uniformly in $t \in [0,t^*]$, 
$$\int_0^t I\{S_q^{0R}(s^*,w|\phi_0) > 0\} E_q^R(s^*,w|\phi_0) \Lambda_{0q}^0(dw)  \stackrel{p}{\rightarrow} 
\int_0^t e_q^R(s^*,w|\phi_0) \Lambda_{0q}^0(dw).$$
\end{proof}

We point out that the $\hat{\Lambda}_{0q}, q=1,\ldots,Q,$ are asymptotically dependent, and are also not independent of $\hat{\phi}(s^*,t^*)$. The asymptotic covariance functions could be estimated consistently by substituting their empirical counterparts. Observe that the results stated above generalize those in \cite{AndGil1982} and \cite{pena2016} to more complex and general situations.

\section{Illustration of Estimation Approach on Simulated Data}
\label{sec-Illustration}

In this section we provide a numerical illustration of the estimation procedure when given a sample data. We generated a sample of $n = 100$ units from the model under a set of parameter values described below. The data generation and the estimation procedure were implemented using {\tt R} and {\tt Fortran} codes that we have developed.

\subsection{Sample Data Generation}
For the $i$th unit, the covariate values are generated according to $X_{i1}\ {\sim}\ \BER(0.5)$, $X_{i2}\ {\sim}\ N(0, 1)$, and $X_{i3} \ {\sim}\ N(0, 1)$ with $(X_{i1}, X_{i2}, X_{i3})$ mutually independent, where $\BER(p)$ is the Bernoulli distribution with success probability of $p$. The end of monitoring time is $\tau_i\ {\sim}\ \mbox{U}[5,15]$, where $\mbox{U}[a,b]$ is the uniform distribution on $[a,b]$. For the RCR component with $Q=4$, the baseline (crude) hazard rate function for risk $q \in \{1,2,3,4\}$, is a two-parameter Weibull given by 
$$\lambda_{0q}(t|\kappa_{0q},\theta_{0q})=(\kappa_{0q} \theta_{0q}) (\theta_{0q} t)^{\kappa_{0q}-1},\quad t \ge 0,$$ 
with $\kappa_{0q} \in \{2, 1.5, 1, .5\}$ and $\theta_{0q} \in \{0.11, .20, .05, .15\}$, respectively. The associated (crude) hazard function and survivor function for risk $q$ are, respectively,
$$\Lambda_{0q}(t|\kappa_{0q},\theta_{0q}) = (\theta_{0q} t)^{\kappa_{0q}}\ \mbox{and}\ \bar{F}_{0q}(t|\kappa_{0q},\theta_{0q}) = \exp\left\{-(\theta_{0q} t)^{\kappa_{0q}}\right\},\quad t \ge 0.$$
For risk $q$, the effective age process function is $$\mathcal{E}_{iq}(s) = s - S_{iqN_{iq}^R(s-)}^R,$$ the backward recurrence time for this risk. For the effects of the accumulating event occurrences, for each $q = 1, 2, 3, 4,$ we set, using the re-labeling $\alpha_{0q} = \theta_{0q}^{RR} \equiv  (\theta_{0qq_1}^{RR}, q_1 = 1,2,3,4)$,
\begin{eqnarray*}
\lefteqn{\rho_q(N^R(s-)|\alpha_{0q})  =  \exp\left\{[1 + N^R(s-)]\trp \theta_{0q}^{RR}\right\} } \\
& = & \exp\left\{\sum_{q_1=1}^Q \theta_{0qq_1}^{RR} [\log(1+N_{q_1}^R(s-))] \right\}
= \prod_{q_1 = 1}^Q \left[1 + N_{q_1}^R(s-)\right]^{\theta_{0qq_1}^{RR}}.
\end{eqnarray*}
The values for these $\theta_{0qq_1}^{RR}$ are given in the table
\begin{center}
\begin{tabular}{|c|c|c|c|c|} \hline
 $\theta_{0qq_1}^{RR}$ & \multicolumn{4}{c|}{$q_1$} \\ \hline
 $q$   &  1 & 2 & 3 & 4 \\ \hline
1 & 0.10 & 0.2 & 0.05 & 0.00 \\
2 & -0.05 & 0.3 & 0.00 & 0.05 \\
3 & 0.00 & 0.2 & -0.10 & 0.06 \\
4 & 0.00 & 0.2 & -0.10 & -0.04 \\ \hline
\end{tabular}
\end{center}

For the HS component, $\mathfrak{V} = \{1,2,3,4\}$ with state `$4$' an absorbing state, so $\gamma_0$ is a $2 \times 1$ vector. For the LM component, $\mathfrak{W} = \{1, 2, 3, 4\}$, so $\kappa_0$ is a $3 \times 1$ vector. 
The infinitesimal generator matrices $\eta_0$ for the LM process and $\xi_0$ for the HS process are, respectively,
\begin{eqnarray*}
\eta_0=\begin{bmatrix}
\eta_{0}(w_1w_2) & 1 & 2 & 3 & 4 \\ 
1 &  -0.80 & 0.5 & 0.25 & 0.05 \\
2 & 0.25 & -0.6 & 0.25 & 0.10 \\
3 & 0.40 & 0.1 & -1.00 & 0.50 \\
4  & 0.20 & 0.1 & 0.60 & -0.90
\end{bmatrix}
\end{eqnarray*}
and
 %\quad \mbox{and}  \quad
 \begin{eqnarray*}
\xi_0=\begin{bmatrix}
  \xi_0(v_1,v_2) &  1 & 2 & 3 &4 \\
1 & -1.0 & 0.5 & 0.46 & 0.04 \\
2 & 0.6  & -0.9 & 0.25 & 0.05 \\
3 & 0.5 & 0.4 & -1.00 & 0.10 \\
4 & 0.0 & 0.0 & 0.00 & 0.00
\end{bmatrix}.
\end{eqnarray*}
The values in the fourth row for the $\xi_0$-matrix are all zeros because state $4$ in HS process is absorbing. The true values of the remaining model parameters are: %given in the second column of Table \ref{oneEs} {in the online Supplementary Materials}.
$$\theta_0^{RL} = (-1,  0,  1)\trp;  \theta_0^{RH} = (1.0, 0.5)\trp; \theta_0^{RX} = (0.5, -0.3, -0.1)\trp;$$
$$\theta_0^{LR} = (-0.5, -0.4,  0.7,  0.3)\trp; \theta_0^{LH} = (1, -1)\trp;  \theta_0^{LX} = (0.2, -0.4,  0.1)\trp;$$
$$\theta_0^{HR} = (-0.2,  0.5, -0.3,  1.0)\trp; \theta_0^{HL} = (0.5,  0.2, -1.0)\trp; \theta_0^{HX} = (-0.5,  0.4, -0.1)\trp.$$

In generating the simulated data for each unit in the sample, we use the fact that between two successive event times, where the event could be of any type (i.e., RCR-, LM-, or HS-type), the values of the internal and external covariates do not change. As such, if we have reached event time $S_{k}$, to generate the {\em potentially} next event time, $S_{k+1}$, for the $q$th risk in the RCR, we generate the {\em residual} time $T_q^R$ associated with a conditional Weibull distribution, given that the lifetime exceeded $\mathcal{E}_q(S_k)$; whereas, in the LM and HS components, we generate times $T^L$ and $T^H$ from exponential distributions. See below for specific details in the generation of $T_q^R$, $T^L$, and $T^H$. The potential next event time $S_{k+1}$ becomes $S_k + \min\{T_q^R, q=1:Q; T^L, T^H\}$. However, if this value exceeds $\tau$, then right-censoring occurs and we set $S_{k+1} = \tau$ and data generation for the unit ceases.

Just after $S_k$, for $q = 1:Q$, the {\em residual time} $T_q^R$ is generated from a conditional Weibull distribution with shape and scale parameters 
$$\left(\kappa_{0q},\theta_{0q}(S_k) = \theta_{0q} \exp\left\{\frac{\phi^R(S_k)}{\kappa_{0q}}\right\}\right),$$ 
given that the Weibull random variable is {\em at least} $\mathcal{E}_q(S_k)$. Here, 
$$\mathcal{E}_q(s) = s - S_{qN^R(s-;q)}^R;$$
$$\phi^R(S_k) \equiv [\log(1+N^R(S_k))] \theta_{0q}^{RR} + \iota_{\mathfrak{W}}(W(S_k)) \theta_{0}^{RL} + \iota_{\mathfrak{V}}(V(S_k)) \theta_0^{RH} + x \theta_0^{RX}.$$
The realization of $T_q^R$ could be obtained through the formula below, which is derived using the Probability Integral Transformation.
\begin{equation}
\label{general weibull residual time}
T_q^R = \left\{\frac{1}{\theta_{0q}} \left[ \left(\theta_{0q} \mathcal{E}_q(S_k) \right)^{\kappa_{0q}} -
\frac{\log(1 - U_q^R)}{\exp(\phi^R(S_k))} \right]^{1/\kappa_{0q}}\right\} -  \mathcal{E}_q(S_k).
\end{equation}
where $U_q^R, q=1:Q,$ are independent standard uniform random variables. This formula is a special case of the more general result:

\begin{proposition}
Let $Z$ be a nonnegative-valued continuous random variable with  hazard rate function $\lambda_0(z) \phi$, where $\phi$ is a positive constant and $\Lambda_0(z) = \int_0^z \lambda_0(v) dv$ has inverse $\Lambda_0^{-1}(y)$. Let $U$ be a standard uniform random variable. Then, for fixed $s \ge 0$, the random variable
$$T = \Lambda_0^{-1}\left\{\Lambda_0(s) - \frac{\log(1-U)}{\phi}\right\} - s$$
has survivor function $t \mapsto \Pr\{Z > s + t | Z > s\}$ for $t \ge 0$.
\end{proposition}

Just after $S_k$, $T^L$ is generated from an exponential distribution with scale parameter $-\eta[W(S_k),W(S_k)] \phi^L(S_k),$ where
$$\phi^L(S_k) \equiv  [\log(1+N^R(S_k))] \theta_{0q}^{LR} +  \iota_{\mathfrak{V}}(V(S_k)) \theta_0^{LH} + x \theta_0^{LX};$$
while $T^H$ is generated from an exponential distribution with scale parameter $-\xi[V(S_k),V(S_k)] \phi^H(S_k),$ where
$$\phi^H(S_k) \equiv  [\log(1+N^R(S_k))] \theta_{0q}^{HR} +  \iota_{\mathfrak{W}}(W(S_k)) \theta_0^{HL} + x \theta_0^{HX}.$$

For our illustrative sample data, we generated one with $n=100$.  A randomly chosen portion (30\% of the 100 units) of this generated sample data is depicted in Figure \ref{sample data picture} on page \pageref{sample data picture} (note that we did not plot all 100 units since the plot would otherwise become too busy). Basic summary statistics for this sample data are as follows. The mean of the observation windows was $\bar{\tau} = 9.90$ with standard deviation of $\mbox{sd}(\tau) =  2.63$. Note, however, that the {\em effective} observation windows will be shorter due to absorptions in the HS processes. The observed number of event occurrences for each of the four risks are, respectively, $(68, 187, 74, 486)$. The number of transitions in the LM and HS components are summarized in the two tables in Table \ref{table-LM and HS transitions}. Note that state 4 in the HS-process is absorbing, hence zero transitions from this state.

\begin{table}[h]
\caption{Observed transitions in the LM and the HS processes for the illustrative sample data.}
\label{table-LM and HS transitions}
\begin{minipage}[t]{.45\textwidth}
\begin{center}
\begin{tabular}{||c||c|c|c|c||} \hline\hline
LM & 1 & 2 & 3 & 4 \\ \hline\hline
1 & - & 130 & 62 & 10 \\
2 & 78 & - & 73 & 19 \\
3 & 103 & 27 & - & 143 \\
4 & 25 & 23 & 136 & -  \\ \hline\hline
\end{tabular}
\end{center}
\end{minipage}
\begin{minipage}[t]{.45\textwidth}
%\begin{table}[ht]
%\caption{Observed transitions in the HS-process.}
%\label{table-HS transitions}
\begin{center}
\begin{tabular}{||c||c|c|c|c||} \hline\hline
HS & 1 & 2 & 3 & 4 \\ \hline\hline
1 & - & 303 & 263 & 19 \\
2 & 332 & - & 136 & 22 \\
3 & 223 & 162 & - & 39 \\
4 & 0 & 0 & 0 & -  \\ \hline\hline
\end{tabular}
\end{center}
\end{minipage}
\end{table}

\subsection{Parameter Estimates for Illustrative Sample Data}

For this illustrative sample data, we implemented the semi-parametric estimation procedure described in the preceding sections using {\tt R} and {\tt Fortran} programs that we developed. {\tt Fortran} subroutines were called by the {\tt R} programs to speed-up the computations of the estimates, especially for the RCR component. The maximizations of the partial likelihoods were through Newton-Raphson iteration. Summary of the estimates, together with their estimated standard errors (both uncorrected and corrected for the $\eta$ and $\xi$ parameters) are contained in columns labeled `For Illustrative Data' in Tables \ref{table-estimates RCR sample}, \ref{table-estimates LM sample theta}, \ref{table-estimates LM sample eta}, \ref{table-estimates HS sample theta}, and \ref{table-estimates HS sample xi}. The second column in each of these tables contains the true parameter values.
For the estimates of the infinitesimal generators $\eta$s and $\xi$s, as well as the survivor functions estimates, the corrections to the standard errors owing to the estimation of the regression coefficients (the $\theta$s) appear not negligible.  Note that each of the true parameter values of the finite-dimensional parameters falls in their respective $[\mbox{Estimate} \pm 2(\mbox{Est. SE})]$ confidence interval.

\begin{table}
\caption{Results of estimation of the $\theta$-parameters in the RCR component for the illustrative data and in the simulation study. Mean and standard deviation of the estimates in the simulation study are reported. The simulation study had {100} replications.}
\label{table-estimates RCR sample}
\begin{center}
\begin{tabular}{|c|c||c|c||c|c||c|c||} \hline\hline
&& \multicolumn{2}{c||}{For Illustrative Data} & \multicolumn{4}{c||}{In Simulation Study} \\
\cline{3-8} 
& True & \multicolumn{2}{c||}{$n=100$} & \multicolumn{2}{c||}{$n=50$} & \multicolumn{2}{c||}{$n=100$} \\
\cline{3-8}
Parameter & Value & EST &  ESE & Mean & SD & Mean & SD \\ \hline\hline
$\theta_{11}^{RR}$  & 0.10 & 0.491     &     0.423 & -0.097   &        0.681& 0.162 & 0.408  \\ 
$\theta_{12}^{RR}$      &0.20 &         0.376     &     0.293 & 0.306         &       0.424 & 0.213 & 0.317  \\ 
$\theta_{13}^{RR}$    &0.05 &          -0.019    &      0.289 & 0.038      &        0.538& 0.021 & 0.323 \\	
$\theta_{14}^{RR} $  &0.00 &           -0.175     &     0.160	& 0.008      &         0.255& -0.014 & 0.180 \\
$\theta_{21}^{RR} $ &-0.05&             0.068     &     0.220	& 0.032        &          0.366& -0.023 & 0.210 \\
$\theta_{22}^{RR} $ &0.30&              0.264   &       0.178	&  0.303       &        0.282& 0.255 & 0.168 \\
$\theta_{23}^{RR} $ &0.00&              0.060     &     0.181	& -0.049      &         0.293& 0.011 & 0.205 \\
$\theta_{24}^{RR} $ &0.05&             -0.063    &      0.097	& 0.048        &       0.144& 0.064 & 0.097 \\
$\theta_{31}^{RR} $ &0.00&             -0.641    &      0.411	& 0.015       &      0.621&  -0.064 & 0.376 \\
$\theta_{32}^{RR} $ &0.20&             0.634    &      0.290	& 0.182       &       0.456& 0.234 & 0.280 \\
$\theta_{33}^{RR}$ &-0.10&             -0.398    &      0.366	& -0.251         &      0.624& -0.197 & 0.382 \\
$\theta_{34}^{RR}$ &0.06&             -0.140    &      0.152	& 0.096           &       0.229& 0.098 & 0.163 \\
$\theta_{41}^{RR} $ &0.00&            -0.006     &     0.156	& 0.033          &       0.217& 0.055 & 0.149 \\
$\theta_{42}^{RR} $ &0.20&             0.045     &     0.118	& 0.260      &       0.155& 0.237 & 0.111 \\
$\theta_{43}^{RR} $ &-0.10&             0.181    &      0.127	& -0.096    &       0.189& -0.119 & 0.136 \\
$\theta_{44}^{RR}$ &-0.04&            -0.071   &       0.072	& 0.023      &     0.100& 0.022 &  0.065 \\ \hline
$\theta_{2}^{RL} $ &-1.00&            -0.925   &       0.149	& -1.054      &     0.221& -1.033 & 0.160 \\
$\theta_{3}^{RL} $ &0.00&             0.048    &      0.112	& -0.034     &     0.155& -0.017 & 0.119 \\
$\theta_{4}^{RL} $ &1.00&            1.016     &     0.102	& 1.015      &      0.135& 0.999 & 0.094 \\ \hline
$\theta_{2}^{RH} $ &1.00&             0.959   &       0.096	& 1.006     &     0.140& 0.999 & 0.093\\
$\theta_{3}^{RH} $ &0.50&             0.538   &       0.104	 & 0.480     &      0.152& 0.510 & 0.105 \\ \hline
$\theta_{1}^{RX} $ 	&0.50&            0.497   &       0.089 & 0.503         &     0.113& 0.485 & 0.084 \\
$\theta_{2}^{RX} $ &-0.30&            -0.338   &       0.047	& -0.297     &       0.050& -0.292 & 0.041 \\
$\theta_{3}^{RX}$ &-0.10&             -0.158    &      0.043	& -0.083     &      0.053& -0.099 & 0.041 \\
\hline\hline
\end{tabular}
\end{center}
\end{table}

\begin{table}
\caption{Results of estimation of the $\theta$-parameters in the LM component for the illustrative data and in the simulation study. Mean and standard deviation of the estimates in the simulation study are reported. The simulation study had {100} replications.}
\label{table-estimates LM sample theta}
\begin{center}
\begin{tabular}{|c|c||c|c||c|c||c|c||} \hline\hline
&& \multicolumn{2}{c||}{For Illustrative Data} & \multicolumn{4}{c||}{In Simulation Study} \\
\cline{3-8} 
& True & \multicolumn{2}{c||}{$n=100$} & \multicolumn{2}{c||}{$n=50$} & \multicolumn{2}{c||}{$n=100$} \\
\cline{3-8}
Parameter & Value & EST &  ESE & Mean & SD & Mean & SD \\ \hline\hline
$\theta_{1}^{LR}$ & -0.50&   -0.605  &     0.123	& -0.476 & 0.195 & -0.498 & 0.118 \\
$\theta_{2}^{LR}$  & -0.40& -0.384   &    0.082	& -0.419 & 0.132 & -0.394 & 0.082 \\
$\theta_{3}^{LR}$   & 0.70& 0.627    &   0.094	& 0.694 & 0.139 & 0.699 & 0.088 \\
$\theta_{4}^{LR}$  & 0.30&  0.371   &    0.049	& 0.298 & 0.077 & 0.297 & 0.057 \\ \hline
$\theta_{2}^{LH}$  & 1.00&   0.954  &     0.079	& 1.011 & 0.112 & 1.002 & 0.085 \\
$\theta_{3}^{LH}$  & -1.00& -1.012   &    0.142	& -1.012 & 0.194 & -1.014 & 0.151 \\ \hline
$\theta_{1}^{LX}$  &	0.20&  0.194    &   0.085 & 0.211 & 0.129 & 0.191 & 0.091 \\
$\theta_{2}^{LX}$  &	-0.40& -0.342   &    0.048 & -0.400 & 0.068 & -0.400 & 0.041 \\
$\theta_{3}^{LX}$   &	0.10&  0.157   &    0.043  & 0.099 & 0.055 & 0.102 & 0.035 \\ \hline\hline
\end{tabular}
\end{center}
\end{table}

\begin{table}
\caption{Results of estimation of the $\eta$-parameter in the LM component for the illustrative data and in the simulation study. Mean and standard deviation of the estimates in the simulation study are reported. The simulation study had {100} replications. {\bf Legend:} EUSE = estimate of the uncorrected standard error; ECSE = estimate of the corrected standard error.}
\label{table-estimates LM sample eta}
\begin{center}
\begin{tabular}{|c|c||c|c|c||c|c||c|c||} \hline\hline
&& \multicolumn{3}{c||}{For Illustrative Data} & \multicolumn{4}{c||}{In Simulation Study} \\
\cline{3-9} 
& True & \multicolumn{3}{c||}{$n=100$} & \multicolumn{2}{c||}{$n=50$} & \multicolumn{2}{c||}{$n=100$} \\
\cline{3-9}
Parameter & Value & EST & EUSE & ECSE & Mean & SD & Mean & SD \\ \hline\hline
$\eta_{12}$ & 0.50&          0.523    &      0.046    &    0.064 & 0.500 & 0.085 & 0.503 & 0.062	 \\
$\eta_{13}$   & 	0.25 &      0.250     &     0.032    &    0.038 & 0.253 & 0.047 & 0.255 & 0.043  \\
$\eta_{14}$  &	0.05   &   0.040    &     0.013    &    0.013 & 0.051 & 0.024 & 0.049 & 0.015 \\
$\eta_{21}$  & 0.25 &      0.248     &     0.028   &     0.035 & 0.248 & 0.053 & 0.255 &	0.039 \\
$\eta_{23}$    &	0.25&      0.232     &     0.027   &     0.033 & 0.248 & 0.052 & 0.249 & 0.037  \\
$\eta_{24}$    &	0.10 &     0.060     &     0.014    &    0.015 & 0.100 & 0.028 & 0.100 & 0.022  \\
$\eta_{31}$     & 0.40 &    0.361    &      0.036   &    0.047 & 0.402 & 0.082 & 0.401 & 0.044	 \\
$\eta_{32}$    & 0.10&      0.095   &       0.018    &    0.020 & 0.105 & 0.035 & 0.101 & 0.023	 \\
$\eta_{34}$     & 0.50 &    0.501   &       0.042    &    0.060 & 0.506 & 0.100 & 0.508 & 0.061	 \\
$\eta_{41}$    &	0.20&     0.126    &      0.025    &    0.028 & 0.208 & 0.062 & 0.199 & 0.036  \\
$\eta_{42}$    &	0.10&     0.116    &      0.024    &    0.026 & 0.100 & 0.040 & 0.097 & 0.025  \\
$\eta_{43}$   &	0.60&      0.686   &       0.059    &    0.085 & 0.611 & 0.117 & 0.599 & 0.075 \\ \hline\hline
\end{tabular}
\end{center}
\end{table}

\begin{table}
\caption{Results of estimation of the $\theta$-parameters in the HS component for the illustrative data and in the simulation study. Mean and standard deviation of the estimates in the simulation study are reported. The simulation study had {100} replications.}
\label{table-estimates HS sample theta}
\begin{center}
\begin{tabular}{|c|c||c|c||c|c||c|c||} \hline\hline
&& \multicolumn{2}{c||}{For Illustrative Data} & \multicolumn{4}{c||}{In Simulation Study} \\
\cline{3-8} 
& True & \multicolumn{2}{c||}{$n=100$} & \multicolumn{2}{c||}{$n=50$} & \multicolumn{2}{c||}{$n=100$} \\
\cline{3-8}
Parameter & Value & EST &  ESE & Mean & SD & Mean & SD \\ \hline\hline
$\theta_1^{HR}$ & -0.20&    -0.157   &    0.084 & -0.207 & 0.128 & -0.200 & 0.084	 \\
$\theta_2^{HR}$  &	0.50 &  0.477  &     0.061 & 0.498 & 0.085 & 0.511 & 0.058  \\
$\theta_3^{HR}$  &	-0.30&  -0.411  &     0.080 & -0.301 & 0.102 & -0.312 & 0.072  \\
$\theta_4^{HR}$  &	1.00 &  1.014  &     0.041 & 1.005 & 0.066 & 1.002 & 0.037  \\ \hline
$\theta_2^{HL}$ &	0.50&    0.380  &     0.066 & 0.503 & 0.121 & 0.498 & 0.072  \\
$\theta_3^{HL}$  &	0.20&   0.177   &    0.073 & 0.211 & 0.117 & 0.197 &  0.077 \\
$\theta_4^{HL}$   &	-1.00& -1.039   &    0.110 & -0.985 & 0.153 & -1.008 & 0.107  \\
$\theta_1^{HX}$    &	-0.50& -0.451   &    0.064 & -0.499 & 0.099 & -0.498 & 0.063 \\
$\theta_2^{HX}$   &	0.40&  0.419   &    0.037 & 0.393 & 0.049 & 0.404 & 0.033 \\
$\theta_3^{HX}$  & -0.10& -0.137  &     0.028 & -0.108 & 0.043 &-0.104 & 0.026 \\ \hline\hline
\end{tabular}
\end{center}
\end{table}

\begin{table}
\caption{Results of estimation of the $\eta$-parameter in the LM component for the illustrative data and in the simulation study. Mean and standard deviation of the estimates in the simulation study are reported. The simulation study had {100} replications. {\bf Legend:} EUSE = estimate of the uncorrected standard error; ECSE = estimate of the corrected standard error.}
\label{table-estimates HS sample xi}
\begin{center}
\begin{tabular}{|c|c||c|c|c||c|c||c|c||} \hline\hline
&& \multicolumn{3}{c||}{For Illustrative Data} & \multicolumn{4}{c||}{In Simulation Study} \\
\cline{3-9} 
& True & \multicolumn{3}{c||}{$n=100$} & \multicolumn{2}{c||}{$n=50$} & \multicolumn{2}{c||}{$n=100$} \\
\cline{3-9}
Parameter & Value & EST & EUSE & ECSE & Mean & SD & Mean & SD \\ \hline\hline
$\xi_{12}$ &	0.50&        0.574     &     0.033   &     0.053  & 0.504  & 0.080 & 0.500 & 0.044 \\
$\xi_{13}$  &	0.46   &    0.498     &     0.031   &     0.047 & 0.451 & 0.068 & 0.457 & 0.041 \\
$\xi_{14}$   &	0.04  &    0.036      &    0.008    &    0.009 & 0.042 & 0.013 & 0.041 & 0.010 \\
$\xi_{21}$   &	0.60&      0.659    &      0.036     &   0.060 & 0.599 & 0.088 & 0.600 & 0.056 \\
$\xi_{23}$     &	0.25&    0.270     &     0.023   &     0.030 & 0.253 & 0.042 & 0.249 & 0.028 \\
$\xi_{24}$    &	0.05&     0.044    &      0.009   &     0.010 & 0.050 & 0.017 & 0.053 & 0.012 \\
$\xi_{31}$   &	0.50&      0.574    &      0.038  &      0.057 & 0.500 & 0.079 & 0.503 & 0.050 \\
$\xi_{32}$      &	0.40&   0.417   &       0.033   &     0.045  & 0.400 & 0.067 & 0.392 & 0.041 \\
$\xi_{34}$   &	0.10&      0.100    &      0.016   &     0.018  & 0.104 & 0.026 & 0.101 & 0.019  \\ \hline\hline
\end{tabular}
\end{center}
\end{table}

Finally, Figure \ref{figure-survivor estimates sample} presents the estimates of the baseline survivor functions for the four competing risks based on the illustrative sample data, together with 95\% approximate point-wise confidence bands.

\medskip
\centerline{[INSERT FIGURE \ref{figure-survivor estimates sample} HERE]}
\medskip

\section{Finite-Sample Properties via Simulation Studies}
\label{sec: simulation}

\subsection{Simulation Design}

We have provided asymptotic results of the estimators in Section \ref{sec: Properties}. In this section we present the results of simulation studies to assess the finite-sample properties of the estimators of model parameters. This will provide some evidence whether the semi-parametric estimation procedure, which appears to perform satisfactorily for the single illustrative sample data set in Section \ref{sec-Illustration}, performs satisfactorily over many sample data sets.
These simulation studies were implemented using the {\tt R} and {\tt Fortran} programs we developed. In these simulation studies, as in the preceding section, when we analyze each of the sample data, the baseline hazard rate functions are estimated {\em semi-parametrically}, even though in the generation of each of the sample data sets, two-parameter Weibull models were used in the RCR components. That is, in the estimation procedure, we totally ignored the knowledge that the baseline hazard functions are of the Weibull-type with two parameters.
Aside from the set of model parameters described in Section \ref{sec-Illustration}, the simulation study have the additional inputs which are the sample size $n$ and the number of simulation replications \Mreps, the latter set to {100}. The sample sizes used in the two  simulation studies are $n \in \{50, 100\}$.
For fixed $n$, for each of the \Mreps\ replications, the sample data generation is as described in Section \ref{sec-Illustration}. For each estimator, for each $n$, we obtained the mean of the estimates, as well as the standard deviation of the estimates, which is a simulated estimate of the standard error of the estimator. In each of the summary tables, we also provide the true parameter value, and the mean and standard deviation of the estimates for each $n \in \{50, 100\}$. These summary tables are provided in the last four columns labelled `In Simulation Study' in Tables \ref{table-estimates RCR sample}, \ref{table-estimates LM sample theta}, \ref{table-estimates LM sample eta}, \ref{table-estimates HS sample theta}, and \ref{table-estimates HS sample xi}. We also depict these simulated estimates pictorially in Figure \ref{figure-comparative boxplots of estimates}, which presents comparative boxplots of the {\em centered} simulated estimates, that is, centered in the sense that the true parameter value was subtracted from their estimates, for each of the finite-dimensional parameters, and for $n=50$ and $n=100$. Examination of the histograms, though these are not presented here, of the simulated estimates shows that they tend to be concentrated in the vicinity of their associated true parameter values for the sample sizes $n \in \{50,100\}$ used in the simulation study.

 The {100} baseline survivor estimates for each of the four risks are depicted in Figure \ref{figure-simulated baseline estimates}, together with their respective true baseline survivor function. Observe that for each of these four risks, the estimates of the baseline survivor function clusters around the true baseline survivor function, indicating that the semi-parametric approach to the estimation of model parameters appears to be tenable.

\medskip
\centerline{[INSERT FIGURE \ref{figure-simulated baseline estimates} HERE]}
\medskip

\subsection{Some Limitations}

In implementing the semi-parametric approach to estimating model parameters, there were also some limitations encountered in the computational implementation. Recall that a Newton-Raphson iteration approach was utilized in maximizing the partial likelihood functions. The perennial limitation of finding appropriate seed values for the iteration were also encountered. In experimenting with the algorithmic implementation, we found that sometimes it helps to start with a gradient descent approach in the iteration, and then to later transition to the Newton-Raphson iteration. In addition, when $n$ is not large enough, especially when $Q > 1$ leading to many parameters, the profile likelihoods could be rather flat, hence convergence may be compromised or not occur. We also note that obtaining estimates of the parameters, both finite- and infinite-dimensional, could be computationally time-consuming, hence we were able to run the simulations with just 100 replications. This is even after some portions were converted to {\tt Fortran} codes. One reason for this computational intensiveness is because in estimating the baseline survivor functions, we required the computation of the whole functional estimate, instead of simply focusing on a pre-specified set of time values. We also mention that in this paper we did not apply the model and the estimation procedure to {\em real} data sets. We hope to apply the modeling approach and the estimation procedure to some real data sets in a separate paper.

\section{Concluding Remarks}
\label{sec: conclusion}

For the general class of joint models for recurrent competing risks, longitudinal marker, and health status proposed in this paper, which encompasses many existing models considered previously, there are still numerous aspects that need to be addressed in future studies. Foremost among these aspects is a more refined analytical study of the finite-sample and asymptotic properties of the estimators of model parameters, together with other inferential and prediction procedures. The finite-sample and asymptotic results could be exploited to enable performing tests of hypothesis and construction of confidence regions for model parameters. There is also the interesting aspect of computationally estimating the standard errors of the estimators. How would a bootstrapping approach be implemented in this situation, or would such an approach be prohibitively computationally expensive? Another important problem that needs to be addressed is how to perform goodness-of-fit and model validation for this joint model. Though the class of models is very general, there are still possibilities of model mis-specifications, such as, for example, in specifying the effective age processes, or in the specification of the $\rho_q(\cdot|\cdot)$-functions. What are the impacts of such model mis-specifications? Do they lead to serious biases that could potentially result in misleading conclusions? These are some of the problems whose solutions await further studies.

A potential promise of this joint class of models is in precision medicine. Because all three components (RCR, LM, HS) are taken into account simultaneously, in contrast to a marginal modeling approach, the synergy that this joint model allows may improve decision-making -- for example, in determining interventions to be performed for individual units. In this context, it is of utmost importance to be able to predict in the future the trajectories of the HS process, given information at a given point in time, about all three processes. Thus, an important problem to be dealt with in future work is the problem of forecasting using this joint model. How should such forecasting be implemented? This further leads to other important questions. One is determining the relative importance of each of the components in this prediction problem. Could one ignore other components and still do as well relative to a joint model-based prediction approach? If there are many covariates, how should the important covariates among these numerous covariates be chosen in order to improve prediction of, say, the time-to-absorption?

Finally, though our class of joint models is a natural general extension of earlier models dealing with either recurrent events, competing recurrent events, longitudinal marker, and terminal events, one may impugn it as perhaps not realistic, but instead view it as more of a futuristic class of models, since existing data sets were not gathered (yet) in the manner for which these joint models apply. For instance, in the example pertaining to gout in Section \ref{sec: scenarios}, the SUR level and CKD status are not continuously monitored. However, with the advent of smart devices, such as smart wrist watches, embedded sensors, black boxes, etc., made possible by miniaturized technology, high-speed computing, almost limitless cloud-based memory  capacity, and availability of rich cloud-based databases, the era is, in our opinion, fast approaching when continuous monitoring of longitudinal markers, health status, occurrences of different types of recurrent events, be it on a human being, an experimental animal or plant, a machine such as a car, an airplane, a nuclear power plant, a medical equipment, or a space satellite in engineering settings, a company in a economic or business setting, etc., will become more of a standard rather than an exception. It is our hope that by developing the models and inferential methods for such future complex and advanced data sets, {\em even} before they become available and real, that this will hasten and prepare us all for their eventual and certain arrival.

%%%%%%%%%%%%%%%%%%%%%%%%%%%%%%%%%%%%%%%%%%%%%%
%% Single Appendix:                         %%
%%%%%%%%%%%%%%%%%%%%%%%%%%%%%%%%%%%%%%%%%%%%%%
%\begin{appendix}
%\section*{???}%% if no title is needed, leave empty \section*{}.
%\end{appendix}
%%%%%%%%%%%%%%%%%%%%%%%%%%%%%%%%%%%%%%%%%%%%%%
%% Multiple Appendixes:                     %%
%%%%%%%%%%%%%%%%%%%%%%%%%%%%%%%%%%%%%%%%%%%%%%
%\begin{appendix}
%\section{???}
%
%\section{???}
%
%\end{appendix}

%\section*{Online Supplementary Materials}
%The {Online Supplementary Materials} referenced in some sections is accessible online as a supplement to this paper.

%%%%%%%%%%%%%%%%%%%%%%%%%%%%%%%%%%%%%%%%%%%%%%
%% Support information, if any,             %%
%% should be provided in the                %%
%% Acknowledgements section.                %%
%%%%%%%%%%%%%%%%%%%%%%%%%%%%%%%%%%%%%%%%%%%%%%
\section*{Acknowledgments}

The authors thank the reviewers and editors for their comments which helped in improving the paper.
E.\ Pe\~na was Program Director in the Division of Mathematical Sciences (DMS) at the National Science Foundation (NSF) from 2020--2023. As a consequence, he received support for research, which included work in this manuscript, under NSF Grant 2049691 to the University of South Carolina. Any opinions, findings, and conclusions or recommendations expressed in this material are those of the authors and do not necessarily reflect the views of NSF.

%%%%%%%%%%%%%%%%%%%%%%%%%%%%%%%%%%%%%%%%%%%%%%
%% Funding information, if any,             %%
%% should be provided in the                %%
%% funding section.                         %%
%%%%%%%%%%%%%%%%%%%%%%%%%%%%%%%%%%%%%%%%%%%%%%
\section*{Funding}
E.\ Pe\~na acknowledges NSF Grant 2049691 and NIH Grant P30GM103336-01A1.
P.\ Liu acknowledges the Summer Research Grants from Bentley University. 

%%%%%%%%%%%%%%%%%%%%%%%%%%%%%%%%%%%%%%%%%%%%%%
%% Supplementary Material, including data   %%
%% sets and code, should be provided in     %%
%% {supplement} environment with title      %%
%% and short description. It cannot be      %%
%% available exclusively as external link.  %%
%% All Supplementary Material must be       %%
%% available to the reader on Project       %%
%% Euclid with the published article.       %%
%%%%%%%%%%%%%%%%%%%%%%%%%%%%%%%%%%%%%%%%%%%%%%
%\begin{supplement}
%\stitle{???}
%\sdescription{???.}
%\end{supplement}

%%  MR numbers will be added by VTeX.                      %%
%%                                                         %%
%%  Use \cite{...} to cite references in text.             %%
%%                                                         %%
%%%%%%%%%%%%%%%%%%%%%%%%%%%%%%%%%%%%%%%%%%%%%%%%%%%%%%%%%%%%%

%This is where your bibliography is generated. Make sure that your .bib file is actually called library.bib
\bibliography{biblio}

%This defines the bibliographies style. Search online for a list of available styles.
\bibliographystyle{abbrv}

%% or include bibliography directly:
% \begin{thebibliography}{}
% \bibitem{b1}
% \end{thebibliography}

%\bibliographystyle{ECA_jasa}
%\bibliography{biblio}

%%%%%FIGURES

\begin{figure}[!h]
\begin{center}
\begin{tabular}{cc}
\includegraphics*[width=2.25in,height=5in]{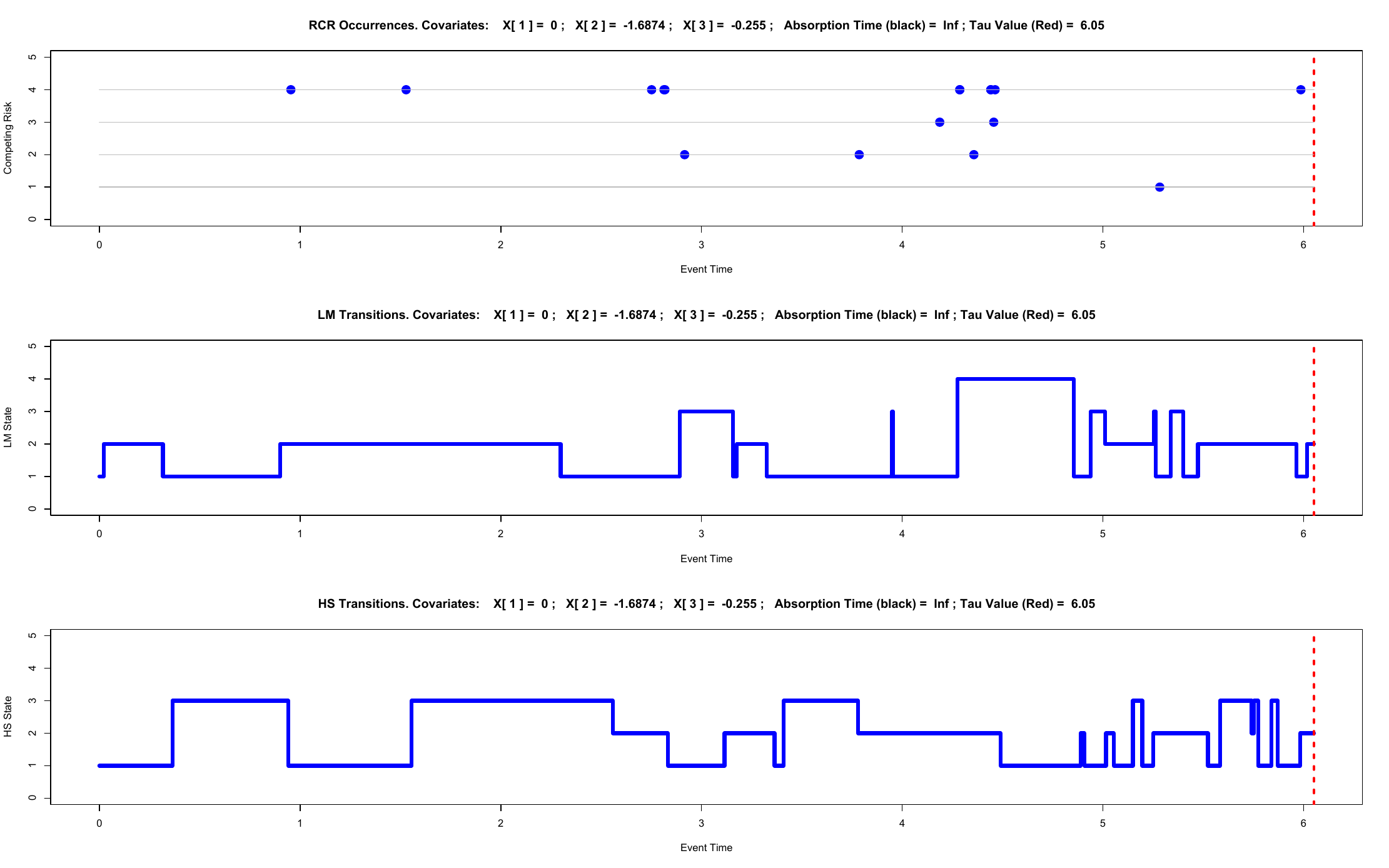} &
\includegraphics*[width=2.25in,height=5in]{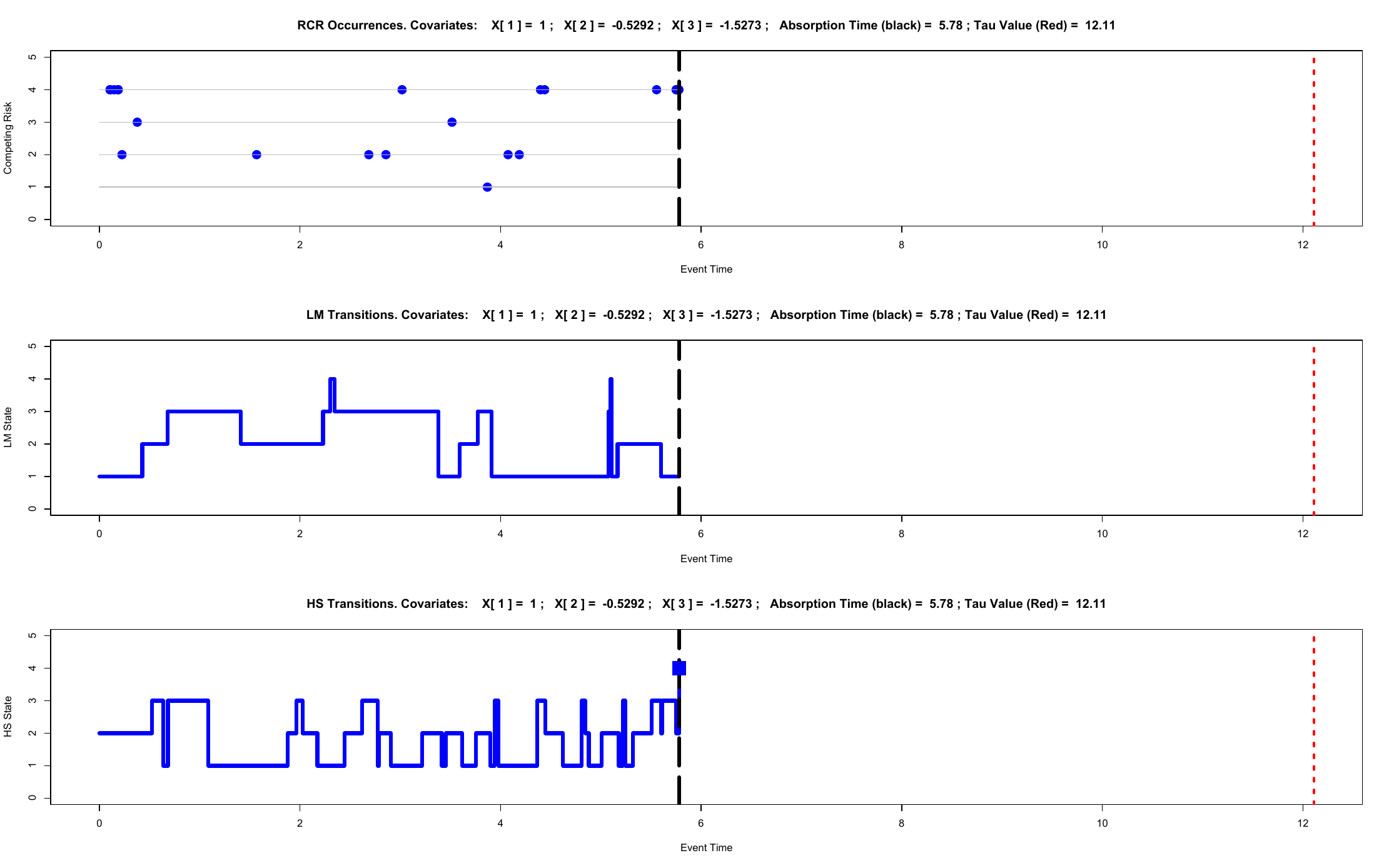} 
\end{tabular}
\end{center}
\caption{Realized (simulated) data observables from two distinct study units. The left plot panel is for a unit which was right-censored prior to reaching an absorbing state, while the right plot panel is for a unit which reached an absorbing state prior to getting right-censored. The red dashed vertical lines correspond to the values of the end of observation windows. For each of these panels, the first plot depicts the recurrences of the four competing risks, the second plot shows the transitions of the longitudinal marker process over four states, and the third plot shows the transitions of the health or performance status process over four states, with state 4 being absorbing, so when the process enters this state, observation of the unit terminates.}
\label{fig: observables}
\end{figure}

\begin{figure}[!h]
\centering
\begin{tabular}{c}
\includegraphics[width=\textwidth,height=5in]{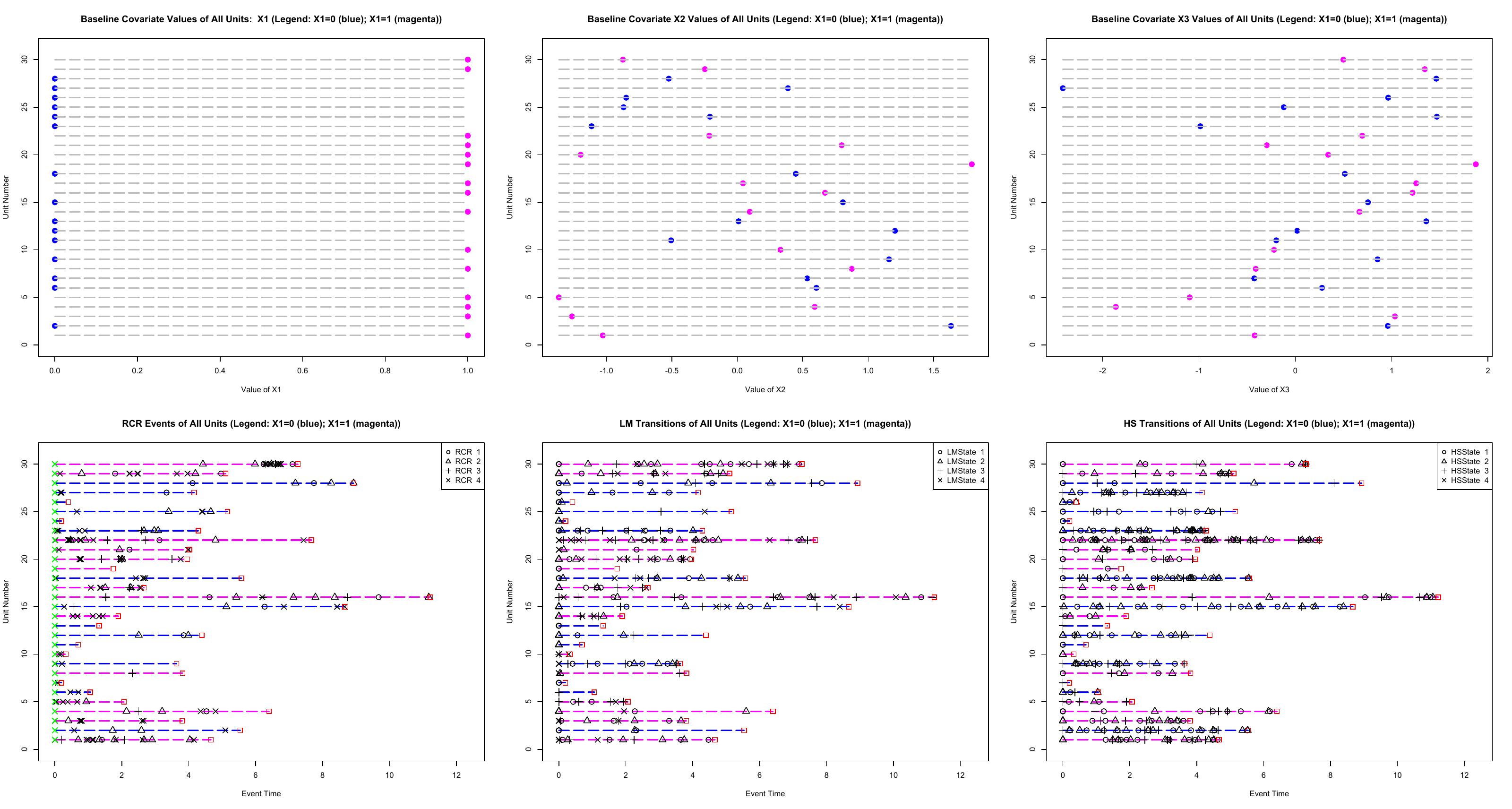}
\end{tabular}
\caption{Depiction of randomly chosen 30 units of a simulated sample data with $n = 100$ units. The top row of plots pertain to the three covariates: $X_1$ which is binary; $X_2$ which is continuous; and $X_3$ which is also continuous. The bottom row of plots pertain to the competing recurrent events (left plot); the longitudinal marker process realizations (middle plot); and the health status process realizations (right plot). For this data, $Q = 4$, $|\mathfrak{W}| = 4$, and $|\mathfrak{V}| = 4$.}
\label{sample data picture}
\end{figure}

\begin{figure}[h]
\includegraphics[width=\textwidth]{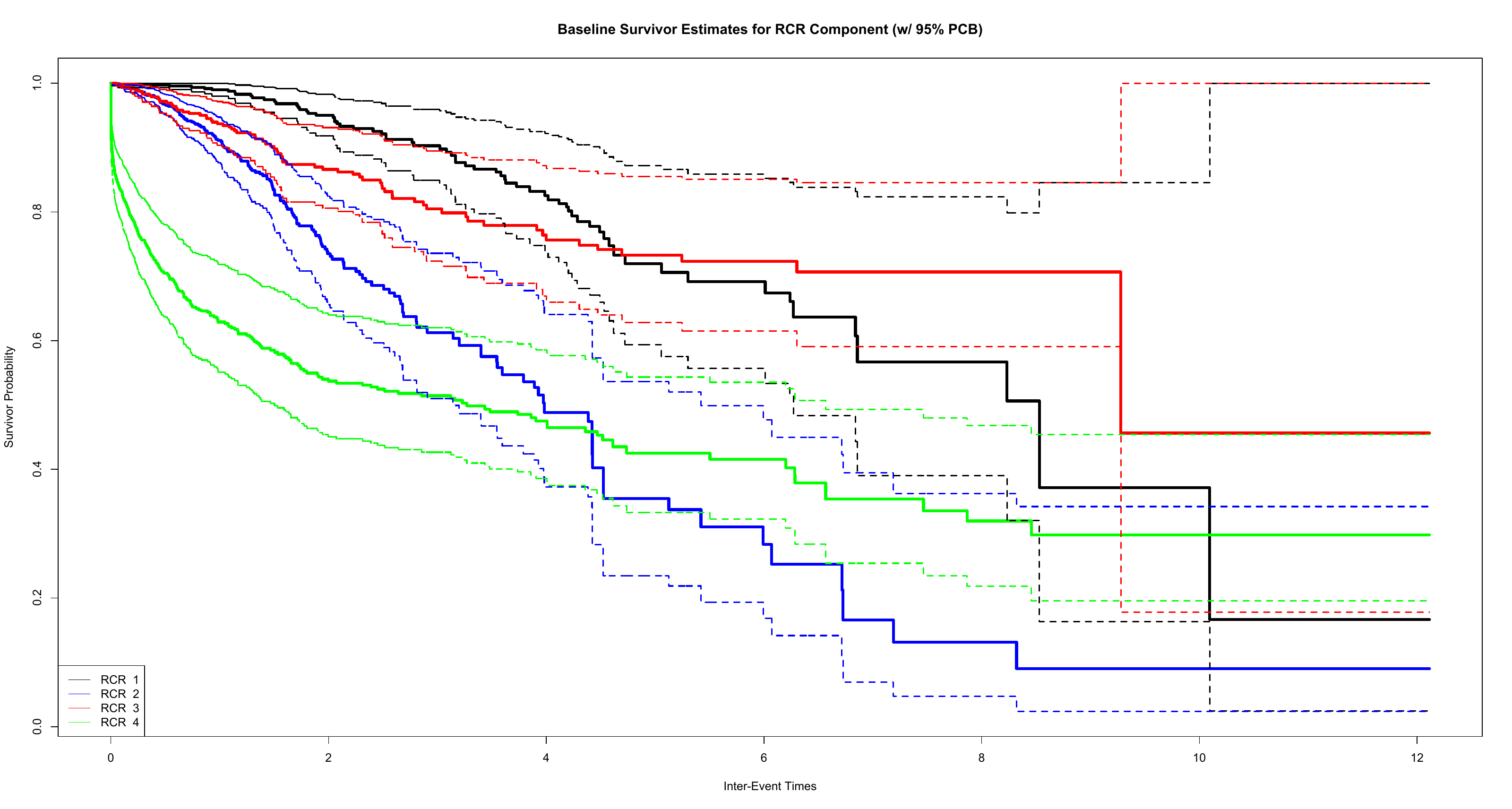}
\caption{Estimates of the baseline survivor functions of the four risks, together with their respective 95\% point-wise confidence bands.}
\label{figure-survivor estimates sample}
\end{figure}

\begin{figure}[h]
\begin{center}
\begin{tabular}{cc}
$n=50$ & $n=100$ \\
\includegraphics[width=.5\textwidth,height=1.5in]{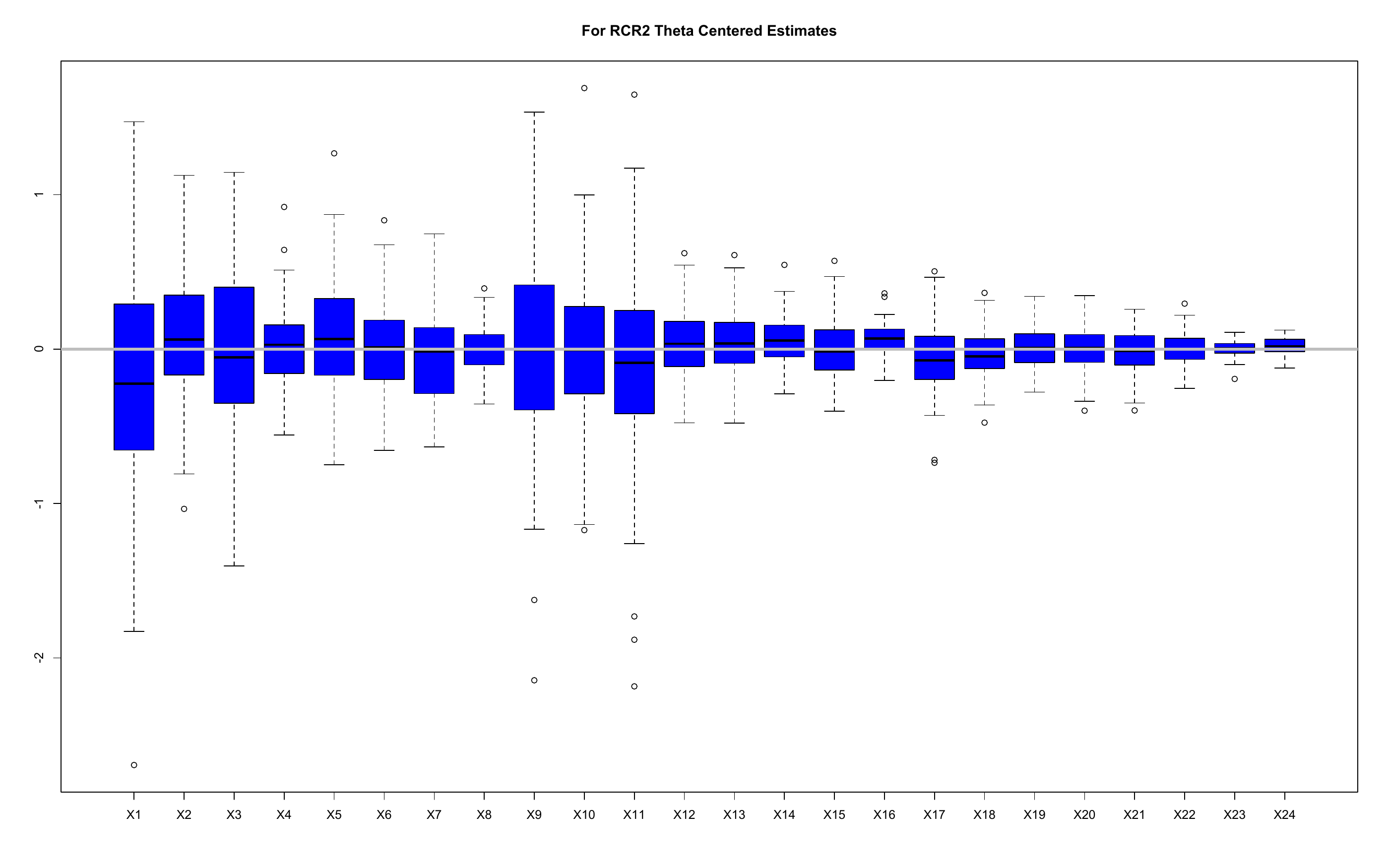} &
\includegraphics[width=.5\textwidth,height=1.5in]{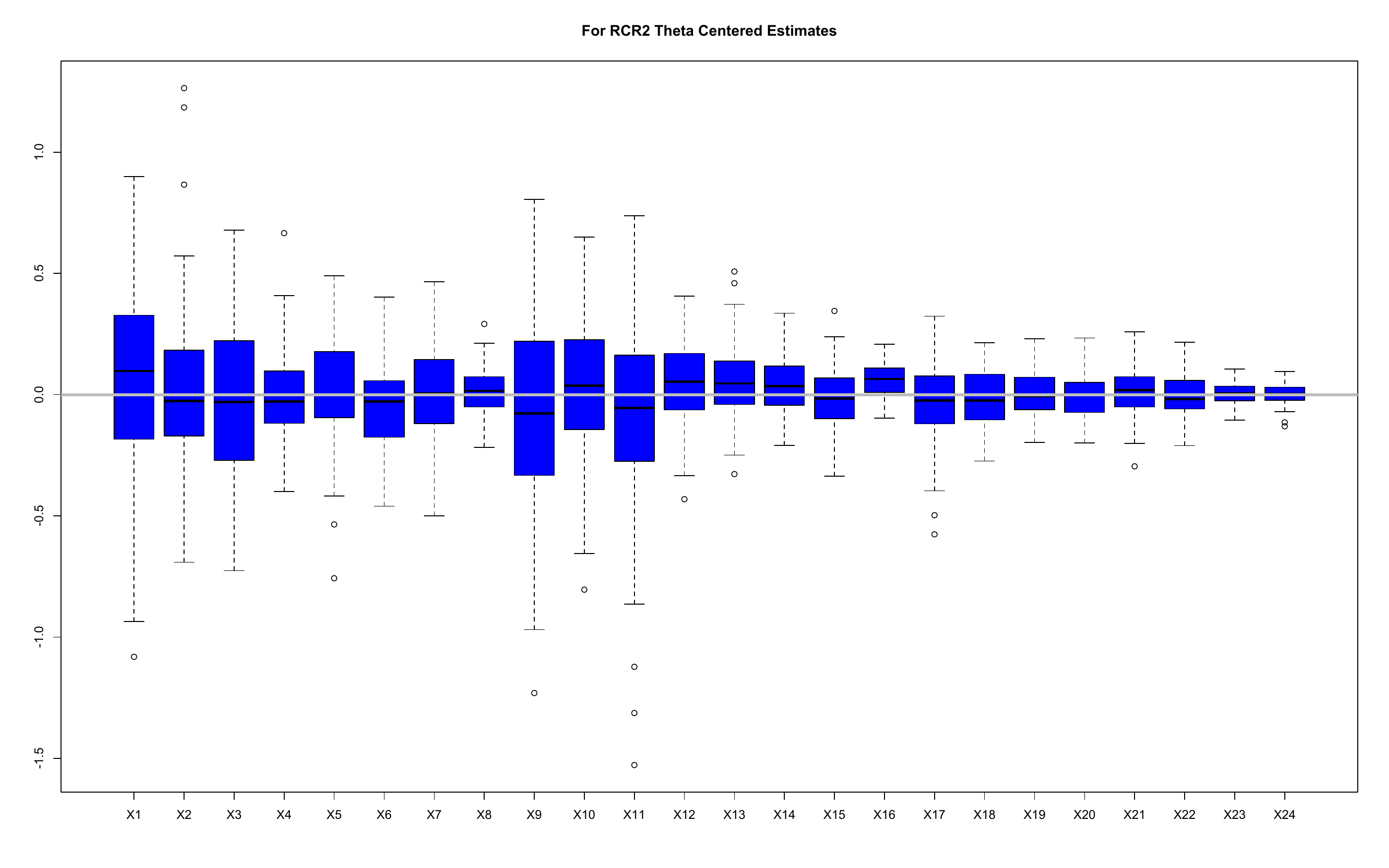} \\ %%Corrected 9/28/24
\includegraphics[width=.5\textwidth,height=1.5in]{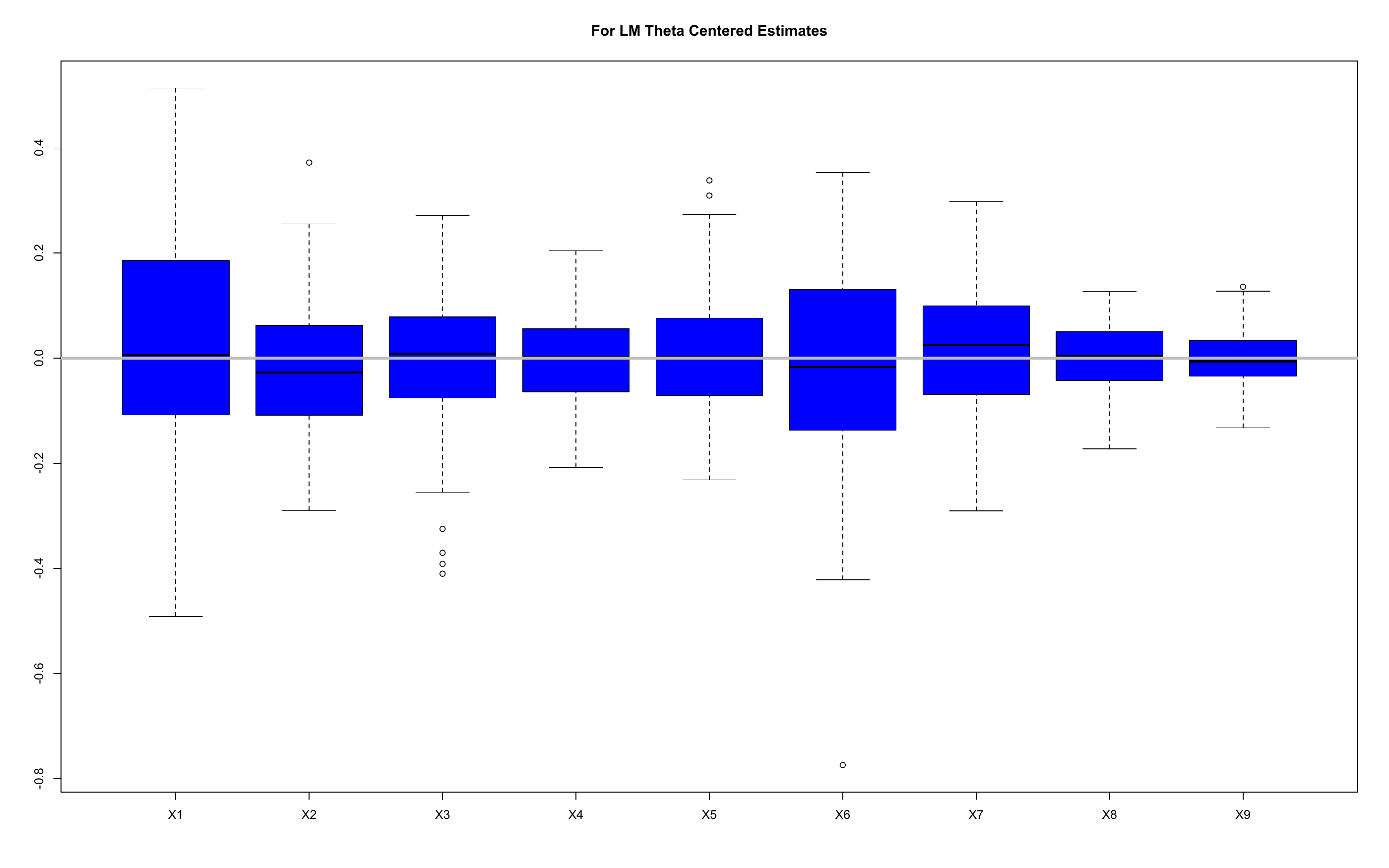} &
\includegraphics[width=.5\textwidth,height=1.5in]{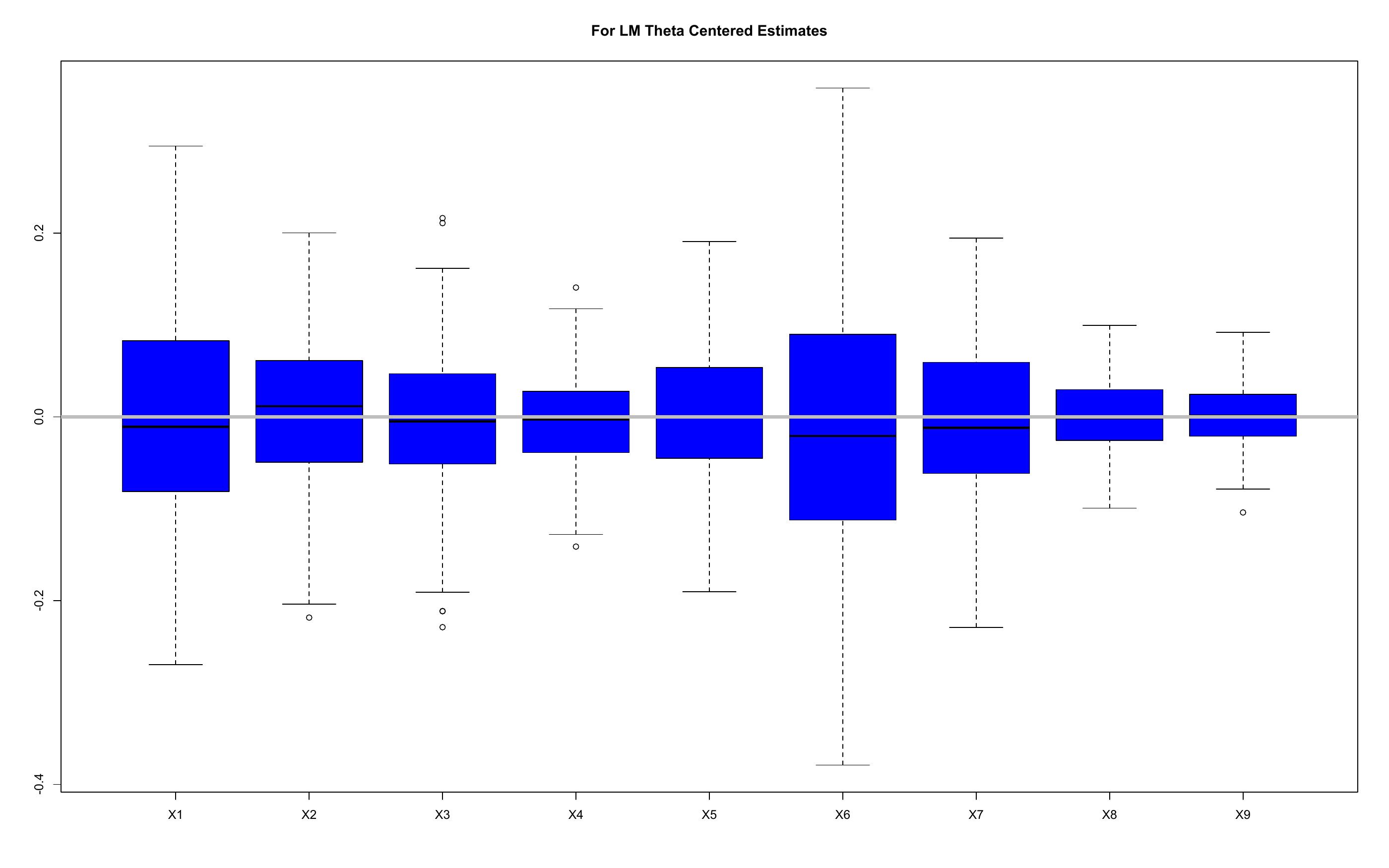} \\
\includegraphics[width=.5\textwidth,height=1.5in]{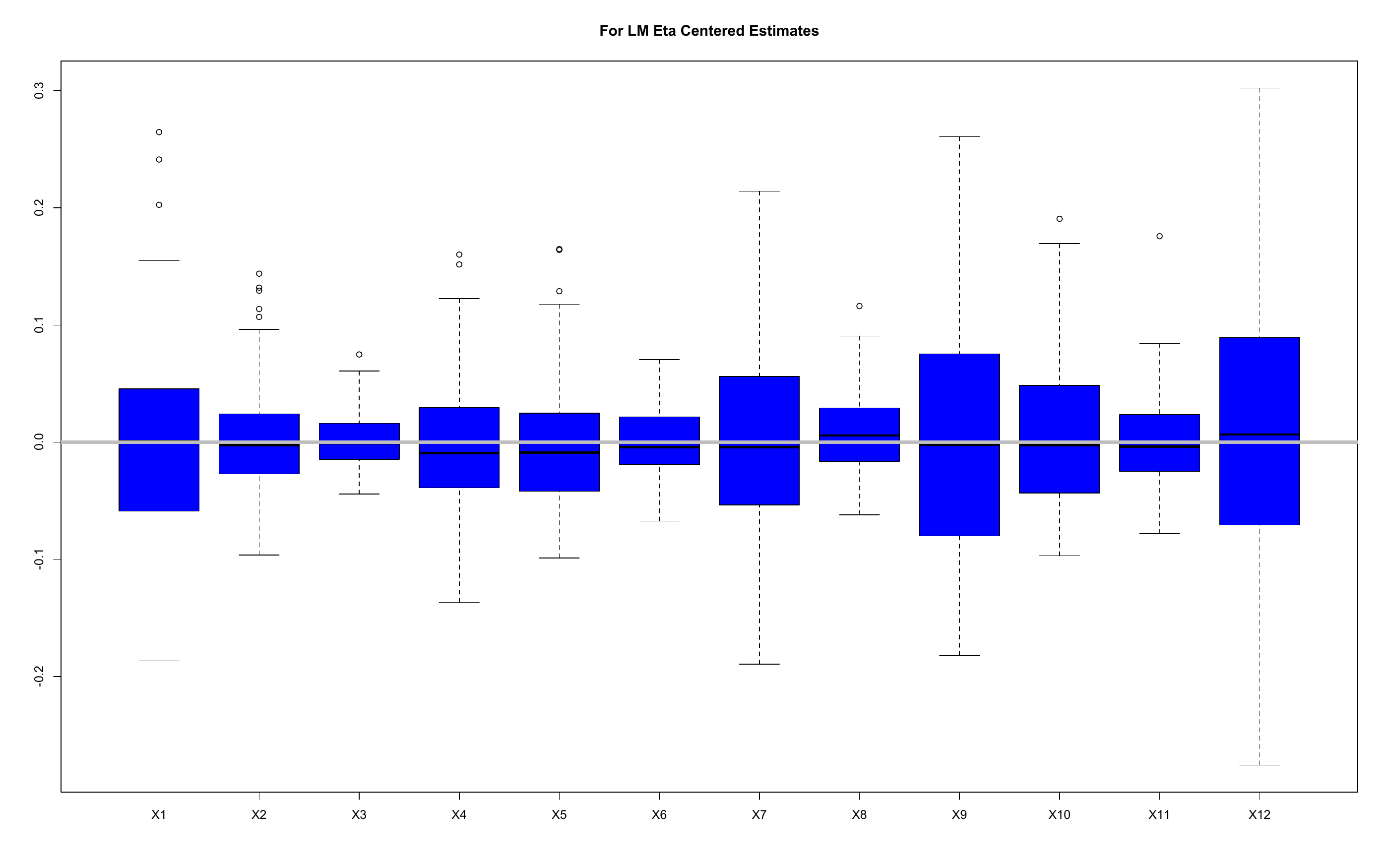} &
\includegraphics[width=.5\textwidth,height=1.5in]{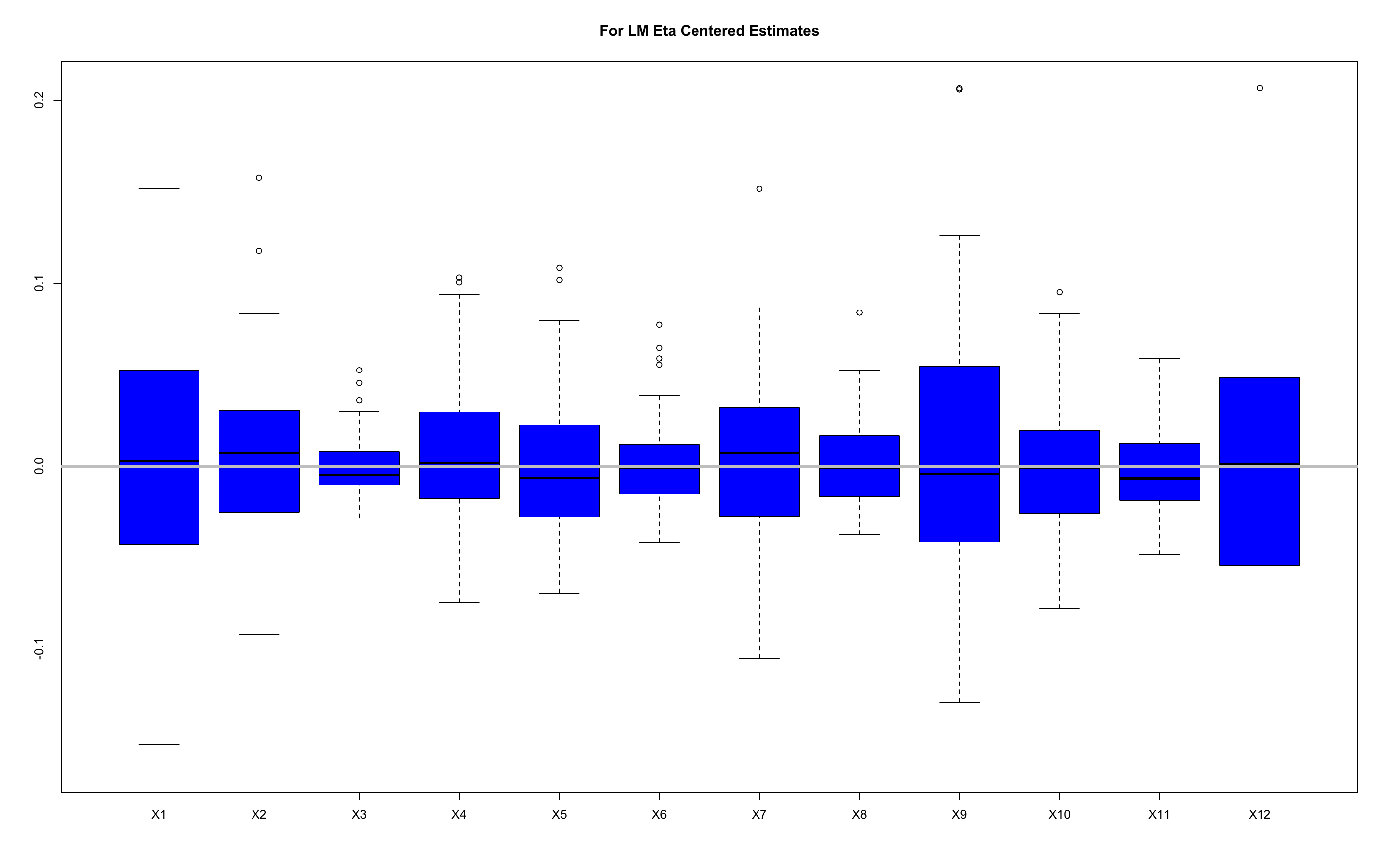} \\
\includegraphics[width=.5\textwidth,height=1.5in]{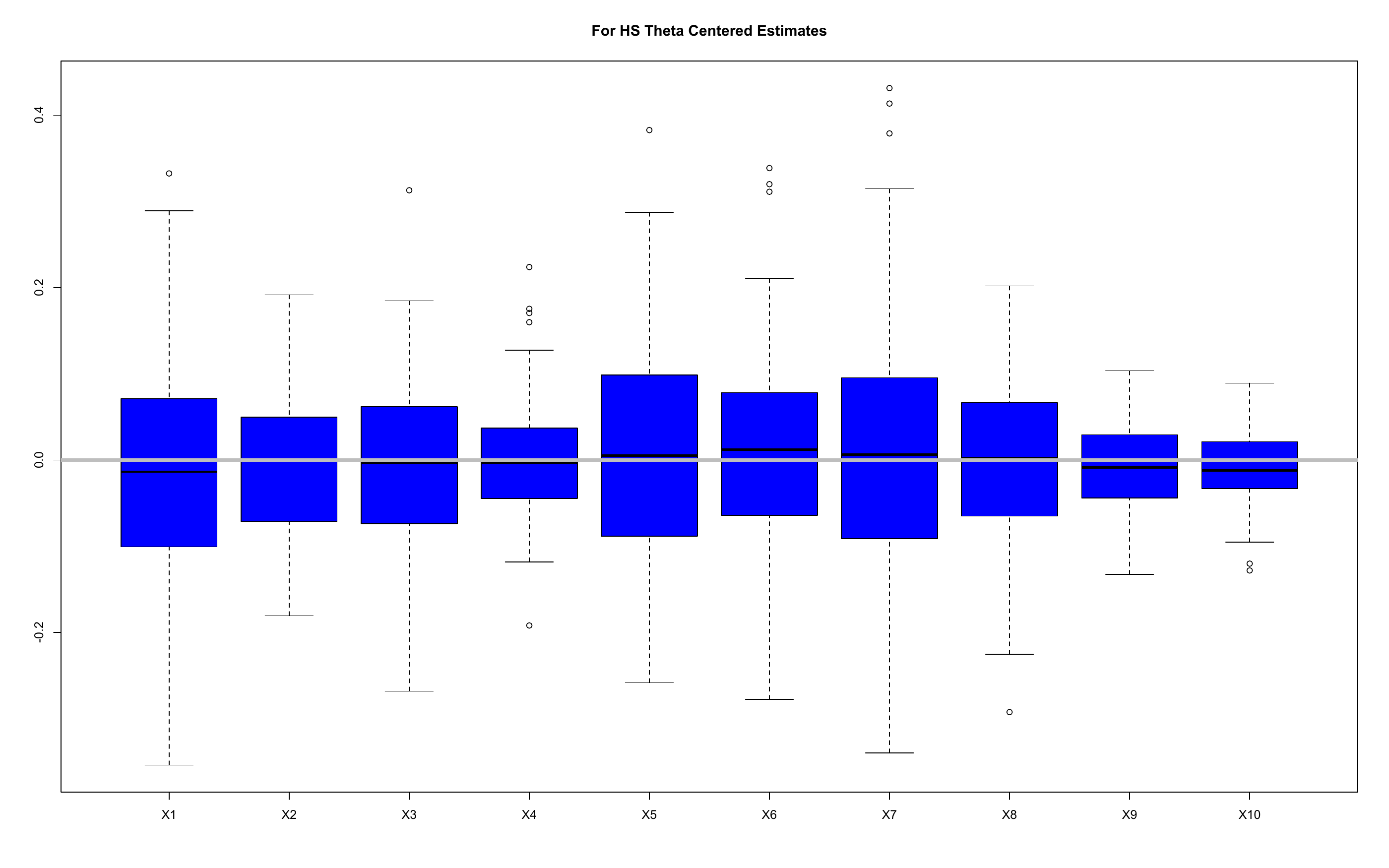} &
\includegraphics[width=.5\textwidth,height=1.5in]{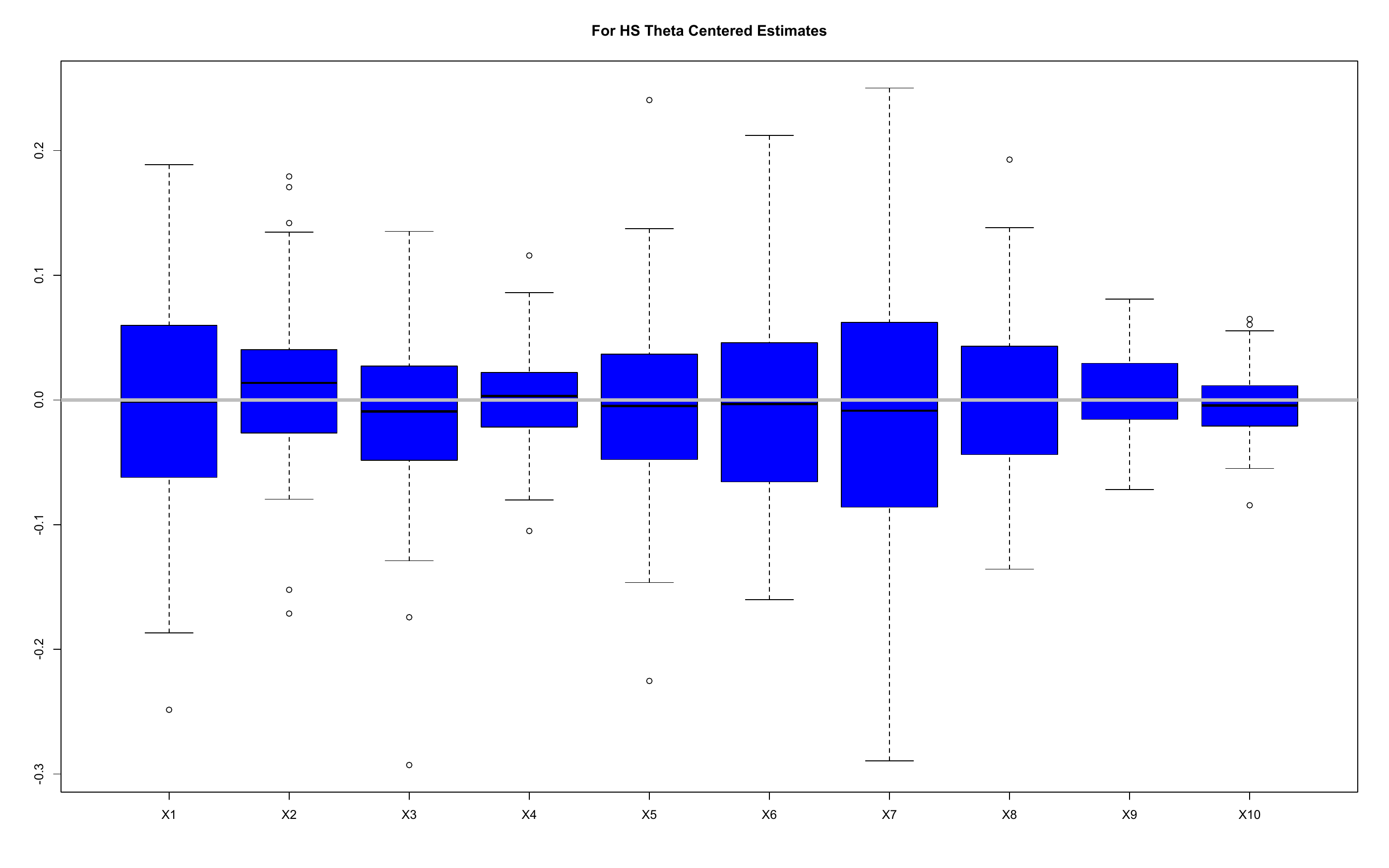} \\
\includegraphics[width=.5\textwidth,height=1.5in]{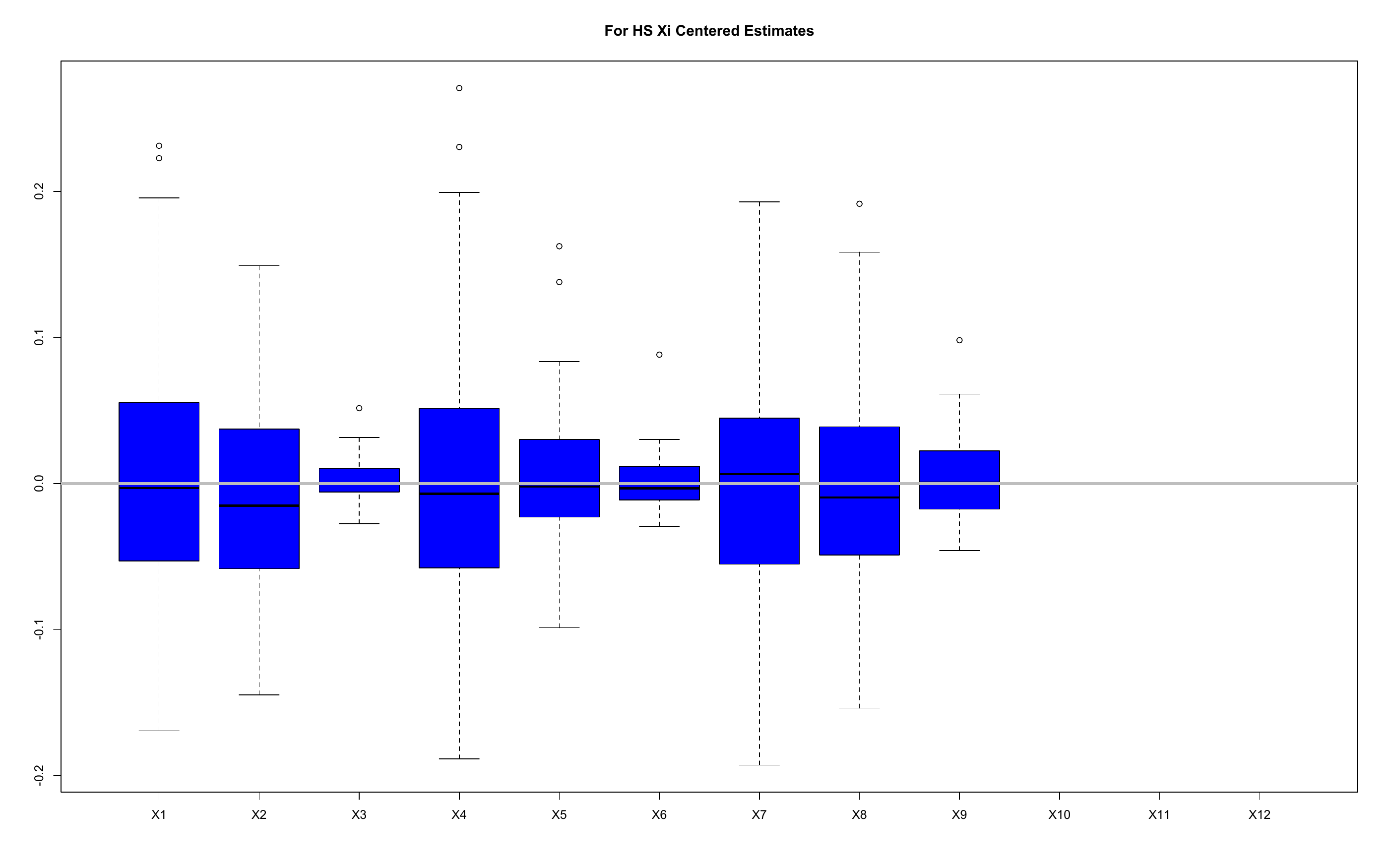} &
\includegraphics[width=.5\textwidth,height=1.5in]{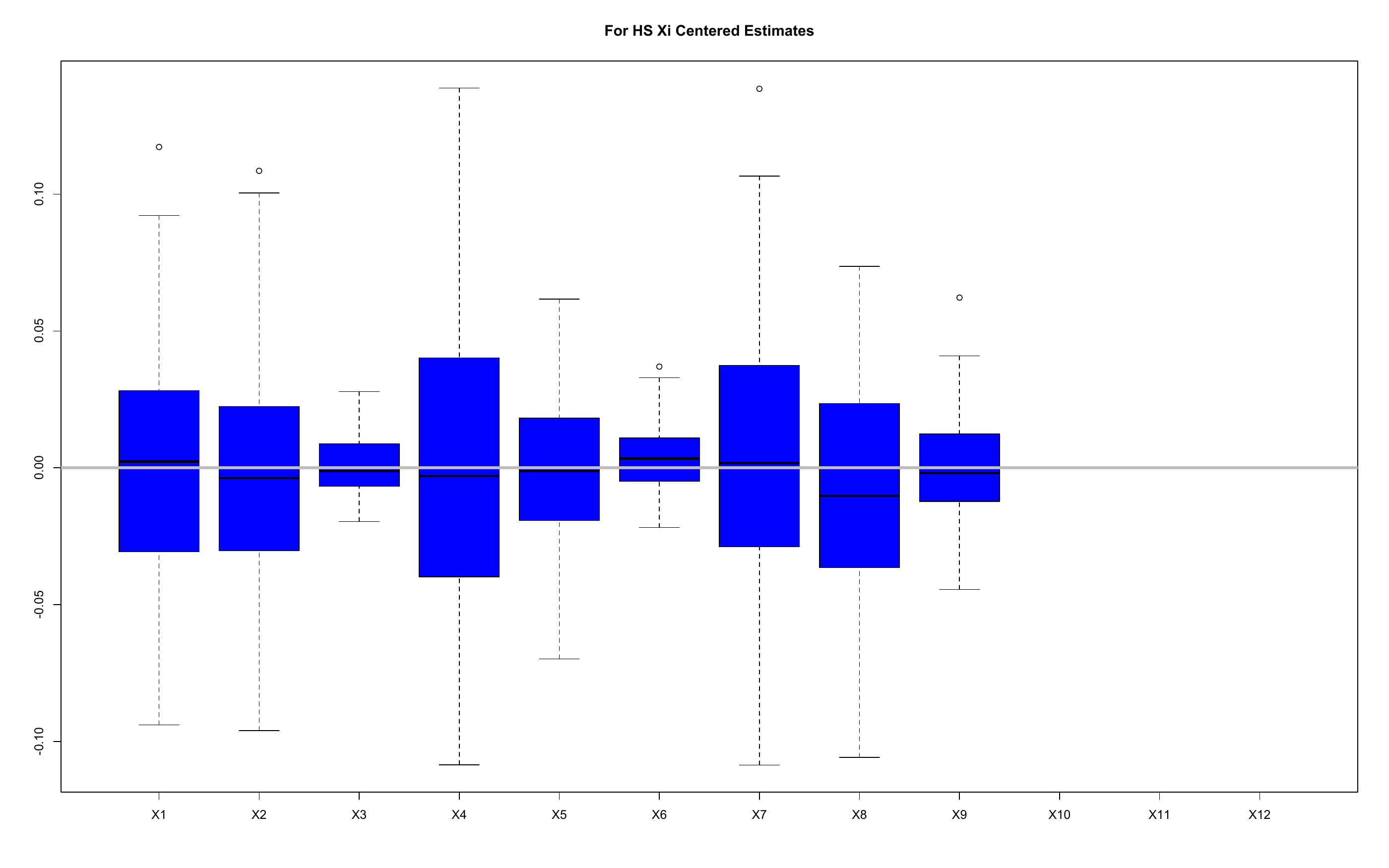}
\end{tabular}
\end{center}
\caption{Comparative boxplots of the centered simulated estimates (estimate minus {\em true} value) of the finite-dimensional model parameters for $n=50$ and $n=100$ based on 100 simulation replications. The labels in the abscissa pertain to the parameters. The first row are for the RCR $\theta$-parameters; the second row are for the LM $\theta$-parameters; the third row are for the RCR $\eta$-parameters; the fourth row are for the HS $\theta$-parameters; and the fifth row are for the HS $\xi$-parameters.}
\label{figure-comparative boxplots of estimates}
\end{figure}

\begin{figure}[h]
\begin{center}
\begin{tabular}{cc} 
$n=50$ & $n=100$ \\ 
\includegraphics[width=.5\textwidth,height=1.7in]{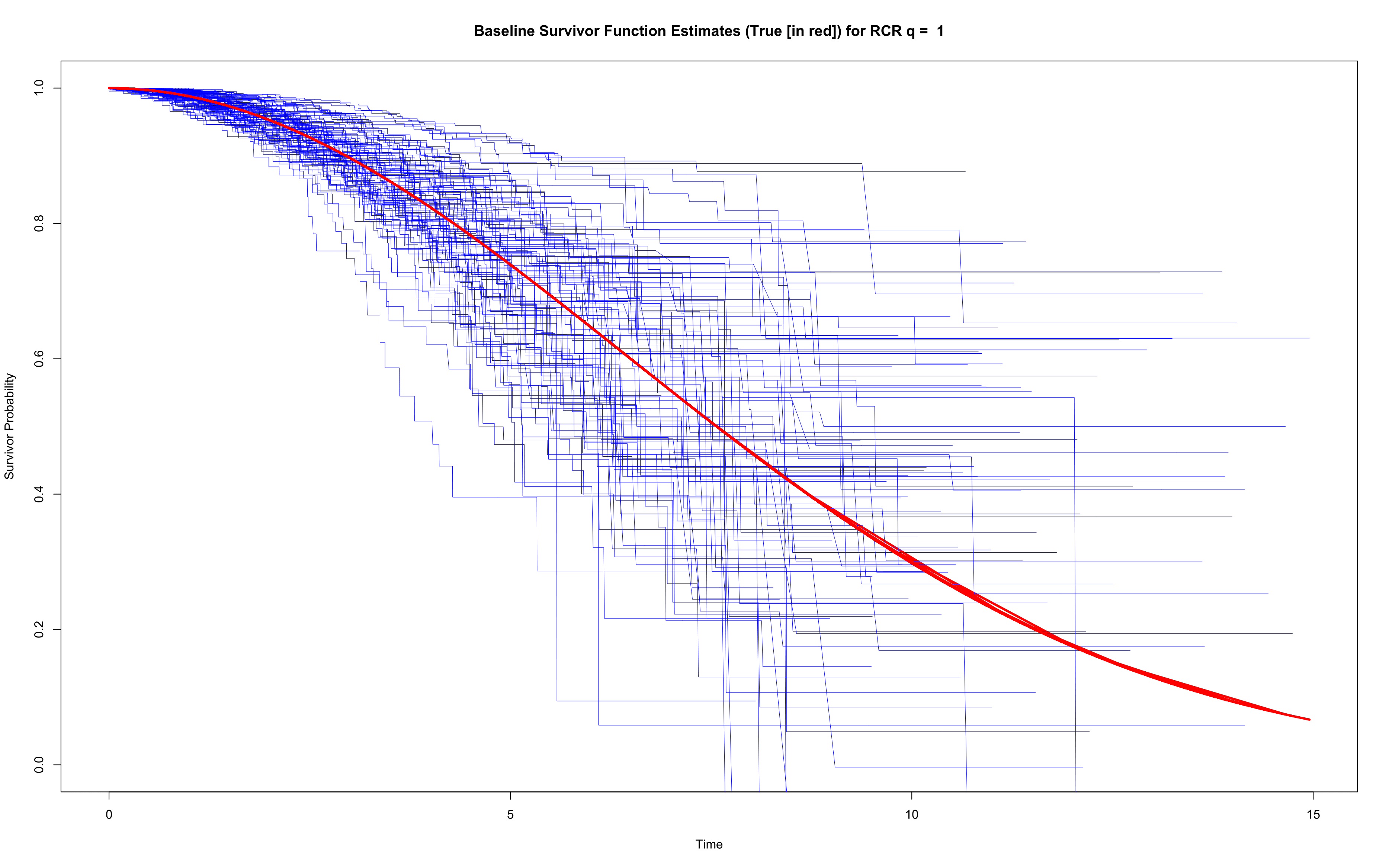} & \includegraphics[width=.5\textwidth,height=1.7in]{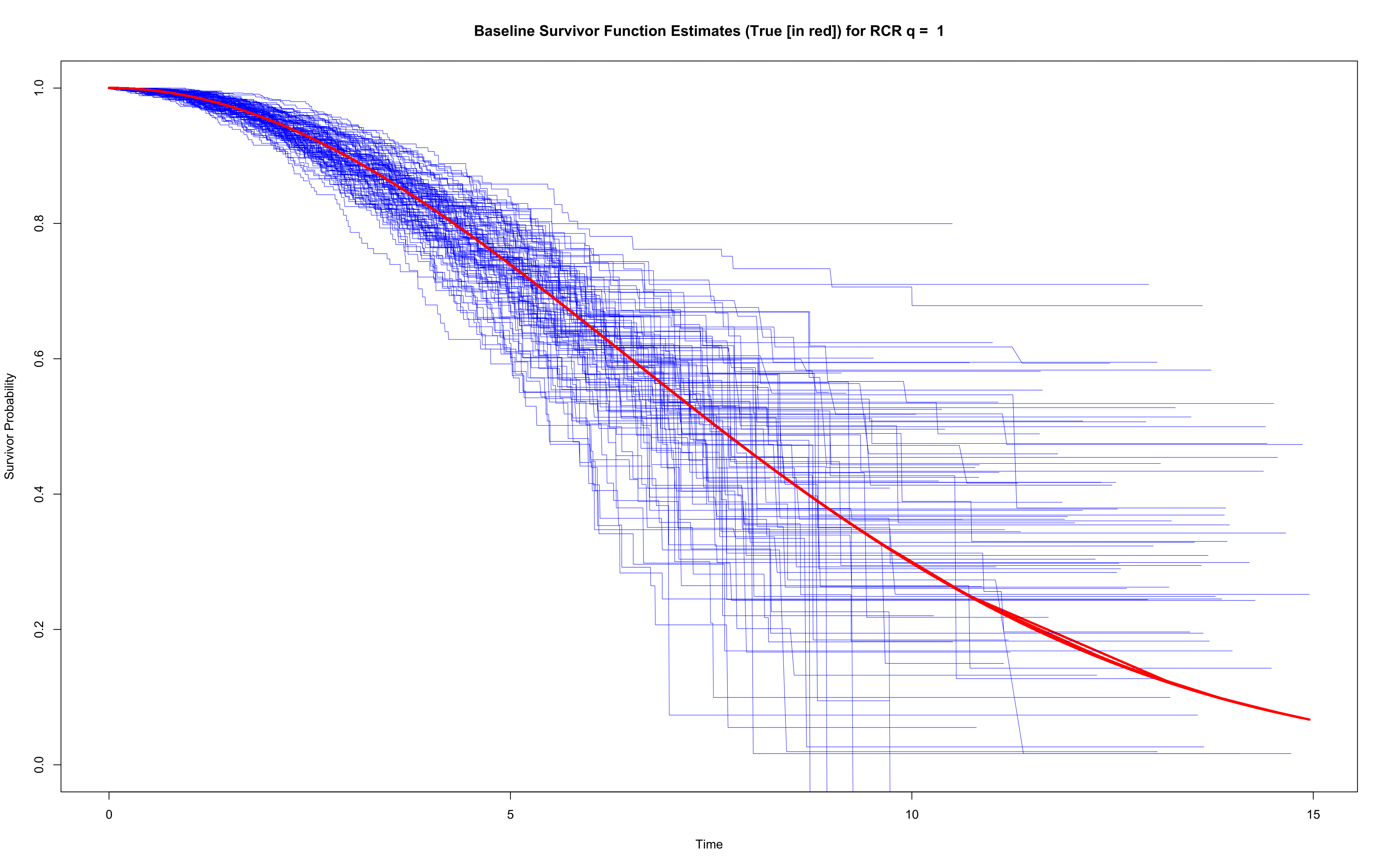} \\ 
 \includegraphics[width=.5\textwidth,height=1.7in]{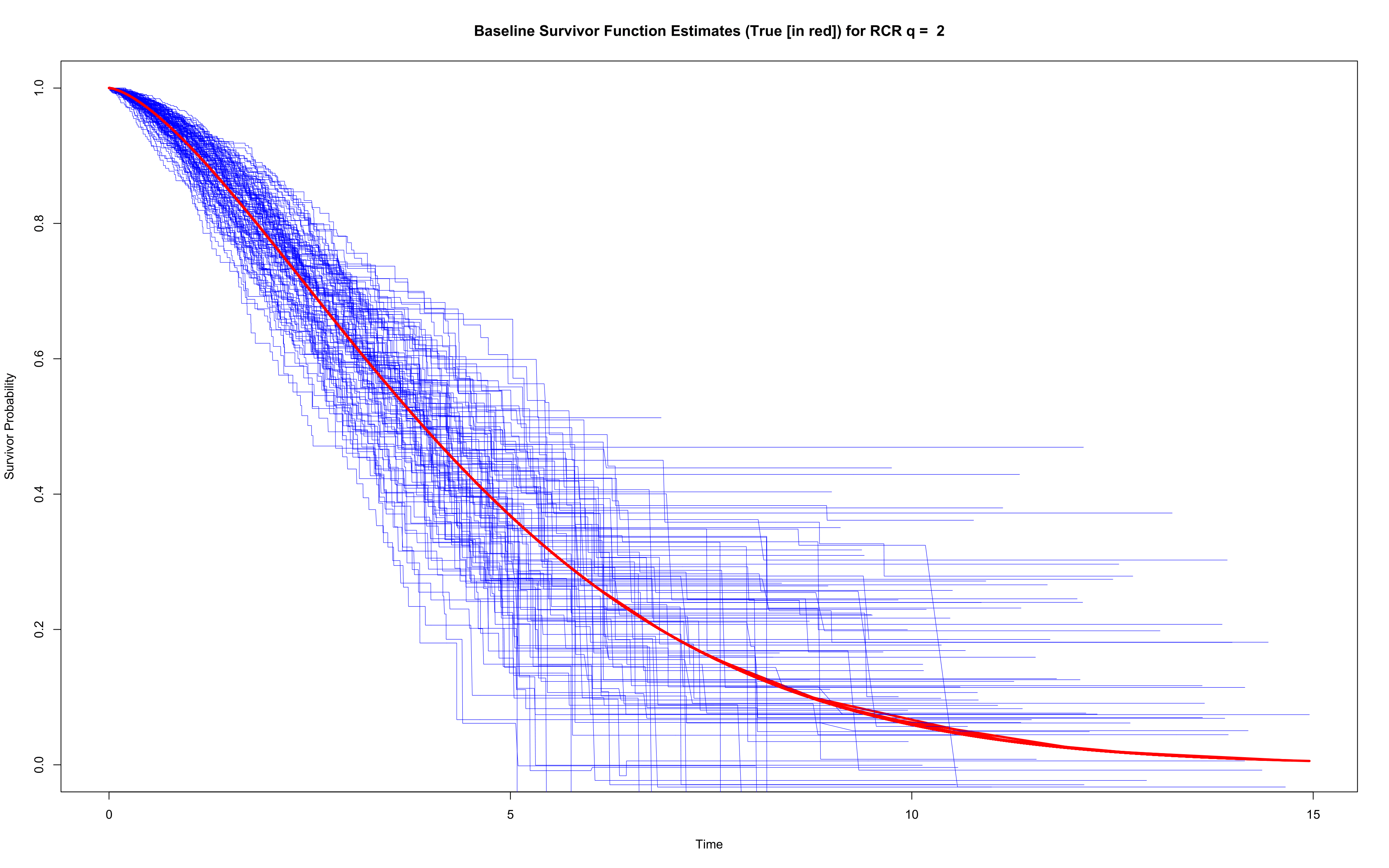} & \includegraphics[width=.5\textwidth,height=1.7in]{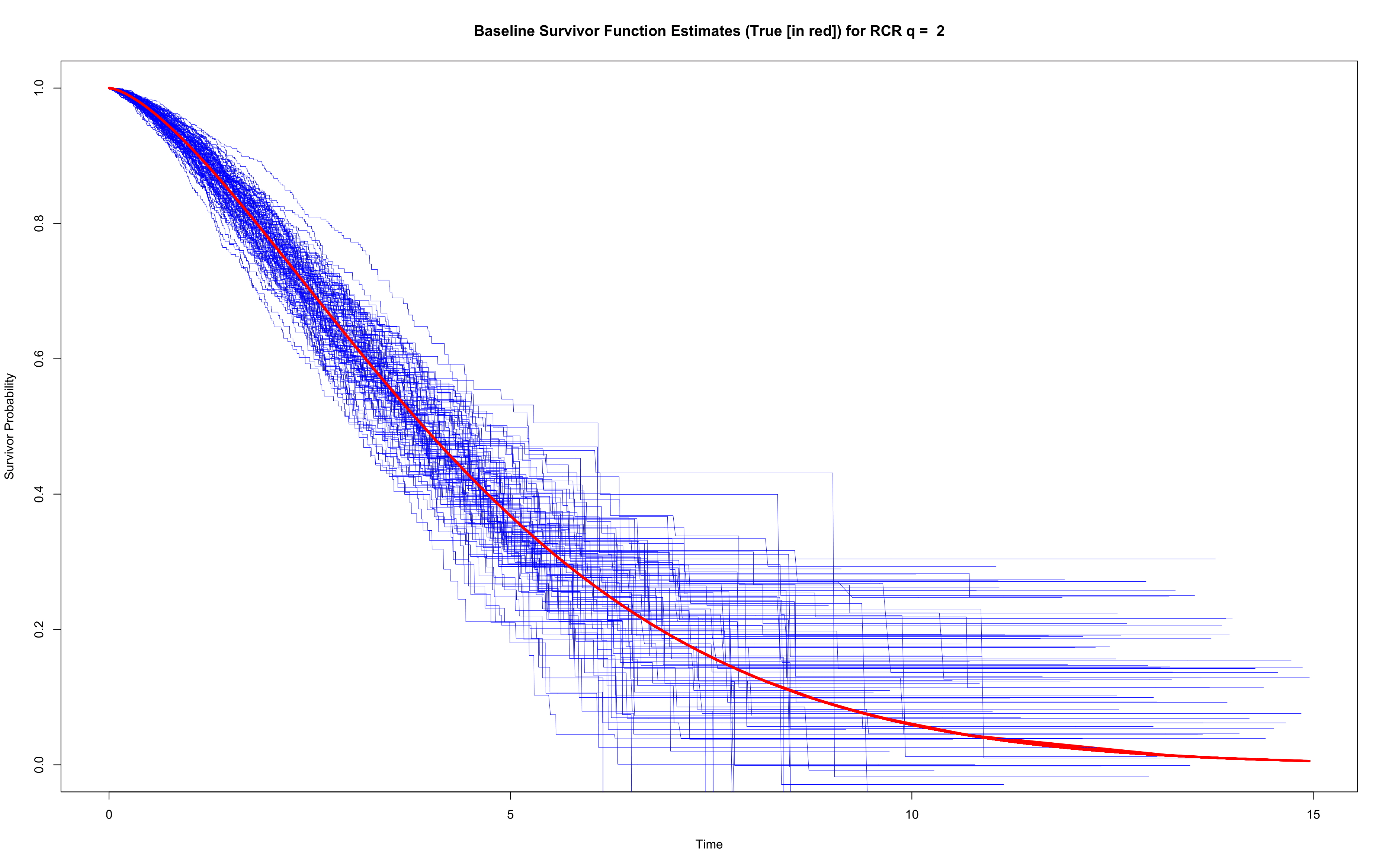} \\ 
 \includegraphics[width=.5\textwidth,height=1.7in]{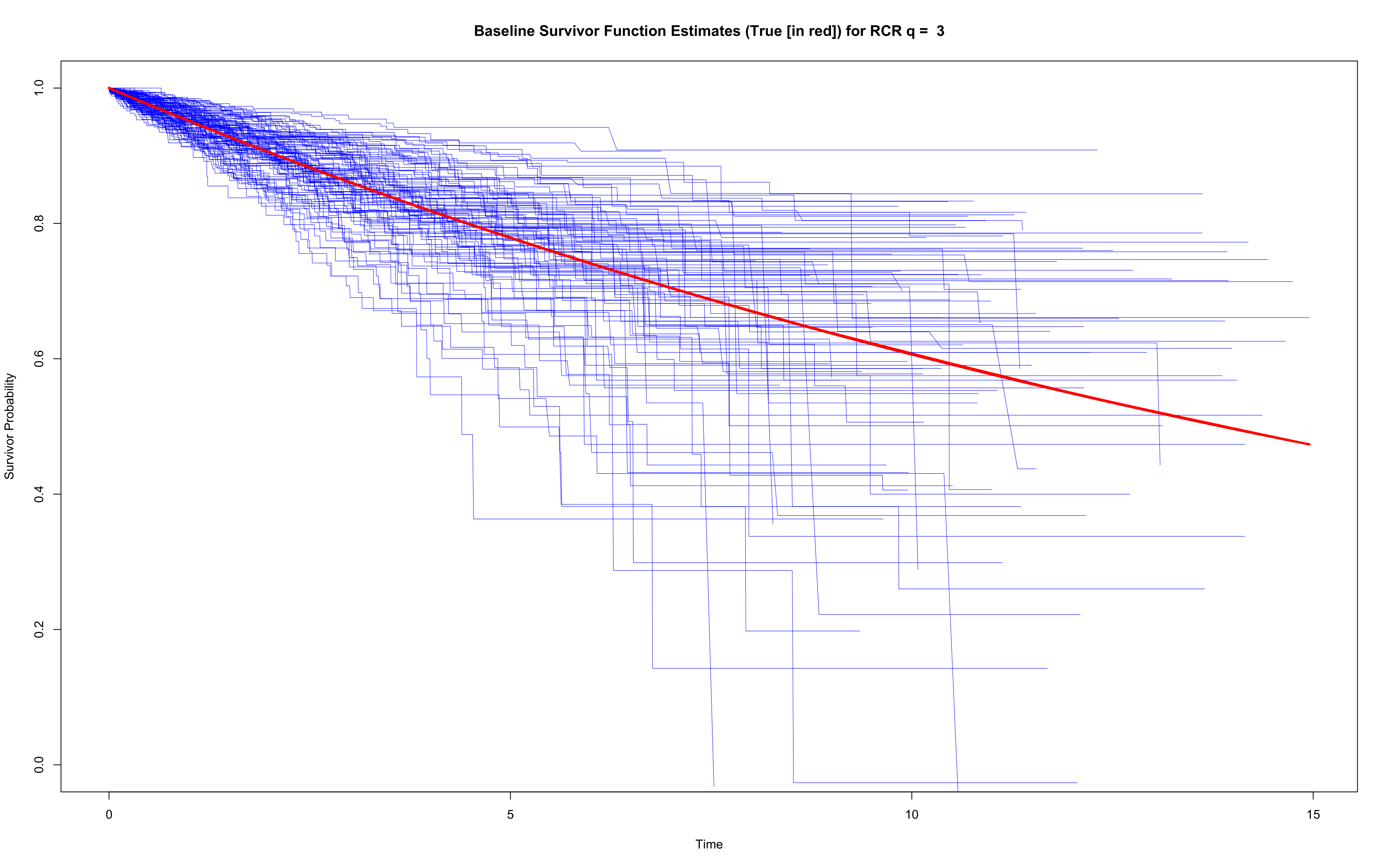} & \includegraphics[width=.5\textwidth,height=1.7in]{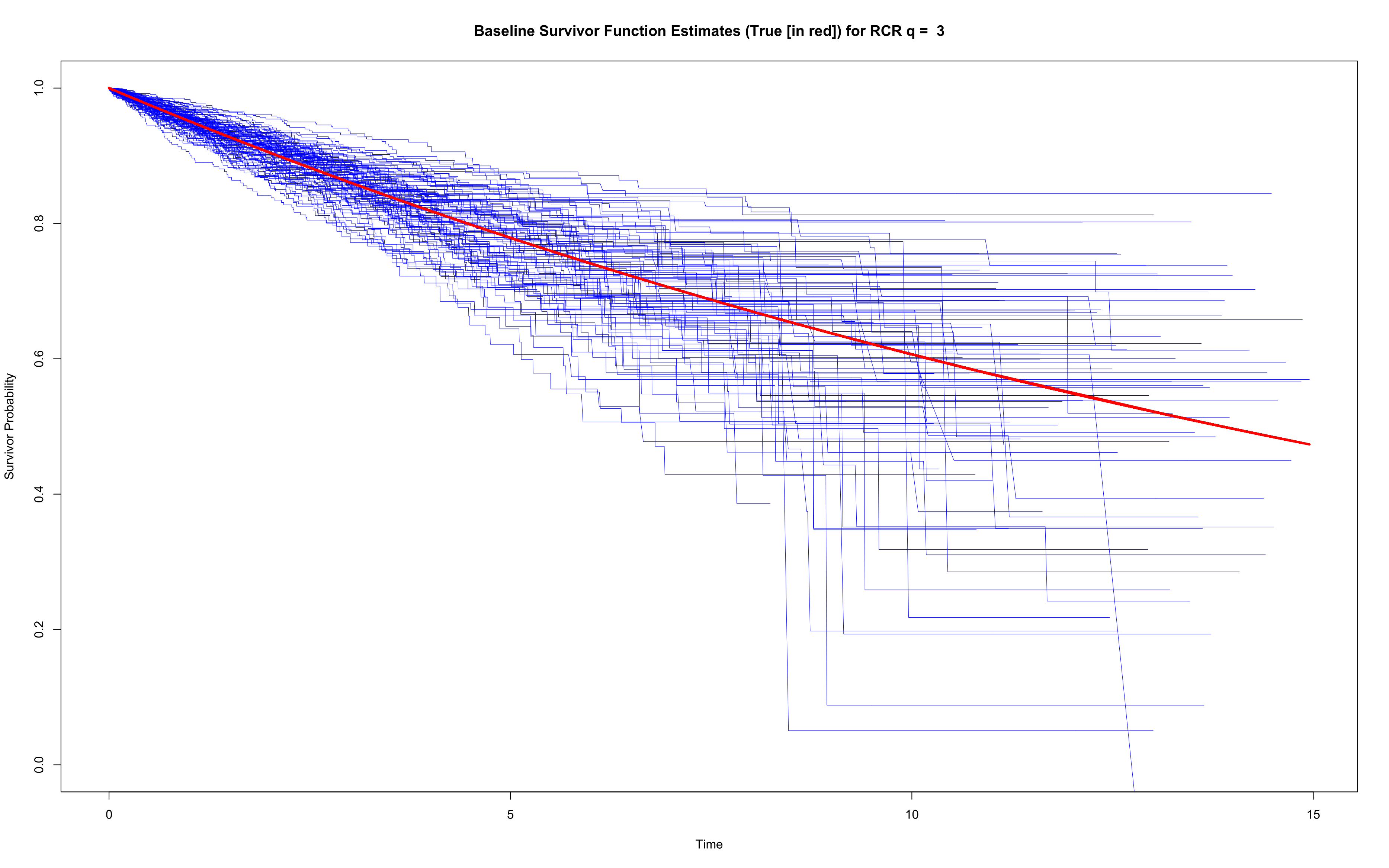} \\ 
 \includegraphics[width=.5\textwidth,height=1.7in]{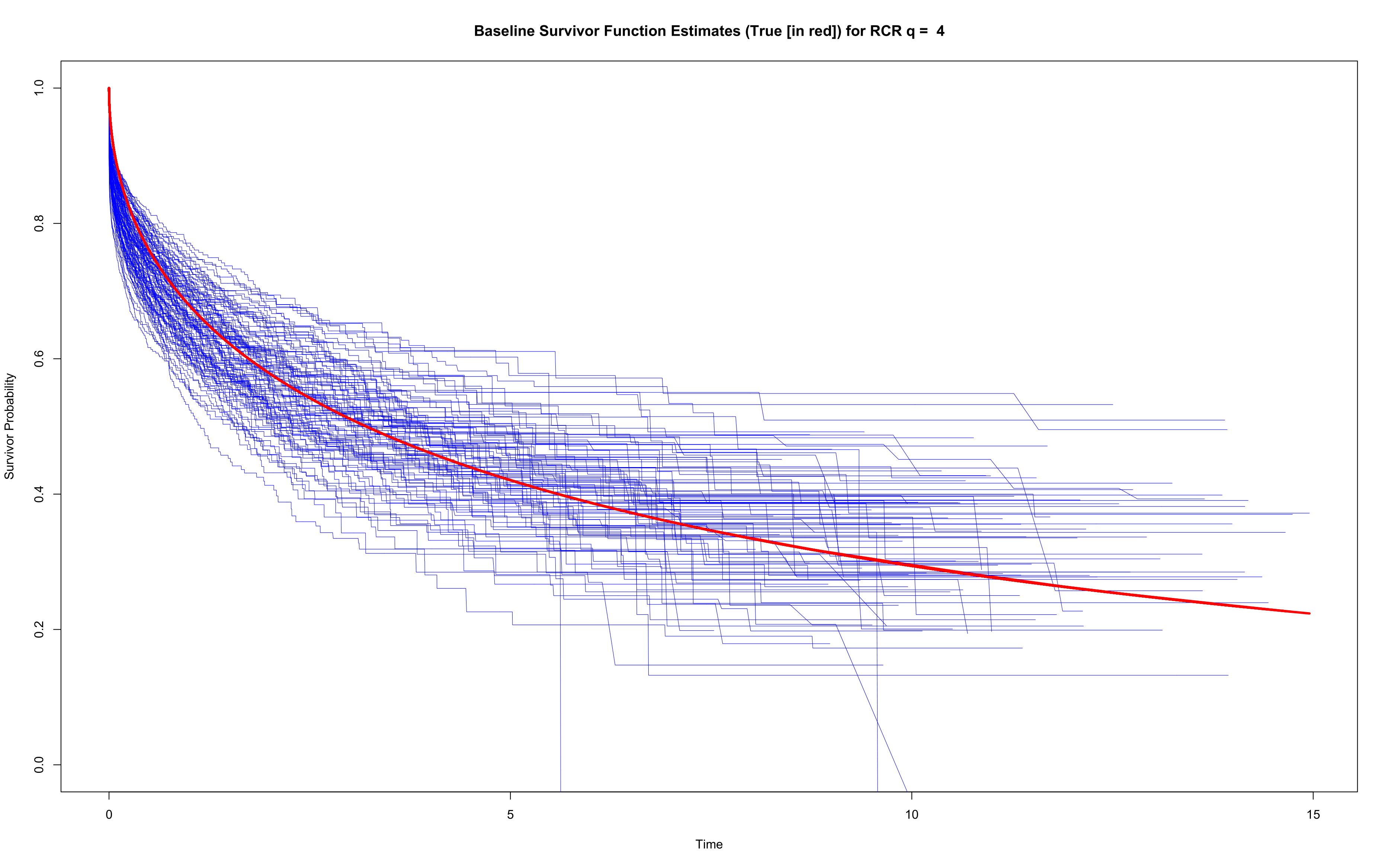} & \includegraphics[width=.5\textwidth,height=1.7in]{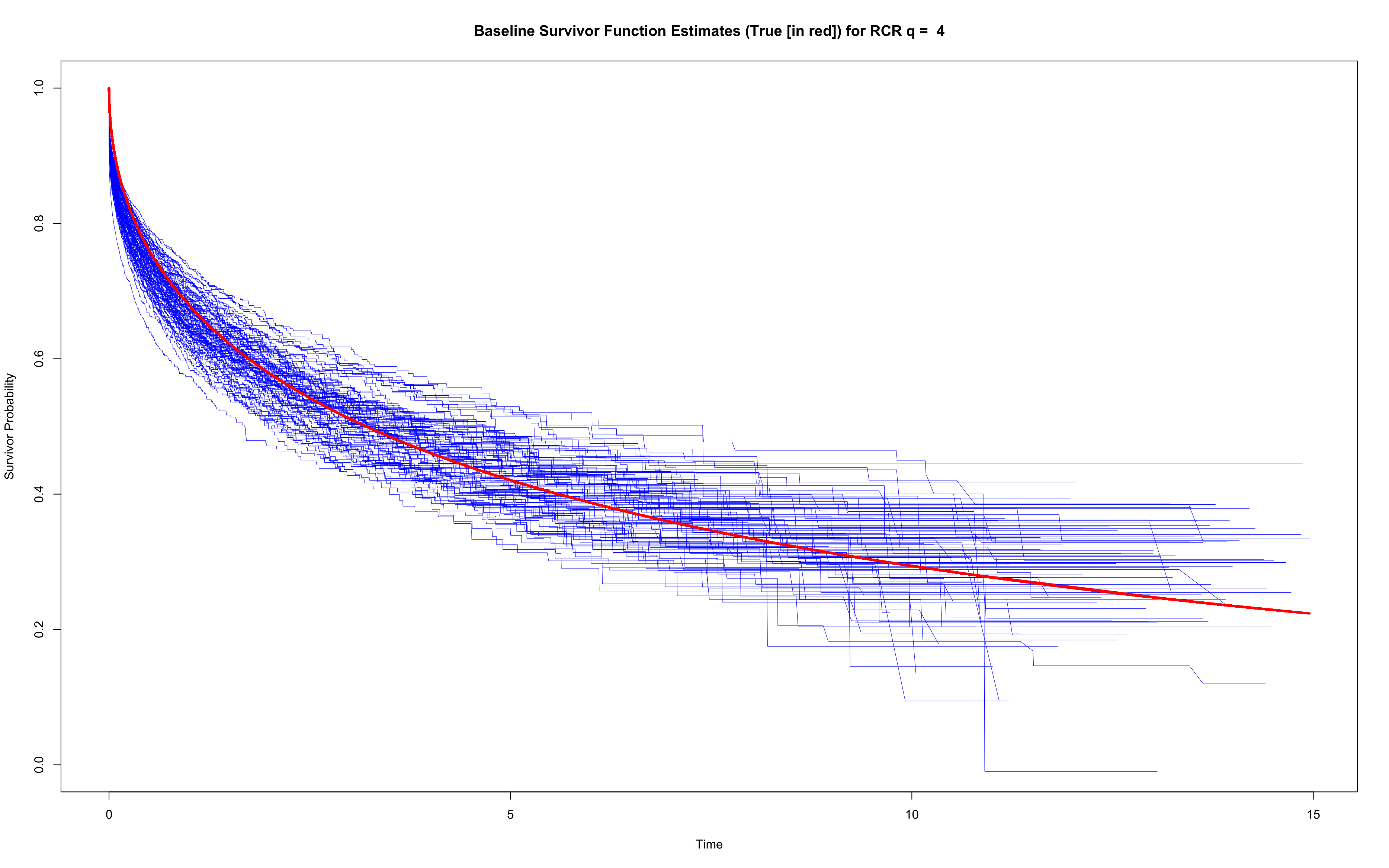}  
\end{tabular}
\end{center}
\caption{Baseline survivor function estimates from the simulation studies (100 replicates) when $n = 50$ and $n=100$ for the four competing risks ($q=1$ top; $q=4$ bottom), together with the true baseline survivor function (in red).}
\label{figure-simulated baseline estimates}
\end{figure}

\end{document}